\documentclass{article}
\usepackage{graphicx} \newcommand{\ifarver}[2]{\ifdefined\arxivversion
    #1\else
    #2\fi
}

\def\arxivversion

\usepackage[ruled]{algorithm2e}
\usepackage{algorithmic}
\usepackage{amsthm}
\usepackage{amsmath, amssymb}
\usepackage{booktabs}
\usepackage{array}
\usepackage{tabularx}
\usepackage{makecell}

\usepackage{mathtools}
\usepackage{xcolor}
\newcommand{\revise}[1]{#1}
\usepackage[round]{natbib}
\usepackage[margin=3.2cm]{geometry}
\usepackage{times}
\usepackage[pagebackref=true]{hyperref}
\usepackage{cleveref}
\usepackage{pifont}

\usepackage{tikz}
\usetikzlibrary{decorations.pathmorphing}

\usepackage{pgfplots}
\pgfplotsset{compat=1.17}
\usepackage{subcaption}

\newcounter{myres}
\newtheorem{definition}[myres]{Definition}
\newtheorem{lemma}[myres]{Lemma}
\newtheorem{proposition}[myres]{Proposition}
\newtheorem{corollary}[myres]{Corollary}
\newtheorem{theorem}[myres]{Theorem}
\newtheorem{remark}[myres]{Remark}
\newtheorem{example}[myres]{Example}

\newcommand{\naturals}{\mathbb{N}}

\newcommand{\rejset}{\mathbf{R}}

\newcommand{\rejcol}{\boldsymbol{\mathcal{R}}}

\newcommand{\evalue}{\mathbf{e}}
\newcommand{\pvalue}{\mathbf{p}}
\newcommand{\ecol}{\mathbf{E}}

\newcommand{\FDP}{\textnormal{FDP}}
\newcommand{\FDR}{\textnormal{FDR}}

\newcommand{\FWER}{\textnormal{FWER}}

\newcommand{\Kn}{\textnormal{Kn}}
\newcommand{\cKn}{\overline{\textnormal{Kn}}}

\newcommand{\BH}{\textnormal{BH}}
\newcommand{\MABH}{\textnormal{MABH}}
\newcommand{\cMABH}{\overline{\textnormal{MABH}}}

\newcommand{\cadaBH}{\overline{\textnormal{adaBH}}}
\newcommand{\adaBH}{\textnormal{adaBH}}

\newcommand{\eBH}{\textnormal{eBH}}
\newcommand{\CeBH}{{\overline{\textnormal{eBH}}}}
\newcommand{\textCeBH}{{\overline{\text{eBH}}}}
\newcommand{\BY}{\textnormal{BY}}
\newcommand{\cBY}{{\overline{\textnormal{BY}}}}
\newcommand{\textcBY}{{\overline{\text{BY}}}}
\newcommand{\Su}{\textnormal{Su}}
\newcommand{\cSu}{{\overline{\textnormal{Su}}}}
\newcommand{\textcSu}{{\overline{\text{Su}}}}
\newcommand{\cBH}{{\overline{\textnormal{BH}}}}
\newcommand{\cKnockoffs}{{\overline{\textnormal{Knockoffs}}}}

\newcommand{\SC}{\textnormal{SC}}

\newcommand{\ER}{\textnormal{ER}}

\newcommand{\ind}{\mathbf{1}}
\newcommand{\F}{\textnormal{f}}
\newcommand{\G}{\textnormal{g}}
\newcommand{\expect}{\mathbb{E}}
\newcommand{\prob}{\mathbb{P}}

\newcommand{\Fcal}{\mathcal{F}}

\newcommand{\Rcal}{\mathcal{R}}

\newcommand{\pdist}{\mathrm{P}}

\usepackage{autonum}

\title{
Bringing Closure to False Discovery Rate Control: A General Principle for Multiple Testing
}

\author{
Ziyu Xu
\thanks{Department of Statistics and Data Science, Carnegie Mellon University, USA. \texttt{xzy@cmu.edu}.}
\and
Aldo Solari
\thanks{Department of Economics, Ca' Foscari University of Venice, Italy. \texttt{aldo.solari@unive.it}.}\
\and
Lasse Fischer
\thanks{Competence Center for Clinical Trials Bremen, University of Bremen, Germany. \texttt{fischer1@uni-bremen.de}.}
\and
Rianne de Heide
\thanks{Department of Applied Mathematics, University of Twente, and Centrum Wiskunde \& Informatica, Amsterdam, The Netherlands \texttt{r.deheide@utwente.nl}.}
\and
Aaditya Ramdas
\thanks{Departments of Statistics and Data Science, and Machine Learning, Carnegie Mellon University, USA. \texttt{aramdas@cmu.edu}.}
\and
Jelle Goeman
\thanks{Department of Biomedical Data Sciences, Leiden University Medical Center, the Netherlands. \texttt{j.j.goeman@lumc.nl}.}
}
\date{\today}

\begin{document}
\footnotetext{This paper merges and subsumes the two parallel works of \citet{xu_bringing_closure_2025} and \citet{goeman_e-partitioning_principle_2025}.\newline}

\maketitle
\begin{abstract}
We present a novel principle that is necessary and sufficient for multiple testing methods controlling an expected loss. This principle asserts that every such method is a special case of a general closed testing procedure based on e-values.
It generalizes the Closure Principle, known to underlie all methods controlling familywise error and tail probabilities of false discovery proportions, to a large class of error rates --- in particular to the false discovery rate (FDR).
By writing existing methods as special cases of this procedure, we can achieve uniform improvements, as we demonstrate for the e-Benjamini-Hochberg and Benjamini-Yekutieli procedures and the self-consistent method of Su (2018).
We also show that methods derived using our novel e-Closure Principle generally control their error rate not just for one rejection set but simultaneously over many, allowing researchers post hoc flexibility.

\revise{Moreover, because all multiple testing methods for the expected loss error metrics covered by our framework are derived from the same procedure, researchers may even choose the error metric post hoc.}
Under certain conditions, this flexibility even extends to post hoc choice of the nominal error rate.
 \end{abstract}
 
\tableofcontents

\section{Introduction}

For family-wise error rate (FWER) control, it has long been known \citep{sonnemann1982allgemeine, sonnemann_general_solutions_2008} that there is a single universal principle that is necessary and sufficient for the construction of valid methods. This Closure Principle says that any method that controls FWER is a special case of a closed testing procedure \citep{marcus1976closed}. The Closure Principle was later challenged by the seemingly more powerful Partitioning Principle \citep{finner2002partitioning}, but \citet{goeman2021only} have shown that the two principles are equivalent. \citet{genovese2006exceedance} and \citet{goeman2011multiple} have extended closed testing to control of false discovery proportions (FDPs), and \citet{goeman2021only} showed that all methods controlling a quantile of the distribution of FDP are either equivalent to a closed testing procedure or are dominated by one, extending the Closure Principle to all methods controlling FDP.

The Closure Principle is useful for method development in FWER and FDP in several ways. In the first place, it reduces the complex task of constructing a multiple testing method to the simpler task of choosing hypothesis tests for intersection hypotheses. After making this choice, method construction reduces to a discrete optimization problem. Moreover, the Closure Principle helps to handle complex situations such as restricted combinations, i.e., logical implications between hypotheses \citep{shaffer1986modified, goeman2021only}. Finally, methods constructed using closed testing often allow for some user flexibility, permitting researchers to modify some aspects of the multiple testing procedure post hoc without compromising error control \citep{goeman2011multiple}.

However, the known universality of the Closure Principle only covers methods controlling a tail probability of a function of the number or proportion of errors made \citep{goeman2021only}. Methods controlling expectations of such functions, such as per-family error rate (PFER) and, most importantly, false discovery rate \citep[FDR;][]{benjamini_controlling_false_1995} are outside the scope of the Closure Principle. We fill this gap by extending the Closure Principle to all expectation-based error rates, including FDR control.

For the construction of FDR control methods, \citet{blanchard_two_simple_2008} have formulated two quite general sufficient conditions, self-consistency and dependence control, under which, if both hold, FDR control is guaranteed. The self-consistency condition simplifies the proof of well-known FDR controlling procedures such as BH \citep{benjamini_controlling_false_1995} and BY \citep{benjamini_control_false_2001} and has been seminal for the creation of many others.
An important recent example is the eBH procedure \citep{wang_false_discovery_2022}, which controls FDR on the basis of per-hypothesis e-values rather than p-values.
The concept of the e-value, a random variable with expectation at most 1 under the null hypothesis, works well with FDR since both concepts are expectation-based. Several other authors \citep{ignatiadis2024values, lee_boosting_e-bh_2024a, ren2024derandomised, ignatiadis_asymptotic_compound_2025, li_note_e-values_2025} have pointed out useful connections between FDR and e-values, and we elaborate on the connections of this work with recent progress in this area in \Cref{sec:related-work}.

However, self-consistency is sufficient but not necessary for FDR control, and there is some indication that it is not optimal. When symmetric, self-consistent methods control FDR at the more stringent level $\pi_0\alpha$ rather than at $\alpha$, but \citet{solari2017minimally} have shown that any method controlling FDR at $\pi_0\alpha$ can be uniformly improved by a method that controls FDR at $\alpha$, and which itself is not self-consistent.
Some methods constructed using self-consistency, e.g., eBH and BY, have a reputation for low power \citep{lee_boosting_e-bh_2024a, xu_more_powerful_2023}.
Furthermore, methods based on self-consistency
\citep{lei2018adapt, lei2021general, katsevich2020simultaneous, katsevich2023filtering} generally lack the full post hoc user flexibility offered by closed testing \citep{goeman2011multiple}. However, \cite{fischer2024online}, using earlier results by \citet{su2018fdr}, showed how post hoc selection can be achieved with FDR control using self-consistency, although at a cost to power.

In this paper, we present a generalization of the Closure Principle, the e-Closure Principle, that covers all expectation-based error rates as well as tail-probability-based ones. In particular, we show that \revise{the e-Closure Principle} is both necessary and sufficient for FDR control. Like the Closure Principle, its generalization facilitates the task of constructing a multiple testing procedure by reducing it to the simpler task of choosing appropriate local e-values for all intersection hypotheses. Once these e-values are chosen, the multiple testing procedure reduces to a discrete optimization problem. Like the Closure Principle for FWER and FDP, the generalized e-Closure Principle offers some post hoc user flexibility for FDR control methods in a natural way. It can be used to formulate uniform improvements of existing methods, including the eBH and BY procedures and the self-consistent procedure of \citet{su2018fdr}. We show the power increase on a suite of real datasets that results from our improvement of the BY procedure, which we call the $\cBY$ (closed BY) procedure, in \Cref{tab:by_real_data_discoveries}.
\ifarver{\footnote{Code accompanying this paper can be found at \url{https://github.com/neilzxu/eclosure}.}}{}
Our improvements are strict in the sense that they will always make as many discoveries as the original procedure, and often more, as is the case in the aforementioned table. We also provide a comprehensive table with results on other methods (such as BH) in \Cref{sec: real data}.
Finally, because the e-Closure Principle is a single principle underlying all error rates, researchers may even choose the error rate post hoc, switching, for example, from FDR control to FWER control if the signal is larger than expected.

\revise{In earlier versions of this paper, we highlighted as an important open problem whether and how BH, by far the most popular FDR control method, can be uniformly improved using the e-Closure Principle. This has now been answered by \citet{goeman_uniform_improvement_2026}, who proposes a closed BH procedure that uniformly improves BH under PRDS and under a weaker sufficient condition. Thus, we are optimistic that many more improvements to existing procedures and new procedures remain to be developed using the frameworks formulated in this work.}

The outline of the paper is as follows. We present a general definition of multiple testing procedures that control an expected loss in \Cref{sec: MT}. Our definition explicitly allows for simultaneity and we make the link between simultaneity and post hoc control explicit. Next, we present the general e-Closure Principle and several of its properties in \Cref{sec: principle}, including how it recovers the classical Closure Principle. We then apply the e-Closure Principle to several well-known multiple testing methods for e-values (\Cref{sec: e-combining}) and p-values (\Cref{sec: p-combining}) to show improvements in both power and simultaneity. 
\ifarver
{
\Cref{sec: bh and variants} presents the e-Closure formulation for the Benjamini-Hochberg procedure and its adaptive variants, \revise{showing that the particular e-collections considered there do not improve their power and provide no nontrivial simultaneity except when BH rejects all hypotheses.}
}{}
In \Cref{sec: flexible} we explore the various types of post hoc flexibility enabled by e-Closure including generalizations of existing multiple testing procedures to their simultaneous versions. We discuss how logical relationships between hypotheses can be used to increase power \citep{shaffer1986modified} in \Cref{sec: Shaffer}. We describe computational shortcuts for the methods we have formulated throughout the paper in \Cref{sec: computation} and highlight some real data applications of our new methods in \Cref{sec: real data}. \revise{The novel methods have been implemented in the \texttt{eClosure} and \texttt{eclosure} packages in R and Python, respectively.}

\begin{table}[htbp]\ifarver{}{\fontsize{10}{12}\selectfont}

\caption{
    Number of discoveries made by the \revise{\citet{benjamini_control_false_2001}} procedure and our new procedure $\cBY$ (closed BY) across multiple simulated and real datasets and FDR thresholds $\alpha$. A bold number indicates that $\cBY$ makes more discoveries than BY.}
\label{tab:by_real_data_discoveries}
    \centering
\begin{tabular}{lcccccc}
\toprule
Dataset & \makecell{\# of hypotheses} & \multicolumn{2}{c}{$\alpha = 0.05$} & \multicolumn{2}{c}{$\alpha = 0.1$} & Reference \\ 
\cmidrule(lr){3-4} \cmidrule(lr){5-6}
 &  & BY & $\cBY$ & BY & $\cBY$ &  \\
\midrule
APSAC & 15 & 3 & 3 & 3 & \textbf{5} & \citet{benjamini_controlling_false_1995} \\
NAEP & 34 & 6 & \textbf{8} & 8 & \textbf{11} & \citet{benjamini2000adaptive} \\
PADJUST & 50 & 12 & \textbf{15} & 17 & \textbf{20} & Ex. in \texttt{p.adjust} R function \\
PVALUES & 4,289 & 129 & \textbf{145} & 225 & \textbf{275} & Data in \texttt{fdrtool} R package \\
VANDEVIJVER & 4,919 & 614 & \textbf{677} & 779 & \textbf{866} & \citet{goeman2014multiple} \\
GOLUB & 7,128 & 617 & \textbf{648} & 743 & \textbf{799} & \citet{efron_computer_age_2021} \\
\bottomrule
\end{tabular}
\end{table}

\section{Contributions}
The novel contributions of this paper are the following.

\begin{enumerate}
    \item We formulate and prove the e-Closure Principle for the general multiple testing problem with an error rate introduced in \Cref{sec: MT}. As a result, we prove that a method controls an expected error loss, such as FDR or FWER, if and only if it is a special case of the e-Closure procedure (Section \ref{sec: principle}). Further, we show that the classical Closure Principle \citep{marcus1976closed,sonnemann_general_solutions_2008} is a special case of the e-Closure Principle in \Cref{sec: classical closure}.
    \item We apply the e-Closure Principle to FDR control and use it to construct substantial uniform improvements of the eBH, BY, and Su methods. Our improvements allow greater flexibility in the rejection sets and often permit rejection of larger sets than the original methods at the same nominal FDR level (Sections \ref{sec: e-combining} and \ref{sec: p-combining}). Polynomial time algorithms for these methods are given in Section \ref{sec: computation}. We summarize the methods we improve in \Cref{tab:contributions}. \revise{Further, we analyze the BH procedure and its minimally adaptive variant \citep{solari2017minimally} through the lens of e-Closure in this paper. The e-collections we study do not add power and, apart from BH rejects all hypotheses, do not add nontrivial simultaneity; a different construction developed by \citet{goeman_uniform_improvement_2026} uniformly improves BH using the e-Closure Principle.}
    \begin{table}[ht]\ifarver{}{\fontsize{10}{12}\selectfont}
    
    \centering
    \begin{tabular}{|l|l|c|c|}
    \hline
    \textbf{Previous Method} & \textbf{Our Procedure} & \textbf{Power} & \textbf{Simultaneity} \\
    \hline
    eBH \citep{wang_false_discovery_2022} & $\CeBH$ (\Cref{sec: e-combining}) & \checkmark & \checkmark \\
    \citet{benjamini_control_false_2001} & $\cBY$ (\Cref{sec: p-combining})& \checkmark & \checkmark \\
    \citet{su2018fdr} & $\cSu $ (\Cref{sec: p-combining})& \checkmark & \checkmark \\
    \citet{benjamini_controlling_false_1995} & $\cBH$ (\Cref{sec: bh and variants})& = & = \\
    \citet{storey2004strong} & adaptive $\cBH$ (\Cref{sec: bh and variants}) &= & = \\
    Knockoffs  \citep{barber_controlling_2015} & $\cKnockoffs$ (\Cref{sec: post hoc sets}) & = &\checkmark \\
    \hline
    \end{tabular}

    \caption{A summary of the improvements on existing methods we make through the e-Closure Principle \revise{in this paper}. The \checkmark refers to an improvement, and $=$ refers to the new procedure being no worse than the original. Simultaneity refers to having a larger set of discovery sets simultaneously available for post hoc selection, and power refers to additionally making more discoveries (see \Cref{def: simultaneity improvement}).}
    \label{tab:contributions}
    \end{table}

    \item We generalize the concept of multiple testing to include simultaneous discovery sets (\Cref{sec: MT}), and show that such simultaneity comes for free (without reducing $\alpha$-levels) as a consequence of the e-Closure Principle.
    \item We show that the e-Closure Principle also characterizes procedures that allow researchers post hoc flexibility in the choice of discovery sets, error rates, and nominal levels in \Cref{sec: flexible}. As an example of choosing error rates, researchers can switch to FWER control if they reject a sufficiently small set (\Cref{sec: post hoc loss}), instead of FDR control. Second, we show that, for some FDR controlling methods, researchers may choose post hoc the $\alpha$-level at which FDR is controlled (\Cref{sec: alpha}).
    \item We extend the possibility of exploiting restricted combinations, i.e., logical relationships between hypotheses \citep{shaffer1986modified} from only FWER control to general error rates, in particular FDR.
    \item We also demonstrate clear improvements of our methods over their non-closure versions on a suite of real datasets of p-values and e-values in \Cref{sec: real data}.
\end{enumerate}

\section{A general definition of multiple testing methods} \label{sec: MT}

We will first give a general definition of a multiple testing method that encompasses all methods controlling an expected loss. This general definition will define the scope of the generalized e-Closure Principle that we will present in Section \ref{sec: principle}. Our definition encompasses many well-known error rates, including FWER, FDR, and per-family error rate. Moreover, we will emphasize simultaneity in our general definition, introducing this notion for error rates for which simultaneous control was not yet known, such as FDR.

We use set-centered notation throughout the paper. We denote all sets by capital letters (e.g., $R$), collections of sets by calligraphic letters (e.g., $\mathcal{R}$), and scalars by lowercase letters (e.g., $r$). We use the shorthand $[i] \coloneqq \{1,2,\ldots, i\}$. The power set of $R$ is $2^R$. Random variables and random sets are shown in boldface (e.g., $\mathbf{e}$, $\mathbf{R}$).
Direct inequalities between random variables should always be understood to hold almost surely\revise{, unless stated otherwise}.
We denote $\min(a,b)$ and $\max(a,b)$ as $a \wedge b$ and $a \vee b$, respectively. Further, any omitted proofs of results will be deferred to \Cref{sec: deferred proofs}.

Suppose our data is given by a random variable $\mathbf{X}:\Omega\to S$, where $(\Omega, \mathcal{F})$ and $(S, \mathcal{A})$ are measurable spaces. Let the statistical model $M$ be a set of probability measures on $(\Omega, \mathcal{F})$, e.g., $M = \{\mathrm{P}_\theta\colon \theta \in \Theta\}$ in a parametric model, and denote $\expect_\mathrm{P}$ as the expectation under $\mathrm{P}$. Hypotheses are restrictions to the model $M$, and therefore subsets of $M$. Suppose that we have \revise{$m$ hypotheses, $H_1, \ldots, H_m$, of interest, where each is a subset of the model $M$}. For every $\mathrm{P} \in M$, some hypotheses are true and others false; let $N_\mathrm{P} = \{i\colon \mathrm{P} \in H_i\}$ be the index set of the \revise{true (null) hypotheses} for $\mathrm{P}$.

If we choose to reject the hypotheses with indices in $R \subseteq [m]$, i.e., to claim them as discoveries, then the discoveries in $R \cap N_\mathrm{P}$ are false. Generally, we want to choose $R$ in such a way that the occurrence of such false discoveries is limited. We do this by bounding, in expectation, some nonnegative error function $\F_N(R)$ of $N$ and $R$. We assume that $\F_N(\emptyset) = 0$ for all $N \in 2^{[m]}$, and $\F_\emptyset(R) = 0$ for all $R \in 2^{[m]}$, i.e., no error is incurred if there are no nulls or we make no discoveries. Both conditions are satisfied by all error metrics in \Cref{tab:error_rates}. 

A multiple testing method for error rate $\mathrm{ER}_\F$ is a procedure that produces a random set $\mathbf{R}$.
Such a method controls an error rate defined by $\F_N(\mathbf{R})$ at level $\alpha$ if its construction guarantees that $\F_N(\mathbf{R})$ is bounded in expectation by $\alpha$.
This is formally stated in \Cref{def: single ER}. 

\begin{definition}[Multiple testing method] \label{def: single ER}
$\mathbf{R} \subseteq [m]$ controls $\mathrm{ER}_\F$ at level $\alpha>0$ if, for every $\mathrm{P} \in M$,
\[
\expect_\mathrm{P}\big(\F_{N_\mathrm{P}}(\mathbf{R})\big) \leq \alpha.
\]
\end{definition}

Our results cover the error rates cataloged in \Cref{tab:error_rates}, such as false discovery rate (FDR) and family-wise error rate (FWER), and more general families considered in the multiple testing literature \citep{malek_sequential_multiple_2017a,maurer_optimal_test_2023b}. We also discuss in \Cref{sec: more MT} how our results apply to other error metrics (e.g., mFDR), which are not explicitly the expectation of an error function.
{
\begin{table}[ht]
\ifarver{}{\fontsize{10}{12}\selectfont}
\centering
\begin{tabular}{|l|l|}
\hline
\textbf{Error Rate ($\ER_\F$)} & \textbf{Loss Function ($\F_N(R)$)} \\
\hline
False discovery rate (FDR) & $\FDP_{N}(R)\coloneqq |R \cap N| / (|R|\vee 1)$ \\
\hline
$k$-Familywise error rate (k-FWER), w/ $k=1$ giving FWER & $\ind\{|N\cap R|\geq k\}$ \\
\hline
Per-family error rate (PFER) & $|N\cap R|$ \\
\hline
False discovery exceedance (FDX-$\gamma$) for some $\gamma\in [0,1)$ & $\ind\{\FDP_N(R)> \gamma\}$  \\
\hline
Average error rate (AER) & $|N\cap R|/(|N| \vee 1)$ \\
\hline
\end{tabular}
\caption{Examples of multiple testing error rates ($\ER_\F$) and their corresponding loss functions $\F_N(R)$ for different error rates.\label{tab:error_rates}}
\end{table}
}

Furthermore, in \Cref{def: simultaneous ER} below, we allow the multiple testing method to return not just a single set $\mathbf{R} \subseteq [m]$, but a collection $\boldsymbol{\mathcal{R}} \subseteq 2^{[m]}$ of such sets.
Since $\F_N(\emptyset) = 0$, we may assume that $\rejcol$ is never empty: if $\rejcol$ controls $\mathrm{ER}_\F$, then so does $\rejcol \cup \{\emptyset\}$.

\begin{definition}[Simultaneous multiple testing method] \label{def: simultaneous ER}
\revise{A random collection of rejection sets} $\boldsymbol{\mathcal{R}} \subseteq 2^{[m]}$ controls $\mathrm{ER}_\F$ simultaneously at level $\alpha$ if, for every $\mathrm{P} \in M$,
$
\expect_\mathrm{P}\Big(\max_{R \in \rejcol} \F_{N_\mathrm{P}}(R)\Big) \leq \alpha.
$
\end{definition}

We note that simultaneous $\ER_\F$ control implies standard $\ER_\F$ control. Hence, we focus on proving validity for simultaneous procedures. Furthermore, for each procedure introduced, we also define its classical version, which outputs the largest discovery set from the simultaneous procedure.
\begin{lemma} \label{lem: duality}
\revise{If a random rejection set $\mathbf{R}$ controls $\ER_\F$ at $\alpha$, then so does the random collection $\boldsymbol{\mathcal{R}} = \{\mathbf{R}\}$;
if a random collection $\boldsymbol{\mathcal{R}}$ controls $\ER_\F$ at $\alpha$, then so does any random rejection set $\mathbf{R}$ satisfying $\mathbf{R} \in \boldsymbol{\mathcal{R}}$ almost surely.}
\end{lemma}

Furthermore, we note that simultaneous control is equivalent to \textit{post hoc} control, which allows the user to choose any rejection set $R \in \boldsymbol{\mathcal{R}}$ in an arbitrarily data-dependent way. 
We provide a deeper discussion on the use of procedures with post hoc control in Section~\ref{sec: post hoc sets}.
\revise{For the following lemma, suppose that
$\boldsymbol{\mathcal R}=\mathcal R(\mathbf{X})$, where
$\mathcal R:S\to 2^{2^{[m]}}$ is
$\mathcal A$-measurable and $\mathbf{X}$ includes any auxiliary randomization
used by the procedure.}\begin{lemma}\label{lemma:simult_posthoc}
    We have for all $\mathrm{P}\in M$,
   \begin{align}
  \mathbb{E}_{\mathrm{P}}\left[ \max_{R\in \boldsymbol{\mathcal{R}}} \textnormal{f}_{N_{\mathrm{P}}}(R) \right] =    \revise{\sup_{\mathrm{R}: \mathrm{R}(\mathbf X)\in \boldsymbol{\mathcal{R}} \text{ a.s.}}} \mathbb{E}_{\mathrm{P}}\left[  \textnormal{f}_{N_{\mathrm{P}}}(\mathrm{R}(\mathbf{X})) \right], \label{eq:simult_select}
\end{align}
\revise{where the supremum is over $\mathcal A$-measurable selectors $\mathrm{R}:S\to 2^{[m]}$ satisfying $\mathrm{R}(\mathbf{X})\in\boldsymbol{\mathcal R}$ almost surely.}
\end{lemma}

Though our general definition includes many error rates that have been considered in the literature, it excludes some error rates that are not expectations of some loss $\F_{N_\mathrm{P}}(R)$, such as positive FDR \citep{storey_direct_approach_2002a} and local false discovery rate \citep{efron_large-scale_inference_2010}. However, it does cover procedures performing multiple testing on a data-driven subset of the hypotheses conditionally on the selected set: the final error rate to be controlled, conditional or unconditional, can be written in the form of \Cref{def: single ER} \citep{goeman2024selection}.

\section{The e-Closure Principle} \label{sec: principle}

The central result of this paper is a generalized Closure Principle, the e-Closure Principle, that provides control for any general loss function of the form discussed in \Cref{sec: MT}.
Analogous to the Closure Principle \citep{marcus1976closed, sonnemann_vollstaendigkeitssaetze_fuer_1988,
sonnemann_general_solutions_2008}, this novel principle gives a general recipe for designing multiple testing methods, and recovers the classical Closure Principle (along with all FWER controlling methods) as a special case. After formulating the e-Closure Principle, we will first demonstrate how it recovers the classic Closure Principle, and then proceed with novel results applying the e-Closure Principle to FDR control.

A nonnegative random variable $\mathbf{e}$ is an e-value for a hypothesis $H$ if $\expect_\mathrm{P}(\mathbf{e}) \leq 1$ for all $\mathrm{P} \in H$. A suite of literature has investigated hypothesis testing based on e-values \citep{shafer_testing_betting_2021,grunwals2024safe} rather than on p-values, and an overview of this area can be found in \citet{ramdas_hypothesis_testing_2025}. For a single hypothesis $H$, we may reject $H$ when $\mathbf{e} \geq \alpha^{-1}$, while controlling Type I error at level $\alpha$, since for every $\mathrm{P} \in H$, the following holds by Markov's inequality:
\[
\mathrm{P}(\mathbf{e} \geq \alpha^{-1}) \leq \alpha \cdot \expect_\mathrm{P}(\mathbf{e}) \leq \alpha.
\]
The e-Closure Principle requires an e-value to be defined for every subset of $[m]$.

\begin{definition}
We define an \emph{e-collection} $\mathbf{E} \coloneqq (\mathbf{e}_S)_{S \in 2^{[m]}}$ \citep{shafer_testing_betting_2021,vovk_confidence_discoveries_2023a,xu_post-selection_inference_2022,grunwald_e-posterior_2023,grunwald_neyman-pearson_e-values_2024} to be a collection of nonnegative random variables where $\expect_\pdist(\evalue_{N_\pdist}) \leq 1$ for all $\pdist \in M$, i.e., $\evalue_{N_\pdist}$ is an e-value for each $\pdist \in M$.
\end{definition}
A sufficient but not necessary way to construct an e-collection is to let each $\mathbf{e}_S$ be an e-value for the corresponding $H_S \coloneqq \cap_{i \in S} H_i$ (see \Cref{rem: partitoning}). \revise{We set $\evalue_\emptyset=0$. Its value is immaterial for the e-Closure Principle since the error rate satisfies $\F_\emptyset(R) = 0$, but this convention makes expressions quantified over all $S\in2^{[m]}$ well-defined.}

Define the set of candidate discovery sets for the error metric $\ER_\F$ at level $\alpha$ as
\begin{equation}
\label{eq: e-closure}
\Rcal^{\ER_\F}_\alpha (\mathbf{E})\coloneqq \left\{R \in 2^{[m]}:
\evalue_S \geq \frac{\F_S(R)}{\alpha} \quad \forall {S \in 2^{[m]}}
\right\}.
\end{equation}
To understand why $\Rcal_\alpha^{\ER_\F}(\mathbf{E})$ would control $\ER_\F$ at $\alpha$, note that
by construction, $\F_S(R) \leq \alpha \evalue_S$ for all $S$, so in particular $\F_{N_\pdist}(R)\leq \alpha \evalue_{N_\pdist}$. Therefore $\ER_\F$ is bounded by $\alpha$, the expectation of $\alpha \evalue_{N_\pdist}$.
However, we claim not just that (\ref{eq: e-closure}) controls its error rate ER, but also that every procedure controlling ER is of the form (\ref{eq: e-closure}). This is our first main result, the e-Closure Principle, formulated as Theorem \ref{thm: e-closure}. It says that the class of all procedures of the form (\ref{eq: e-closure}) is complete \citep[Section 1.8]{lehmann2005testing} for the control of error rates of the form $\ER_\F$.

\begin{theorem}[The e-Closure Principle]\label{thm: e-closure}
    $\rejcol$ is \revise{a simultaneous $\ER_\F$ controlling procedure} at level $\alpha$ if and only if there exists an e-collection $\mathbf{E}$ such that $\rejcol \subseteq \Rcal^{\ER_\F}_\alpha(\mathbf{E})$.\end{theorem}
\begin{proof}For any e-collection $\mathbf{E}$ and distribution $\pdist$, we see that
    $$
    \expect_\pdist\left(\max_{R \in \Rcal^{\ER_\F}_\alpha(\mathbf{E})}\ \F_{N_\pdist}(R)\right)
    \leq \expect_\pdist\left(\alpha \cdot \evalue_{N_\pdist}\right)
    \leq \alpha.$$

    The first inequality is by definition of $\Rcal^{\ER_\F}_\alpha(\mathbf{E})$, and the last inequality is from $\evalue_{N_\pdist}$ being an e-value for $\pdist$.

    Suppose $\rejcol$ simultaneously controls $\ER_\F$ over its elements. One can set
    \begin{align}
        \evalue_S =
        \frac{\max_{R \in \rejcol}\ \F_S(R)}{\alpha}\label{eq: max-compat-F}
    \end{align}
    for each $S \in 2^{[m]}$.
For each $\pdist \in M$, we have that
    \begin{align}
        \expect_\pdist(\evalue_{N_\pdist})
        =
        \frac{\expect_\pdist\left(\max_{R \in \rejcol}\ \F_{N_\pdist}(R)\right)}{\alpha}
        \leq 1.
    \end{align}
The inequality is by $\rejcol$ being a simultaneous $\ER_\F$ controlling procedure.
    For each $R \in \rejcol$, we have that:
    \begin{align}
        \frac{\F_S(R)}{\alpha} \leq \frac{\max_{R \in \rejcol}\F_S(R)}{\alpha} = \evalue_S
    \end{align} for all $S \in 2^{[m]}$, which implies that $R \in \Rcal_\alpha^{\ER_\F}(\mathbf{E})$. Thus, we have shown that $\rejcol \subseteq \Rcal_\alpha^{\ER_\F}(\mathbf{E})$.
\end{proof}

One way to view the e-Closure Principle (\ref{eq: e-closure}) is that it divides the model $M$ into parts, designs an error-bound for each part, and combines the results into an overall $\ER_\F$ control procedure. When constructing each of the partial bounds, the method designer may assume that $N_\mathrm{P}$ is known, which gives access to powerful oracle-like information. In FDR control, knowledge of $\pi_{0, \mathrm{P}} = |N_\mathrm{P}|/m$ is often already extremely valuable \citep{storey2004strong, benjamini2006adaptive, blanchard2009adaptive}.

The e-Closure Principle reduces the task of constructing an error control procedure to the much simpler task of constructing e-values for intersection hypotheses. When constructing these e-values, the researcher should take into account any knowledge, or lack of it, on the joint distribution of the data. After choosing e-values, the remaining task, implementing (\ref{eq: e-closure}), is purely computational. This computation may have exponential complexity, since $2^m-1$ e-values must be taken into account. However, computation can be done in polynomial time in important cases, as we shall see in Section \ref{sec: computation}.

Note that \Cref{thm: e-closure} is not tied to our \Cref{def: simultaneous ER} of simultaneous error control; it also characterizes classical procedures that output a single discovery set (i.e., those that follow \Cref{def: single ER}). We make this clear in the following \Cref{cor: principle}, which combines \Cref{thm: e-closure} and \Cref{lem: duality}. In fact, writing an ER controlling set $\mathbf{R}$ as an e-collection is a way of obtaining simultaneous FDR control for a classical FDR control method in a less trivial way than was done by \Cref{lem: duality} ---  we explore this further in \Cref{sec: flexible}.

\begin{corollary} \label{cor: principle}
$\mathbf{R}$ controls $\ER_\F$ at level $\alpha$ if and only if $\mathbf{R} \in \Rcal_\alpha^{\ER_\F}(\mathbf{E})$ for an e-collection $\mathbf{E}$.
\end{corollary}

Following \Cref{cor: principle} we will generally think of e-Closure methods in two different ways: as a method returning a full collection $\Rcal_\alpha^{\ER_\F}(\mathbf{E})$ of rejected sets, and as a method returning a single set $\mathbf{R} \in \Rcal_\alpha^{\ER_\F}(\mathbf{E})$, which is often one of the sets with the largest cardinality in $\Rcal_\alpha^{\ER_\F}(\mathbf{E})$.

Further, once a procedure has been passed through the e-Closure construction, applying the e-Closure construction again will result in the same procedure. This is formally defined as follows.
\begin{proposition} \label{thm: circ}
    Let $\mathbf{E}$ be defined from an arbitrary level $\alpha$ $\ER_\F$ controlling procedure $\rejcol$ as in \eqref{eq: max-compat-F}. Then, if we construct the e-collection $\mathbf{E}' \coloneqq (\evalue_S')_{S \in 2^{[m]}}$ by applying \eqref{eq: max-compat-F} to $\Rcal_\alpha^{\ER_\F}(\mathbf{E})$, we have that $\mathbf{E}' = \mathbf{E}$.
\end{proposition}
The proof follows directly from the definition of e-values in \eqref{eq: max-compat-F} and the definition of $\Rcal_\alpha^{\ER_\F}$ in \eqref{eq: e-closure}. \Cref{thm: circ} \revise{shows that the reconstruction in \Cref{thm: e-closure} recovers a given procedure but should not be expected to improve its power directly. To obtain a power improvement, one must instead construct stronger or better adapted local e-values, as we do later for $\CeBH$, $\cBY$, and $\cSu$. The reconstruction can nevertheless enlarge a classical procedure's collection of simultaneously valid discovery sets, as discussed in \Cref{sec: flexible}.}

\begin{remark}
    Control of $\ER_\F$ in \Cref{thm: e-closure} hinges on the validity of the single e-value $\evalue_{N_\mathrm{P}}$. This implies that relevant properties of $\evalue_{N_\mathrm{P}}$ translate directly to properties of the multiple testing procedure. For example, if $\evalue_{N_\mathrm{P}}$ is an e-value under one of the asymptotic notions considered by \citet{ignatiadis_asymptotic_compound_2025}, control of $\ER_\F$ in the resulting procedure converges to $\alpha$ at the rate at which $\expect(\evalue_{N_\mathrm{P}})$ converges to 1.
\end{remark}
\begin{remark}
The definition of the e-collection immediately allows exploitation of restricted combinations \citep{shaffer1986modified}, since it imposes no conditions on $\evalue_S$ for sets $S$ that do not equal $N_\mathrm{P}$ for any $\mathrm{P} \in M$. We return to this point in \Cref{sec: Shaffer}.
\end{remark}
\begin{remark}
The general multiple testing problem in \Cref{sec: MT} can be viewed as a composite generalized Neyman-Pearson (GNP) problem \citep{grunwald_neyman-pearson_e-values_2024}, and hence the e-Closure Principle as an instance of a maximally compatible decision rule in the language of \citet{grunwald_neyman-pearson_e-values_2024}. \revise{We elaborate on this in more detail in \Cref{sec:related-work}.}
\end{remark}

\subsection{Recovering the classical Closure Principle for FWER control}\label{sec: classical closure}
Now we will see how e-Closure recovers the standard Closure Principle for FWER. Define the following collection of discovery sets:
$$\Rcal^\FWER_\alpha(\mathbf{E}) \coloneqq \left\{R \in 2^{[m]}: \evalue_S \geq \frac{\ind\{|R \cap S| > 0\}}{\alpha}\text{ for all }S \in 2^{[m]}\right\}.$$

Let $R^\FWER_\alpha(\mathbf{E})$ be the largest set in $\Rcal^\FWER_\alpha(\mathbf{E})$ --- we refer to this as the \emph{e-closed procedure} for FWER control. \revise{Such a largest set exists and equals the union of all sets in $\Rcal^\FWER_\alpha(\mathbf{E})$, since the defining condition is preserved under unions: if each set in a collection has empty intersection with every $S$ satisfying $\evalue_S<\alpha^{-1}$, then so does their union.} Note that $R^\FWER_\alpha(\mathbf{E})$ is unique and all members of $\Rcal^\FWER_\alpha(\mathbf{E})$ are subsets of $R^\FWER_\alpha(\mathbf{E})$.

\begin{theorem}[The e-Closure Principle for FWER control]\label{thm: fwer-closure}
    $\rejset$ is a \revise{(simultaneous)} FWER controlling procedure \revise{at level $\alpha$} if and only if there exists an e-collection $\mathbf{E}$ s.t.\ $\rejset \subseteq R^\FWER_\alpha(\mathbf{E})$. \revise{Moreover, the e-collection can be chosen so that every $\evalue_S$ lies in $\{0, \alpha^{-1}\}$.}
\end{theorem}
\begin{proof}
The completeness result follows from applying \Cref{thm: e-closure} to the loss function $$\F_{N}(R) = \ind\{|N \cap R| > 0\}.$$
\revise{To prove the result on the support, start from any e-collection $\mathbf E$ and define}
    \[
    \revise{\evalue'_S=\alpha^{-1}\ind\{\evalue_S\geq\alpha^{-1}\}.}
    \]
    \revise{Then $\evalue'_S\leq\evalue_S$ and $\evalue'_S$ is also an e-value.}
    \revise{Thus $\mathbf E'$ is an e-collection. Because the FWER formulation of the e-Closure Principle compares local e-values only with $0$ and $\alpha^{-1}$, $\mathbf E'$ yields exactly the same rejection set as $\mathbf E$.}
\end{proof}

\begin{corollary}\label{cor: fwer-classical-closure}
\revise{Let $\boldsymbol{\Phi} \coloneqq (\boldsymbol{\varphi}_S)_{S \in 2^{[m]}}$ be a family of level $\alpha$ local intersection tests, and define the closed testing rejection set via the classical Closure Principle \citep{marcus1976closed,sonnemann1982allgemeine,sonnemann_vollstaendigkeitssaetze_fuer_1988,sonnemann_general_solutions_2008}} as
\[
\revise{
R_\alpha(\boldsymbol{\Phi})
\coloneqq
\{i \in [m]: \boldsymbol{\varphi}_S=1 \text{ for all } S \ni i\}.
}
\]
\revise{If we set $\evalue_S=\alpha^{-1} \cdot \boldsymbol{\varphi}_S$ for each $S\in 2^{[m]}$, then we have a valid e-collection by virtue of $\boldsymbol{\Phi}$ being a family of valid local intersection tests and that}
\[
\revise{
R_\alpha(\boldsymbol{\Phi})
=
R^\FWER_\alpha(\mathbf E).
}
\]
\revise{Conversely, suppose that $\mathbf{E}$ is an e-collection for which each $\evalue_S$ is an e-value for the intersection hypothesis $H_S$. If we let $\boldsymbol{\varphi}_S=\ind\{\evalue_S\geq \alpha^{-1}\}$, then $\boldsymbol{\Phi}$ is a family of valid level $\alpha$ local intersection tests, and the standard Closure Principle applied to $\boldsymbol{\Phi}$ yields $R_\alpha(\boldsymbol{\Phi}) = R^\FWER_\alpha(\mathbf E)$.}
\end{corollary}

\revise{Thus, the e-Closure Principle recovers the standard Closure Principle for FWER through the thresholding correspondence between local e-values and local tests.}

In \Cref{sec: post hoc loss}, we will show how the generalization by \citet{genovese2006exceedance} and \citet{goeman2011multiple} of closed testing to simultaneous FDP control is also covered by the e-Closure principle.

\subsection{A novel Closure Principle for FDR control}
Where e-Closure merely recovers known results for FWER, it can achieve uniform improvements of methods in the case of FDR. In the remainder of the paper, we will therefore focus on FDR control and define the following e-Closure procedure for FDR control.
\begin{equation} \label{eq: e-closure-fdr}
\Rcal_\alpha^\FDR(\mathbf{E}) = \Big\{ R \in 2^{[m]}\colon \mathbf{e}_S \geq \frac{\FDP_S(R)}{\alpha} \textrm{\ for all $S \in 2^{[m]}$}\Big\}.
\end{equation}
Equation (\ref{eq: e-closure-fdr}) bounds the FDP for all $S=N_\mathrm{P}$ by $\alpha \mathbf{e}_S$, which has expectation at most $\alpha$.

\begin{theorem}[The e-Closure Principle for FDR control] \label{thm: fdr-closure} $\boldsymbol{\mathcal{R}}$ \revise{simultaneously controls FDR} at level $\alpha$ if and only if $\boldsymbol{\mathcal{R}} \subseteq \Rcal_\alpha^\FDR(\mathbf{E})$ for an e-collection $\mathbf{E}$.
\end{theorem}

\label{eq: closure evalue}

The e-Closure Principle is helpful for constructing novel methods but also for improving existing methods. The ``only if'' part of \Cref{thm: e-closure} asserts that every existing method controlling FDR has an implicit e-collection that can be used to reconstruct the method as a special case of (\ref{eq: e-closure-fdr}). These e-values may be inefficient, e.g.\ because they have an expectation strictly smaller than 1; in such cases, \Cref{thm: fdr-closure} can sometimes be used to propose a superior method based on a suite of stochastically larger e-values. We will give several examples of such improvements in Sections \ref{sec: e-combining} and \ref{sec: p-combining}. However, \Cref{thm: circ} implies that the e-collection constructed in the proof of \Cref{thm: e-closure} cannot be directly used for finding improvements.
It is better to reverse engineer the proof and improve the implicit (or explicit) e-values directly.

\begin{remark}
    Similar to the restriction of support in \Cref{thm: fwer-closure}, one can restrict the support of $\evalue_S$ in any e-collection applicable to \Cref{thm: fdr-closure} to be in $\{k / (\alpha r): r \in [m], k \in \{0\} \cup [r \wedge |S|]\}$ for each $S \in 2^{[m]}$. We expand on the utility of this property for boosting and randomization in \Cref{sec: extensions}.
\end{remark}

The e-Closure Principle is not the only characterization of all FDR controlling procedures.
\citet[Theorem~6.5]{ignatiadis_asymptotic_compound_2025} prove that every FDR controlling procedure can be recovered by eBH applied to compound e-values. We will derive direct connections between the e-Closure Principle and its compound e-value construction in \Cref{sec: e-combining}. Further, the compound e-value formulation of \citet{ignatiadis_asymptotic_compound_2025} is restricted to FDR with a fixed set of base hypotheses, while the e-Closure Principle in \Cref{thm: e-closure} can be applied to problems with richer spaces of discovery sets and risk metrics.

Many FDR control methods, originally formulated in terms of p-values, have recently been reformulated in terms of e-values, generally as a special case of eBH on (compound) e-values. We summarize in \Cref{tab:evalues} how each of the methods we improve upon in this paper has been interpreted as a form of compound e-values and eBH, and give the local e-value that can be used with the e-Closure Principle to derive new procedures or recover the existing procedure, possibly with some added simultaneity. Such simultaneity is often appreciable, as we shall see in \Cref{sec:BY}, but should not be expected to be complete: \citet{finner2001false} showed that any FDR controlling procedure that controls FDR simultaneously for $\rejcol = 2^\rejset$, must control FWER for $\rejset$ (see also \Cref{thm: FWER} below).

\ifarver{
\begin{table}[h]
\caption{E-value based procedures for which we can derive a novel local e-value to either recover or improve. Each original method is equivalent to eBH applied to compound e-values \citep{ignatiadis_asymptotic_compound_2025} and can be improved or recovered via the local e-value for each intersection hypothesis applied to the e-Closure Principle. The local e-values for eBH, BY, and Su are derived in Sections \ref{sec: e-combining}, \ref{sec:BY}, and \ref{sec:Su}, respectively. Here, $\rejset^\BH_\alpha$ and $\rejset^\adaBH_\alpha$ are the sets of rejections of the Benjamini-Hochberg and adaptive Benjamini-Hochberg procedures, respectively, at level $\alpha$. For each $i \in [m]$, $\mathbf{w}_i$ is the coin-flip knockoff sign that results from the knockoff procedure \citep{barber_controlling_2015,candes_panning_gold_2018} and $\mathbf{c}_\alpha$ is the critical value for the knockoff procedure at level $\alpha$.}
\label{tab:evalues}
\centering
\renewcommand{\arraystretch}{2.0}
\begin{tabular}{|c|c|c|}
\hline
\textbf{Method} & \textbf{Implicit (compound) e-value} & \textbf{Local e-value} \\
\hline
\makecell{eBH} & \makecell{$\evalue_i$ \\ \citep{wang_false_discovery_2022}} & \makecell{$\displaystyle \frac{1}{|S|} \sum_{i \in S}\evalue_i$ \eqref{eq: mean e}} \\
\hline
\makecell{Benjamini-Yekutieli  \\ } & \makecell{$\displaystyle e_{m}(\pvalue_i)$ \\ \citep{xu_post-selection_inference_2022}} & $\displaystyle  \frac{1}{|S|}\sum_{i \in S} e_{|S|}(\pvalue_i)$ \eqref{eq: e for BY} \\
\hline
\makecell{Su} & \makecell{$(\ell_\alpha \pvalue_i \vee \alpha)^{-1}$\\ (this paper --- eq. \ref{eq: e Su})} & $(\ell_\alpha \pvalue_S \vee \alpha)^{-1} \eqref{eq: e Su}$ \\
\hline
\makecell{Benjamini-Hochberg } & \makecell{$\displaystyle \frac{m}{\alpha |\rejset_\alpha^\BH|} \mathbf{1}\Big\{\pvalue_i \leq \frac{\alpha |\rejset_\alpha^\BH|}{m} \Big\}$ \\ \citep{li_note_e-values_2025}} & \makecell{$\displaystyle  \frac{m}{ |S|} \sum_{i \in S} \frac{1}{\alpha |\rejset_\alpha^\BH|} \mathbf{1}\Big\{\pvalue_i \leq \frac{\alpha |\rejset_\alpha^\BH|}{m} \Big\}$ \eqref{eq: BH local e-value}} \\
\hline
\makecell{Adaptive BH \\ (e.g., Storey)} & \makecell{$\displaystyle \frac{m}{\alpha |\rejset_\alpha^\adaBH|} \mathbf{1}\Big\{\pvalue_i \leq \frac{\alpha |\rejset_\alpha^\adaBH|}{\hat{\boldsymbol{\pi}}_0 m} \Big\}$\\
\citep{li_note_e-values_2025}
} & $\displaystyle \sum_{i \in S}\frac{1}{\alpha |\rejset_\alpha^\adaBH|} \mathbf{1}\Big\{\pvalue_i \leq \frac{\alpha |\rejset_\alpha^\adaBH|}{\hat{\boldsymbol{\pi}}_0 m} \Big\}$ \eqref{eq: adaBH local e-value}  \\
\hline
\makecell{Knockoffs} & \makecell{$\displaystyle \frac{m \mathbf{1}\{\mathbf{w}_i \geq \mathbf{c}_\alpha^\Kn\}}{1+\sum_{i \in [m]}1\{\mathbf{w}_i \leq -\mathbf{c}_\alpha^\Kn\}}$ \\ \citep{ren2024derandomised}} & $\displaystyle \frac{ \sum_{i \in S} \mathbf{1}\{\mathbf{w}_i \geq \mathbf{c}_\alpha^\Kn\} }{ 1 + \sum_{i \in S}\mathbf{1}\{\mathbf{w}_i \leq -\mathbf{c}_\alpha^\Kn\} }$ \eqref{eq: knockoff local e-value}\\
\hline
\end{tabular}
\end{table}
}{
\begin{table}[h]\small
\caption{E-value based procedures for which we can derive a novel local e-value to either recover or improve. Each original method is equivalent to eBH applied to compound e-values \citep{ignatiadis_asymptotic_compound_2025} and can be improved or recovered via the local e-value for each intersection hypothesis applied to the e-Closure Principle. Here, $\rejset^\BH_\alpha$ and $\rejset^\adaBH_\alpha$ are the sets of rejections of the Benjamini-Hochberg and adaptive Benjamini-Hochberg procedures, respectively, at level $\alpha$. For each $i \in [m]$, $\mathbf{w}_i$ is the coin-flip knockoff sign that results from the knockoff procedure \citep{barber_controlling_2015,candes_panning_gold_2018} and \revise{$\mathbf{c}_\alpha^\Kn$} is the critical value for the knockoff procedure at level $\alpha$.}
\label{tab:evalues}
\centering
\renewcommand{\arraystretch}{2.0}
\begin{tabular}{|c|c|c|}
\hline
\textbf{Method} & \textbf{Implicit (compound) e-value} & \textbf{Local e-value} \\
\hline
eBH & $\evalue_i$ \citep{wang_false_discovery_2022} & $\displaystyle \frac{1}{|S|} \sum_{i \in S}\evalue_i$ \eqref{eq: mean e} \\
\hline
\makecell{Benjamini-\\Yekutieli} & $\displaystyle e_{m}(\pvalue_i)$ \citep{xu_post-selection_inference_2022} & $\displaystyle  \frac{1}{|S|}\sum_{i \in S} e_{|S|}(\pvalue_i)$ \eqref{eq: e for BY} \\
\hline
Su & $(\ell_\alpha \pvalue_i \vee \alpha)^{-1}$ (this paper --- eq. \ref{eq: e Su}) & $(\ell_\alpha \pvalue_S \vee \alpha)^{-1}$ \eqref{eq: e Su} \\
\hline
\makecell{Benjamini-\\Hochberg} & $\displaystyle \frac{m}{\alpha |\rejset_\alpha^\BH|} \mathbf{1}\Big\{\pvalue_i \leq \frac{\alpha |\rejset_\alpha^\BH|}{m} \Big\}$ \citep{li_note_e-values_2025} & $\displaystyle  \frac{m}{ |S|} \sum_{i \in S} \frac{1}{\alpha |\rejset_\alpha^\BH|} \mathbf{1}\Big\{\pvalue_i \leq \frac{\alpha |\rejset_\alpha^\BH|}{m} \Big\}$ \eqref{eq: BH local e-value} \\
\hline
\makecell{Adaptive BH \\ (e.g., Storey)} & $\displaystyle \frac{m}{\alpha |\rejset_\alpha^\adaBH|} \mathbf{1}\Big\{\pvalue_i \leq \frac{\alpha |\rejset_\alpha^\adaBH|}{\hat{\boldsymbol{\pi}}_0 m} \Big\}$ \citep{li_note_e-values_2025} & $\displaystyle \sum_{i \in S}\frac{1}{\alpha |\rejset_\alpha^\adaBH|} \mathbf{1}\Big\{\pvalue_i \leq \frac{\alpha |\rejset_\alpha^\adaBH|}{\hat{\boldsymbol{\pi}}_0 m} \Big\}$ \eqref{eq: adaBH local e-value}  \\
\hline
Knockoffs & $\displaystyle \frac{m \mathbf{1}\{\mathbf{w}_i \geq \mathbf{c}_\alpha^\Kn\}}{1+\sum_{i \in [m]}1\{\mathbf{w}_i \leq -\mathbf{c}_\alpha^\Kn\}}$ \citep{ren2024derandomised} & $\displaystyle \frac{ \sum_{i \in S} \mathbf{1}\{\mathbf{w}_i \geq \mathbf{c}_\alpha^\Kn\} }{ 1 + \sum_{i \in S}\mathbf{1}\{\mathbf{w}_i \leq -\mathbf{c}_\alpha^\Kn\} }$ \eqref{eq: knockoff local e-value}\\
\hline
\end{tabular}
\end{table}
}
\section{Improving eBH through merging e-values} \label{sec: e-combining}

As a first application of the e-Closure Principle, we will look at the important special case that we have e-values available for the hypotheses $H_1, \ldots, H_m$, and that the FDR controlling $\boldsymbol{\mathcal{R}}_\alpha$ should be a function of these e-values. Let $\mathbf{e}_{(i)}$ denote the $i$th largest e-value for $i \in [m]$. We make no assumptions on the joint distribution of the e-values.

For this situation, \citet{wang_false_discovery_2022} proposed the eBH procedure. It is equivalent to the BH procedure \citep{benjamini_controlling_false_1995} applied to $\mathbf{p}_1 =1/\mathbf{e}_1, \ldots, \mathbf{p}_m = 1/\mathbf{e}_m$. \revise{Of course, unlike eBH, BH does not provide FDR control under arbitrary dependence; its classical finite-sample validity result holds under PRDS \citep{benjamini_control_false_2001}. In contrast, the eBH procedure is valid for any joint distribution of the e-values.}

\revise{Let $\rejset^{(i)}$ denote the indices corresponding to the $i$ largest e-values (or smallest p-values), breaking ties by the smallest hypothesis index, and let $\rejset^{(0)} \coloneqq \emptyset$.}
Then $\eBH$ rejects:
\begin{equation} \label{eq: eBH}
    \rejset_\alpha^\eBH = \rejset^{(\mathbf{r}_\alpha^\eBH)}, \text{ where }\mathbf{r}_\alpha^\eBH  \coloneqq \max\{1 \leq r \leq m\colon r\mathbf{e}_{(r)} \geq m/\alpha\}
\end{equation}
where $\mathbf{r}_\alpha^\eBH$ is taken as 0 if the maximum does not exist.
We will show in the next section that eBH, though admissible for compound e-values \citep{ignatiadis_asymptotic_compound_2025}, is not admissible for regular ones.

\subsection{The closed $\textCeBH$ procedure}
We will now use the e-Closure Principle to propose an alternative method for controlling FDR based on e-values under arbitrary dependence. \revise{A natural way to merge e-values is to take their arithmetic average. Furthermore, \citet{vovk_e-values_calibration_2021} showed that the arithmetic average, possibly mixed with the trivial e-value of 1, is the only admissible symmetric e-merging function for arbitrarily dependent e-values. We will use the unmixed average}
\begin{equation} \label{eq: mean e}
\mathbf{e}_S = \frac1{|S|} \sum_{i \in S} \mathbf{e}_i
\end{equation}
as an e-value for $H_S$, \revise{for each nonempty $S \in 2^{[m]}$, together with the convention $\evalue_\emptyset=0$.}
We can now apply \eqref{eq: e-closure} using $\mathbf{E} = (\mathbf{e}_S)_{S \in 2^{[m]}}$ as the e-collection, to obtain the simultaneous $\CeBH$ procedure:
\begin{equation} \label{eq: proc mean e}
    \rejcol^{\CeBH}_\alpha \coloneqq \Rcal_\alpha^\FDR(\mathbf{E}) = \left\{ R \in 2^{[m]}\colon  \frac{1}{|S|} \sum_{i \in S} \mathbf{e}_i \geq \frac{\FDP_S(R)}{\alpha} \textrm{\ for all $S \in 2^{[m]}$}\right\}.
\end{equation}
We will refer to any discovery set that belongs in the simultaneous $\CeBH$ procedure as being \emph{mean-consistent}, as an analog of the self-consistency condition of \citet{blanchard_two_simple_2008}.

Often we are most interested in the largest set.
\revise{When one uses symmetric local e-values that are coordinatewise increasing in the input e-values (or coordinatewise decreasing in the input p-values), at least one of the largest discovery sets within the FDR controlling e-Closure collection can always be taken to be ordered, i.e., it consists of the largest e-values or smallest p-values. \Cref{lem: ordered closure sets} formalizes this exchange argument for a general symmetric loss.}

Thus, we define
\begin{align}
    \rejset^\CeBH_\alpha \coloneqq \rejset_\alpha^{(\mathbf{r}_\alpha^\CeBH)}, \text{ where }\mathbf{r}_\alpha^\CeBH \coloneqq \max \{r \in [m]: \rejset^{(r)} \in \rejcol^\CeBH_\alpha\}.
\end{align}
This is one of the sets with the largest possible size in $\rejcol^\CeBH_\alpha$, i.e., $\rejset^\CeBH_\alpha \in \operatorname*{argmax}_{R \in \rejcol^\CeBH_\alpha} |R|$.
However, the largest set does not have to be unique (see \Cref{ex: non-unique}).
This procedure controls FDR under any joint distribution of the e-values by \Cref{thm: fdr-closure} and the arithmetic mean being a valid e-merging function under arbitrary dependence. Moreover, the procedure uniformly improves upon eBH in both the choice of discoveries that are simultaneously available and in the maximum number of discoveries, as \Cref{thm: eBH} asserts.
\revise{To make this terminology precise, let $\boldsymbol{\mathcal{R}}(x) \subseteq 2^{[m]}$ be the collection returned by a simultaneous procedure and let $\mathbf{S}(x)$ be the set returned by a classical procedure at data realization $x$.}
\begin{definition}[Simultaneity and power improvements] \label{def: simultaneity improvement}
    \revise{The simultaneous procedure $\boldsymbol{\mathcal R}$ is a \emph{uniform simultaneity improvement} over $\mathbf S$ if $\mathbf S(x)\in\boldsymbol{\mathcal R}(x)$ for every data realization $x$ and, for at least one realization $x$, the collection $\boldsymbol{\mathcal R}(x)$ also contains a nonempty set $R\neq\mathbf S(x)$.
    It is a \emph{uniform power improvement} if, for every realization of $x$, there is an $R\in\boldsymbol{\mathcal R}(x)$ with $R\supseteq\mathbf S(x)$ and, for at least one realization $x$, the collection $\boldsymbol{\mathcal R}(x)$ also contains a nonempty set $R\supsetneq\mathbf S(x)$.}
\end{definition}\revise{This definition is deterministic as “uniform” refers to the pointwise containment for every possible sample realization in some sense, while improvement requires that the relevant improvement region is nonempty.
    If that region has positive probability under some $\pdist\in M$, the corresponding random procedures also differ with positive probability under $\pdist$.
    That additional probabilistic conclusion depends on the joint distribution; for example, some discrete distributions may assign the displayed improvement regions probability zero.
    In \Cref{sec: real data}, we observe the improvements empirically.
}
\begin{theorem} \label{thm: eBH}
    \revise{When $m>1$, $\rejcol^\CeBH_\alpha$ is a uniform improvement in both simultaneity and power over $\rejset_\alpha^\eBH$.}
\end{theorem}
\begin{proof}
Let $R$ be any self-consistent discovery set \citep{wang_false_discovery_2022}, which means that $\mathbf{e}_i \geq m/(\alpha |R|)$ for all $i \in R$. Note that $\rejset^\eBH_\alpha$ is the largest self-consistent set. \revise{If $R\cap S=\emptyset$, the e-Closure inequality for $S$ is immediate; otherwise,}
\begin{align}
\frac{1}{|S|} \sum_{i \in S} \mathbf{e}_i
&\geq
\frac{|R \cap S|}{|S|} \cdot \min_{i \in  R \cap S}\mathbf{e}_{i}
\geq
\frac{|R \cap S|}{|S|} \cdot \min_{i \in  R}\mathbf{e}_{i}
\geq
\frac{|R \cap S|}{m} \cdot \min_{i \in  R}\mathbf{e}_{i}
\geq
\frac{|R \cap S|}{m} \cdot \frac{m}{\alpha |R|}
=
\frac{\FDP_S(R)}{\alpha}
\end{align}
which shows that $R \in \rejcol_\alpha^\CeBH$. The same is trivially true if $R=\emptyset$.
To show a power improvement, let $m>1$ and consider the event that $\mathbf{e}_1=(m-1/2)/\alpha$, $\mathbf{e}_2 = 1/(2\alpha)$, $\mathbf{e}_3 = \ldots = \mathbf{e}_m = 0$. Then $\mathbf{R}_\alpha^\eBH = \emptyset$, but $\{1\} \in  \rejset^\CeBH_\alpha$, because $|S|^{-1} \sum_{i \in S} \mathbf{e}_i \geq \alpha^{-1}$ whenever $\{1\} \in S$.
\end{proof}

\begin{remark} \citet{ignatiadis2024values} formulated a slight uniform improvement on the eBH procedure, called minimally adaptive eBH. \revise{For $m>4$, a result analogous to \Cref{thm: eBH} holds about this procedure. We prove in \Cref{sec: minimally adaptive eBH} that $\CeBH$ also improves it in both simultaneity and power, subject to the same positive probability qualification.}
\end{remark}

\begin{remark}
It is not immediate that the $\CeBH$ procedure is admissible, despite it utilizing an admissible e-merging function for each intersection hypothesis. \revise{We present a necessary condition for admissibility in \Cref{sec: admissibility} via a characterization of generalized e-merging functions \citet{vovk2022admissible,wang_only_admissible_2025,clerico_simple_geometric_2026}, and $\textCeBH$ does satisfy such a condition. In an earlier draft, we posed a complete characterization of admissible e-value based procedures FDR controlling procedures as an open problem. Recently, \citet{sun_admissibility_complete_2026} provided such a characterization and showed that every constant-free $\textCeBH$ procedure derived from weighted mean local e-values is an admissible procedure at every $\alpha$, and within the class of symmetric procedures, $\textCeBH$ is not the largest element when $\alpha$ is sufficiently large.} We also elaborate on some extensions of $\textCeBH$ in \Cref{sec: extensions} that can improve or modify the procedure further.
\end{remark}

\subsection{When does $\textCeBH$ improve over eBH?}
The improvement from eBH to $\textCeBH$ can be substantial, as can be gauged from the proof of \Cref{thm: eBH}, in which each of the four inequalities leaves a substantial amount of room, and each has a different worst case in which it is an equality. The $\CeBH$ procedure often rejects more hypotheses than eBH, and may reject some hypotheses in the event that eBH does not reject any. An extreme example of this is given in Example \ref{ex: eBH+}, in which $\CeBH$ rejects all hypotheses, but eBH none. It is illustrated in Figure \ref{fig:example1}.

\begin{example} \label{ex: eBH+}
Suppose $m >1$, and $(2m-2i+1)/(m\alpha) \leq \mathbf{e}_i < m/(i\alpha)$ for $i=1,\ldots, m$. Then eBH rejects nothing, but $\textCeBH$ rejects $[m]$.
\end{example}
We defer a proof of this to \Cref{sec: eBH+ example proof}.

\begin{figure}[!ht]
\centering
\begin{tikzpicture}[scale=0.7]
\begin{axis}[
    xmin=1, xmax=20,
    xtick={5,10,15,20},
    ymin=0, ymax=400,
    ylabel=$e$-value,
    xlabel=rank,
    height=10cm, width=10cm,
    legend style={legend columns=1}
]

    \addplot[color=blue, very thick, mark=*, mark size=2pt, fill=blue, fill opacity=0.5, draw opacity=0.75] coordinates {
( 1,400 ) ( 2,200 ) ( 3,133.333 ) ( 4,100 ) ( 5,80 ) ( 6,66.667 ) ( 7,57.143 ) ( 8,50 ) ( 9,44.444 ) ( 10,40 ) ( 11,36.364 ) ( 12,33.333 ) ( 13,30.769 ) ( 14,28.571 ) ( 15,26.667 ) ( 16,25 ) ( 17,23.529 ) ( 18,22.222 ) ( 19,21.053 ) ( 20,20 ) ( 20,0 ) ( 19,0 ) ( 18,0 ) ( 17,0 ) ( 16,0 ) ( 15,0 ) ( 14,0 ) ( 13,0 ) ( 12,0 ) ( 11,0 ) ( 10,0 ) ( 9,0 ) ( 8,0 ) ( 7,40 ) ( 6,45.714 ) ( 5,51.429 ) ( 4,57.143 ) ( 3,62.857 ) ( 2,68.571 ) ( 1,74.286 )
    };
    \addlegendentry{$k=7$}

    \addplot[color=red, very thick, mark=*, mark size=2pt, fill=red, fill opacity=0.5, draw opacity=0.75] coordinates {
        (1,400) (2,200) (3,133.333) (4,100) (5,80) (6,66.667) (7,57.143) (8,50)
        (9,44.444) (10,40) (11,36.364) (12,33.333) (13,30.769) (14,28.571)
        (15,26.667) (16,25) (17,23.529) (18,22.222) (19,21.053) (20,20)
        (20,1) (19,3) (18,5) (17,7) (16,9) (15,11) (14,13) (13,15)
        (12,17) (11,19) (10,21) (9,23) (8,25) (7,27) (6,29) (5,31)
        (4,33) (3,35) (2,37) (1,39)
    };
    \addlegendentry{$k=20$}

    \addplot[color=black, very thick, mark=*, mark size=2pt] coordinates {
        (1,400) (2,200) (3,133.333) (4,100) (5,80) (6,66.667) (7,57.143) (8,50)
        (9,44.444) (10,40) (11,36.364) (12,33.333) (13,30.769) (14,28.571)
        (15,26.667) (16,25) (17,23.529) (18,22.222) (19,21.053) (20,20)
    };
\end{axis}
\end{tikzpicture}~
\begin{tikzpicture}[scale=0.7, define rgb/.code={\definecolor{mycolor}{RGB}{#1}}, rgb color/.style={define rgb={#1},mycolor}]
\begin{axis}[
    xmin=1, xmax=20,
    xtick={5,10,15,20},
    ymin=0, ymax=400,
    xlabel=rank,
    height=10cm, width=10cm,
    ylabel=,
    legend style={legend columns=2}
]
\addplot[rgb color={231, 29, 67}, very thick, mark=., mark size=2pt, fill opacity=0.5, draw opacity=0.75] coordinates {
    (1,400) (2,200) (3,133.333) (4,100) (5,80) (6,66.667) (7,57.143) (8,50)
    (9,44.444) (10,40) (11,36.364) (12,33.333) (13,30.769) (14,28.571)
    (15,26.667) (16,25) (17,23.529) (18,22.222) (19,21.053) (20,20)
    (20,1) (19,3) (18,5) (17,7) (16,9) (15,11) (14,13) (13,15)
    (12,17) (11,19) (10,21) (9,23) (8,25) (7,27) (6,29) (5,31)
    (4,33) (3,35) (2,37) (1,39)
};
\addplot[rgb color={255, 0, 0}, very thick,  mark size=2pt, mark=., fill opacity=0.5, draw opacity=0.75] coordinates {
( 1,400 ) ( 2,200 ) ( 3,133.333 ) ( 4,100 ) ( 5,80 ) ( 6,66.667 ) ( 7,57.143 ) ( 8,50 ) ( 9,44.444 ) ( 10,40 ) ( 11,36.364 ) ( 12,33.333 ) ( 13,30.769 ) ( 14,28.571 ) ( 15,26.667 ) ( 16,25 ) ( 17,23.529 ) ( 18,22.222 ) ( 19,21.053 ) ( 20,20 ) ( 20,0 ) ( 19,2.105 ) ( 18,4.211 ) ( 17,6.316 ) ( 16,8.421 ) ( 15,10.526 ) ( 14,12.632 ) ( 13,14.737 ) ( 12,16.842 ) ( 11,18.947 ) ( 10,21.053 ) ( 9,23.158 ) ( 8,25.263 ) ( 7,27.368 ) ( 6,29.474 ) ( 5,31.579 ) ( 4,33.684 ) ( 3,35.789 ) ( 2,37.895 ) ( 1,40 )
};
\addplot[rgb color={255, 55, 0}, very thick,  mark size=2pt, mark=., fill opacity=0.5, draw opacity=0.75] coordinates {
( 1,400 ) ( 2,200 ) ( 3,133.333 ) ( 4,100 ) ( 5,80 ) ( 6,66.667 ) ( 7,57.143 ) ( 8,50 ) ( 9,44.444 ) ( 10,40 ) ( 11,36.364 ) ( 12,33.333 ) ( 13,30.769 ) ( 14,28.571 ) ( 15,26.667 ) ( 16,25 ) ( 17,23.529 ) ( 18,22.222 ) ( 19,21.053 ) ( 20,20 ) ( 20,0 ) ( 19,0 ) ( 18,3.333 ) ( 17,5.556 ) ( 16,7.778 ) ( 15,10 ) ( 14,12.222 ) ( 13,14.444 ) ( 12,16.667 ) ( 11,18.889 ) ( 10,21.111 ) ( 9,23.333 ) ( 8,25.556 ) ( 7,27.778 ) ( 6,30 ) ( 5,32.222 ) ( 4,34.444 ) ( 3,36.667 ) ( 2,38.889 ) ( 1,41.111 )
};
\addplot[rgb color={255, 110, 0}, very thick,  mark size=2pt, mark=., fill opacity=0.5, draw opacity=0.75] coordinates {
( 1,400 ) ( 2,200 ) ( 3,133.333 ) ( 4,100 ) ( 5,80 ) ( 6,66.667 ) ( 7,57.143 ) ( 8,50 ) ( 9,44.444 ) ( 10,40 ) ( 11,36.364 ) ( 12,33.333 ) ( 13,30.769 ) ( 14,28.571 ) ( 15,26.667 ) ( 16,25 ) ( 17,23.529 ) ( 18,22.222 ) ( 19,21.053 ) ( 20,20 ) ( 20,0 ) ( 19,0 ) ( 18,0 ) ( 17,4.706 ) ( 16,7.059 ) ( 15,9.412 ) ( 14,11.765 ) ( 13,14.118 ) ( 12,16.471 ) ( 11,18.824 ) ( 10,21.176 ) ( 9,23.529 ) ( 8,25.882 ) ( 7,28.235 ) ( 6,30.588 ) ( 5,32.941 ) ( 4,35.294 ) ( 3,37.647 ) ( 2,40 ) ( 1,42.353 )
};
\addplot[rgb color={255, 165, 0}, very thick,  mark size=2pt, mark=., fill opacity=0.5, draw opacity=0.75] coordinates {
( 1,400 ) ( 2,200 ) ( 3,133.333 ) ( 4,100 ) ( 5,80 ) ( 6,66.667 ) ( 7,57.143 ) ( 8,50 ) ( 9,44.444 ) ( 10,40 ) ( 11,36.364 ) ( 12,33.333 ) ( 13,30.769 ) ( 14,28.571 ) ( 15,26.667 ) ( 16,25 ) ( 17,23.529 ) ( 18,22.222 ) ( 19,21.053 ) ( 20,20 ) ( 20,0 ) ( 19,0 ) ( 18,0 ) ( 17,0 ) ( 16,6.25 ) ( 15,8.75 ) ( 14,11.25 ) ( 13,13.75 ) ( 12,16.25 ) ( 11,18.75 ) ( 10,21.25 ) ( 9,23.75 ) ( 8,26.25 ) ( 7,28.75 ) ( 6,31.25 ) ( 5,33.75 ) ( 4,36.25 ) ( 3,38.75 ) ( 2,41.25 ) ( 1,43.75 )
};
\addplot[rgb color={255, 195, 0}, very thick,  mark size=2pt, mark=., fill opacity=0.5, draw opacity=0.75] coordinates {
( 1,400 ) ( 2,200 ) ( 3,133.333 ) ( 4,100 ) ( 5,80 ) ( 6,66.667 ) ( 7,57.143 ) ( 8,50 ) ( 9,44.444 ) ( 10,40 ) ( 11,36.364 ) ( 12,33.333 ) ( 13,30.769 ) ( 14,28.571 ) ( 15,26.667 ) ( 16,25 ) ( 17,23.529 ) ( 18,22.222 ) ( 19,21.053 ) ( 20,20 ) ( 20,0 ) ( 19,0 ) ( 18,0 ) ( 17,0 ) ( 16,0 ) ( 15,8 ) ( 14,10.667 ) ( 13,13.333 ) ( 12,16 ) ( 11,18.667 ) ( 10,21.333 ) ( 9,24 ) ( 8,26.667 ) ( 7,29.333 ) ( 6,32 ) ( 5,34.667 ) ( 4,37.333 ) ( 3,40 ) ( 2,42.667 ) ( 1,45.333 )
};
\addplot[rgb color={255, 225, 0}, very thick,  mark size=2pt, mark=., fill opacity=0.5, draw opacity=0.75] coordinates {
( 1,400 ) ( 2,200 ) ( 3,133.333 ) ( 4,100 ) ( 5,80 ) ( 6,66.667 ) ( 7,57.143 ) ( 8,50 ) ( 9,44.444 ) ( 10,40 ) ( 11,36.364 ) ( 12,33.333 ) ( 13,30.769 ) ( 14,28.571 ) ( 15,26.667 ) ( 16,25 ) ( 17,23.529 ) ( 18,22.222 ) ( 19,21.053 ) ( 20,20 ) ( 20,0 ) ( 19,0 ) ( 18,0 ) ( 17,0 ) ( 16,0 ) ( 15,0 ) ( 14,10 ) ( 13,12.857 ) ( 12,15.714 ) ( 11,18.571 ) ( 10,21.429 ) ( 9,24.286 ) ( 8,27.143 ) ( 7,30 ) ( 6,32.857 ) ( 5,35.714 ) ( 4,38.571 ) ( 3,41.429 ) ( 2,44.286 ) ( 1,47.143 )
};
\addplot[rgb color={255, 255, 0}, very thick,  mark size=2pt, mark=., fill opacity=0.5, draw opacity=0.75] coordinates {
( 1,400 ) ( 2,200 ) ( 3,133.333 ) ( 4,100 ) ( 5,80 ) ( 6,66.667 ) ( 7,57.143 ) ( 8,50 ) ( 9,44.444 ) ( 10,40 ) ( 11,36.364 ) ( 12,33.333 ) ( 13,30.769 ) ( 14,28.571 ) ( 15,26.667 ) ( 16,25 ) ( 17,23.529 ) ( 18,22.222 ) ( 19,21.053 ) ( 20,20 ) ( 20,0 ) ( 19,0 ) ( 18,0 ) ( 17,0 ) ( 16,0 ) ( 15,0 ) ( 14,0 ) ( 13,12.308 ) ( 12,15.385 ) ( 11,18.462 ) ( 10,21.538 ) ( 9,24.615 ) ( 8,27.692 ) ( 7,30.769 ) ( 6,33.846 ) ( 5,36.923 ) ( 4,40 ) ( 3,43.077 ) ( 2,46.154 ) ( 1,49.231 )
};
\addplot[rgb color={170, 213, 0}, very thick,  mark size=2pt, mark=., fill opacity=0.5, draw opacity=0.75] coordinates {
( 1,400 ) ( 2,200 ) ( 3,133.333 ) ( 4,100 ) ( 5,80 ) ( 6,66.667 ) ( 7,57.143 ) ( 8,50 ) ( 9,44.444 ) ( 10,40 ) ( 11,36.364 ) ( 12,33.333 ) ( 13,30.769 ) ( 14,28.571 ) ( 15,26.667 ) ( 16,25 ) ( 17,23.529 ) ( 18,22.222 ) ( 19,21.053 ) ( 20,20 ) ( 20,0 ) ( 19,0 ) ( 18,0 ) ( 17,0 ) ( 16,0 ) ( 15,0 ) ( 14,0 ) ( 13,0 ) ( 12,15 ) ( 11,18.333 ) ( 10,21.667 ) ( 9,25 ) ( 8,28.333 ) ( 7,31.667 ) ( 6,35 ) ( 5,38.333 ) ( 4,41.667 ) ( 3,45 ) ( 2,48.333 ) ( 1,51.667 )
};
\addplot[rgb color={85, 170, 0}, very thick,  mark size=2pt, mark=., fill opacity=0.5, draw opacity=0.75] coordinates {
( 1,400 ) ( 2,200 ) ( 3,133.333 ) ( 4,100 ) ( 5,80 ) ( 6,66.667 ) ( 7,57.143 ) ( 8,50 ) ( 9,44.444 ) ( 10,40 ) ( 11,36.364 ) ( 12,33.333 ) ( 13,30.769 ) ( 14,28.571 ) ( 15,26.667 ) ( 16,25 ) ( 17,23.529 ) ( 18,22.222 ) ( 19,21.053 ) ( 20,20 ) ( 20,0 ) ( 19,0 ) ( 18,0 ) ( 17,0 ) ( 16,0 ) ( 15,0 ) ( 14,0 ) ( 13,0 ) ( 12,0 ) ( 11,18.182 ) ( 10,21.818 ) ( 9,25.455 ) ( 8,29.091 ) ( 7,32.727 ) ( 6,36.364 ) ( 5,40 ) ( 4,43.636 ) ( 3,47.273 ) ( 2,50.909 ) ( 1,54.545 )
};
\addplot[rgb color={0, 128, 0}, very thick,  mark size=2pt, mark=., fill opacity=0.5, draw opacity=0.75] coordinates {
( 1,400 ) ( 2,200 ) ( 3,133.333 ) ( 4,100 ) ( 5,80 ) ( 6,66.667 ) ( 7,57.143 ) ( 8,50 ) ( 9,44.444 ) ( 10,40 ) ( 11,36.364 ) ( 12,33.333 ) ( 13,30.769 ) ( 14,28.571 ) ( 15,26.667 ) ( 16,25 ) ( 17,23.529 ) ( 18,22.222 ) ( 19,21.053 ) ( 20,20 ) ( 20,0 ) ( 19,0 ) ( 18,0 ) ( 17,0 ) ( 16,0 ) ( 15,0 ) ( 14,0 ) ( 13,0 ) ( 12,0 ) ( 11,0 ) ( 10,22 ) ( 9,26 ) ( 8,30 ) ( 7,34 ) ( 6,38 ) ( 5,42 ) ( 4,46 ) ( 3,50 ) ( 2,54 ) ( 1,58 )
};
\addplot[rgb color={0, 85, 85}, very thick,  mark size=2pt, mark=., fill opacity=0.5, draw opacity=0.75] coordinates {
( 1,400 ) ( 2,200 ) ( 3,133.333 ) ( 4,100 ) ( 5,80 ) ( 6,66.667 ) ( 7,57.143 ) ( 8,50 ) ( 9,44.444 ) ( 10,40 ) ( 11,36.364 ) ( 12,33.333 ) ( 13,30.769 ) ( 14,28.571 ) ( 15,26.667 ) ( 16,25 ) ( 17,23.529 ) ( 18,22.222 ) ( 19,21.053 ) ( 20,20 ) ( 20,0 ) ( 19,0 ) ( 18,0 ) ( 17,0 ) ( 16,0 ) ( 15,0 ) ( 14,0 ) ( 13,0 ) ( 12,0 ) ( 11,0 ) ( 10,0 ) ( 9,26.667 ) ( 8,31.111 ) ( 7,35.556 ) ( 6,40 ) ( 5,44.444 ) ( 4,48.889 ) ( 3,53.333 ) ( 2,57.778 ) ( 1,62.222 )
};
\addplot[rgb color={0, 43, 170}, very thick,  mark size=2pt, mark=., fill opacity=0.5, draw opacity=0.75] coordinates {
( 1,400 ) ( 2,200 ) ( 3,133.333 ) ( 4,100 ) ( 5,80 ) ( 6,66.667 ) ( 7,57.143 ) ( 8,50 ) ( 9,44.444 ) ( 10,40 ) ( 11,36.364 ) ( 12,33.333 ) ( 13,30.769 ) ( 14,28.571 ) ( 15,26.667 ) ( 16,25 ) ( 17,23.529 ) ( 18,22.222 ) ( 19,21.053 ) ( 20,20 ) ( 20,0 ) ( 19,0 ) ( 18,0 ) ( 17,0 ) ( 16,0 ) ( 15,0 ) ( 14,0 ) ( 13,0 ) ( 12,0 ) ( 11,0 ) ( 10,0 ) ( 9,0 ) ( 8,32.5 ) ( 7,37.5 ) ( 6,42.5 ) ( 5,47.5 ) ( 4,52.5 ) ( 3,57.5 ) ( 2,62.5 ) ( 1,67.5 )
};
\addplot[rgb color={0, 0, 255}, very thick,  mark size=2pt, mark=., fill opacity=0.5, draw opacity=0.75] coordinates {
( 1,400 ) ( 2,200 ) ( 3,133.333 ) ( 4,100 ) ( 5,80 ) ( 6,66.667 ) ( 7,57.143 ) ( 8,50 ) ( 9,44.444 ) ( 10,40 ) ( 11,36.364 ) ( 12,33.333 ) ( 13,30.769 ) ( 14,28.571 ) ( 15,26.667 ) ( 16,25 ) ( 17,23.529 ) ( 18,22.222 ) ( 19,21.053 ) ( 20,20 ) ( 20,0 ) ( 19,0 ) ( 18,0 ) ( 17,0 ) ( 16,0 ) ( 15,0 ) ( 14,0 ) ( 13,0 ) ( 12,0 ) ( 11,0 ) ( 10,0 ) ( 9,0 ) ( 8,0 ) ( 7,40 ) ( 6,45.714 ) ( 5,51.429 ) ( 4,57.143 ) ( 3,62.857 ) ( 2,68.571 ) ( 1,74.286 )
};
\addplot[rgb color={25, 0, 213}, very thick,  mark size=2pt, mark=., fill opacity=0.5, draw opacity=0.75] coordinates {
( 1,400 ) ( 2,200 ) ( 3,133.333 ) ( 4,100 ) ( 5,80 ) ( 6,66.667 ) ( 7,57.143 ) ( 8,50 ) ( 9,44.444 ) ( 10,40 ) ( 11,36.364 ) ( 12,33.333 ) ( 13,30.769 ) ( 14,28.571 ) ( 15,26.667 ) ( 16,25 ) ( 17,23.529 ) ( 18,22.222 ) ( 19,21.053 ) ( 20,20 ) ( 20,0 ) ( 19,0 ) ( 18,0 ) ( 17,0 ) ( 16,0 ) ( 15,0 ) ( 14,0 ) ( 13,0 ) ( 12,0 ) ( 11,0 ) ( 10,0 ) ( 9,0 ) ( 8,0 ) ( 7,0 ) ( 6,50 ) ( 5,56.667 ) ( 4,63.333 ) ( 3,70 ) ( 2,76.667 ) ( 1,83.333 )
};
\addplot[rgb color={50, 0, 172}, very thick,  mark size=2pt, mark=., fill opacity=0.5, draw opacity=0.75] coordinates {
( 1,400 ) ( 2,200 ) ( 3,133.333 ) ( 4,100 ) ( 5,80 ) ( 6,66.667 ) ( 7,57.143 ) ( 8,50 ) ( 9,44.444 ) ( 10,40 ) ( 11,36.364 ) ( 12,33.333 ) ( 13,30.769 ) ( 14,28.571 ) ( 15,26.667 ) ( 16,25 ) ( 17,23.529 ) ( 18,22.222 ) ( 19,21.053 ) ( 20,20 ) ( 20,0 ) ( 19,0 ) ( 18,0 ) ( 17,0 ) ( 16,0 ) ( 15,0 ) ( 14,0 ) ( 13,0 ) ( 12,0 ) ( 11,0 ) ( 10,0 ) ( 9,0 ) ( 8,0 ) ( 7,0 ) ( 6,0 ) ( 5,64 ) ( 4,72 ) ( 3,80 ) ( 2,88 ) ( 1,96 )
};
\addplot[rgb color={75, 0, 130}, very thick,  mark size=2pt, mark=., fill opacity=0.5, draw opacity=0.75] coordinates {
( 1,400 ) ( 2,200 ) ( 3,133.333 ) ( 4,100 ) ( 5,80 ) ( 6,66.667 ) ( 7,57.143 ) ( 8,50 ) ( 9,44.444 ) ( 10,40 ) ( 11,36.364 ) ( 12,33.333 ) ( 13,30.769 ) ( 14,28.571 ) ( 15,26.667 ) ( 16,25 ) ( 17,23.529 ) ( 18,22.222 ) ( 19,21.053 ) ( 20,20 ) ( 20,0 ) ( 19,0 ) ( 18,0 ) ( 17,0 ) ( 16,0 ) ( 15,0 ) ( 14,0 ) ( 13,0 ) ( 12,0 ) ( 11,0 ) ( 10,0 ) ( 9,0 ) ( 8,0 ) ( 7,0 ) ( 6,0 ) ( 5,0 ) ( 4,85 ) ( 3,95 ) ( 2,105 ) ( 1,115 )
};
\addplot[rgb color={129, 43, 166}, very thick,  mark size=2pt, mark=., fill opacity=0.5, draw opacity=0.75] coordinates {
( 1,400 ) ( 2,200 ) ( 3,133.333 ) ( 4,100 ) ( 5,80 ) ( 6,66.667 ) ( 7,57.143 ) ( 8,50 ) ( 9,44.444 ) ( 10,40 ) ( 11,36.364 ) ( 12,33.333 ) ( 13,30.769 ) ( 14,28.571 ) ( 15,26.667 ) ( 16,25 ) ( 17,23.529 ) ( 18,22.222 ) ( 19,21.053 ) ( 20,20 ) ( 20,0 ) ( 19,0 ) ( 18,0 ) ( 17,0 ) ( 16,0 ) ( 15,0 ) ( 14,0 ) ( 13,0 ) ( 12,0 ) ( 11,0 ) ( 10,0 ) ( 9,0 ) ( 8,0 ) ( 7,0 ) ( 6,0 ) ( 5,0 ) ( 4,0 ) ( 3,120 ) ( 2,133.333 ) ( 1,146.667 )
};
\addplot[rgb color={184, 87, 202}, very thick,  mark size=2pt, mark=., fill opacity=0.5, draw opacity=0.75] coordinates {
( 1,400 ) ( 2,200 ) ( 3,133.333 ) ( 4,100 ) ( 5,80 ) ( 6,66.667 ) ( 7,57.143 ) ( 8,50 ) ( 9,44.444 ) ( 10,40 ) ( 11,36.364 ) ( 12,33.333 ) ( 13,30.769 ) ( 14,28.571 ) ( 15,26.667 ) ( 16,25 ) ( 17,23.529 ) ( 18,22.222 ) ( 19,21.053 ) ( 20,20 ) ( 20,0 ) ( 19,0 ) ( 18,0 ) ( 17,0 ) ( 16,0 ) ( 15,0 ) ( 14,0 ) ( 13,0 ) ( 12,0 ) ( 11,0 ) ( 10,0 ) ( 9,0 ) ( 8,0 ) ( 7,0 ) ( 6,0 ) ( 5,0 ) ( 4,0 ) ( 3,0 ) ( 2,190 ) ( 1,210 )
};
\addplot[color=black, very thick,  mark size=2pt, mark=.] coordinates {
( 1,400 ) ( 2,200 ) ( 3,133.333 ) ( 4,100 ) ( 5,80 ) ( 6,66.667 ) ( 7,57.143 ) ( 8,50 ) ( 9,44.444 ) ( 10,40 ) ( 11,36.364 ) ( 12,33.333 ) ( 13,30.769 ) ( 14,28.571 ) ( 15,26.667 ) ( 16,25 ) ( 17,23.529 ) ( 18,22.222 ) ( 19,21.053 ) ( 20,20 )
};
\end{axis}
\end{tikzpicture}
 \caption{\revise{Regions of ordered e-value sequences, for $m=20$ and $\alpha=0.05$.
The $x$-axis gives the rank and the $y$-axis gives the ordered e-values $\evalue_{(1)}\geq \cdots \geq \evalue_{(m)}$.
The black curve is the eBH rejection boundary, i.e., an ordered sequence that lies entirely below the black curve leads eBH to reject no hypotheses.
The colored regions show ordered sequences for which $\CeBH$ rejects at least $k$ hypotheses while eBH rejects none; the case $k=20$ corresponds to \Cref{ex: eBH+}.
These regions were constructed by greedily choosing the smallest $\evalue_{(m)}, \ldots, \evalue_{(1)}$, in that order, to achieve the desired mean-consistency.}}
    \label{fig:example1}
\end{figure}

The $\CeBH$ procedure and mean-consistency are qualitatively different from the eBH procedure and self-consistency. We will enumerate some unique properties of discovery sets satisfying mean-consistency and the discoveries made by the $\CeBH$ procedure.
First, rejection of a particular set $\mathbf{R}$ depends not only on the e-values corresponding to hypotheses in $\mathbf{R}$, but also on the e-values corresponding to other hypotheses.
\begin{example}[Dependence of non-rejected e-values]
Consider a case where $m=2$, and $\mathbf{e}_1=9/(5\alpha)$ and $\mathbf{e}_2=0$. In this case, nothing is rejected. However, if we increase $\mathbf{e}_2$ to $1/(5\alpha)$, we obtain $\rejcol_\alpha^\CeBH = \{\emptyset, \{1\}\}$. Increasing $\mathbf{e}_2$, apparently, facilitates the rejection of a set of hypotheses that does not include $H_2$.
\end{example}

Secondly, $\rejcol_\alpha^\CeBH$ may reject hypotheses for which the corresponding e-value is less than $\alpha^{-1}$.
\begin{example}[Rejecting e-values less than $1 / \alpha$]
To see an example, consider $m=2$, $\mathbf{e}_1=3/(2\alpha)$ and $\mathbf{e}_2=1/(2\alpha)$. It is easily checked that $\{1,2\} \in \rejcol_\alpha^\CeBH$, implying that $H_2$ can be rejected with FDR control at $\alpha$, although $\mathbf{e}_2<\alpha^{-1}$.
\end{example}
Translated to p-values, these properties of the procedure of \eqref{eq: proc mean e} are shared by adaptive FDR control procedures that plug in an estimate $\boldsymbol{\hat\pi}_0$ of $\pi_{0,\mathrm{P}} = |N_\mathrm{P}|/m$ into a procedure controlling FDR at $\pi_{0,\mathrm{P}}\alpha$
\citep{storey2004strong, benjamini2006adaptive, blanchard2009adaptive}. For such procedures, rejection of a hypothesis
may also depend on p-values of other hypotheses through $\boldsymbol{\hat\pi}_0$, and the procedure may reject hypotheses with p-values up to $\alpha/\boldsymbol{\hat\pi}_0 > \alpha$. The $\CeBH$ procedure can therefore be seen as an adaptive procedure, even though it was not explicitly constructed using an estimate of $\pi_{0,\mathrm{P}}$.

Lastly, $\textCeBH$ does not necessarily output a single mean-consistent discovery set that has the largest cardinality (containing all other mean-consistent sets as a subset), a property which holds for self-consistent sets. There can be multiple discovery sets of the largest cardinality in $\rejcol_\alpha^\CeBH$.
\begin{example}[Multiple largest mean-consistent discovery sets]
    Let $m=3$ and $\evalue_1 = 2/\alpha$, and $\evalue_2 = 1/(2\alpha)$, and $\evalue_3 = 1 /(2\alpha)$. Then, both $\rejset^\CeBH_\alpha = \{1, 2\}$ and $\{1, 3\}$ are in $\rejcol^\CeBH_\alpha$, but $\{1,2,3\}$ is not since $\evalue_{\{2, 3\}} = 1 / (2\alpha) < 2 / (3\alpha)$. \label{ex: non-unique}
\end{example}

\subsection{Compound e-values and the universality of eBH}
\citet{ignatiadis_asymptotic_compound_2025} showed that eBH, when applied to compound e-values, is a universal procedure for FDR control. They showed that every FDR controlling procedure can be written as the eBH procedure applied to some set of compound e-values. \emph{Compound e-values} $(\widetilde{\evalue}_1, \dots, \widetilde{\evalue}_m)$ are nonnegative random variables that satisfy the following condition
\begin{equation}\label{eq:compound-def}
    \sum_{i \in N_\pdist}\expect_\pdist[\widetilde{\evalue}_i] \leq m \text{ for all }\pdist \in M.
\end{equation}
In words, no matter which hypotheses constitute the true set of null hypotheses $N_\pdist$ for a distribution $\pdist$, the expected sum of the corresponding compound e-values cannot exceed $m$.

\revise{Our results are complementary to the above universality result on eBH. Instead of reducing every FDR procedure to using eBH with compound e-values, we reduce it to the $\textCeBH$ procedure with an e-value for each intersection hypothesis. In this section we describe two different conversions, and we keep the direction of each conversion explicit. First, starting from compound e-values, we construct an e-collection for all intersection hypotheses. Second, starting from an e-collection and a rejected set produced by e-Closure, we construct compound e-values.}
\revise{We first start with compound e-values $(\widetilde{\evalue}_1,\ldots,\widetilde{\evalue}_m)$ satisfying \eqref{eq:compound-def}. From these given compound e-values, define for any $S \subseteq [m]$}
\begin{align}
    \evalue_S = \frac{1}{m}\sum_{i \in S}\widetilde{\evalue}_i \label{eq:c-int-evalue}
\end{align}
\revise{\begin{proposition}    
Let $(\widetilde{\evalue}_1, \dots, \widetilde{\evalue}_m)$ be compound e-values. Then $\evalue_S$, as defined in \eqref{eq:c-int-evalue}, is a valid e-value for $H_S$ for each $S \subseteq [m]$.
\end{proposition}
This is true via the definition of compound e-values in \eqref{eq:compound-def}. The factor $1/m$ is important: this is not the ordinary average over $S$, which would use $1/|S|$, but the scaling required by the compound-e-value condition. Thus, \eqref{eq:c-int-evalue} constructs an e-collection $\mathbf{E}=(\evalue_S)_{S\subseteq [m]}$ from the compound e-values. Consequently, we can define a corresponding \emph{compound mean-consistent} set $\widetilde{\rejcol}^\CeBH_\alpha \coloneqq \Rcal^\FDR_\alpha(\mathbf{E})$.}
\revise{Define the \emph{compound $\textCeBH$ procedure} w.r.t.\ $\widetilde{\rejcol}^\CeBH_\alpha$ as}
\[
\revise{\widetilde{\rejset}^\CeBH_\alpha} \coloneqq \rejset^{(\revise{\widetilde{\mathbf{r}}_\alpha^\CeBH})},
\qquad
\revise{\widetilde{\mathbf{r}}_\alpha^\CeBH} \coloneqq \max\ \{r \in [m]: \rejset^{(r)} \in \revise{\widetilde{\rejcol}^\CeBH_\alpha}\},
\] \revise{where $\rejset^{(r)}$ here refers to the set of the $r$ largest compound e-values.}
\revise{Similarly, we can define the \emph{compound self-consistent} set as $\widetilde{\rejcol}^{\SC}_\alpha \coloneqq \{R \in 2^{[m]}: \widetilde{\evalue}_i \geq m / (\alpha|R|)\text{ for all }i \in R\}$
and the \emph{compound eBH procedure} as outputting $\widetilde{\rejset}^\eBH_\alpha$, the largest set in $\widetilde{\rejcol}^\SC_\alpha$.}
\begin{theorem}\label{thm: compound ebh equal eclosure}
    We have that $\revise{\widetilde{\rejcol}^\SC_\alpha = \widetilde{\rejcol}^\CeBH_\alpha}$ and consequently $\revise{\widetilde{\rejset}^\eBH_\alpha = \widetilde{\rejset}^\CeBH_\alpha}$.
\end{theorem}

\revise{Thus, when we start from compound e-values and construct the e-collection by \eqref{eq:c-int-evalue}, eBH and $\CeBH$ are equivalent. We now switch directions. Suppose instead that we start from an e-collection $\mathbf{E}=(\evalue_S)_{S\subseteq[m]}$ and the corresponding discovery set $\rejset \in \Rcal^\FDR_\alpha(\mathbf{E})$ produced by the e-Closure Principle. From this given e-collection and rejected set, we can construct compound e-values. For each $i \in [m]$, the universal compound e-value derived from $\rejset$ \citep{ignatiadis_asymptotic_compound_2025} is}
\begin{align}
    \widetilde{\evalue}_i = \frac{m}{\alpha(|\rejset|\vee 1)} \cdot \mathbf{1}\{i \in \rejset\}. \label{eq: universal compound evalue}
\end{align}
\revise{More generally, a sufficient condition for constructing compound e-values from an e-collection is that the constructed variables satisfy the following inequality with respect to the given e-collection:}
\begin{align}
   \sum_{i \in S} \widetilde{\evalue}_i \leq m \cdot \evalue_S \text{ for all } S \in 2^{[m]} \label{eq: intersection bound}
\end{align}
\begin{proposition}\label{prop: compound e-values from e-collections}
    \revise{Let $\mathbf{E}=(\evalue_S)_{S\subseteq[m]}$ be an e-collection. Any nonnegative random variables $(\widetilde{\evalue}_1, \ldots, \widetilde{\evalue}_m)$ that satisfy \eqref{eq: intersection bound} almost surely are compound e-values. Further, if $\rejset \in \Rcal^\FDR_\alpha(\mathbf{E})$, then the universal compound e-values in \eqref{eq: universal compound evalue} satisfy \eqref{eq: intersection bound}.}
\end{proposition}

\section{FDR control via merging p-values} \label{sec: p-combining}
Many FDR control procedures start from p-values $\mathbf{p}_1, \ldots, \mathbf{p}_m$ for hypotheses $H_1, \ldots, H_m$. Let $\pvalue_{(i)}$ denote the $i$th smallest p-value for $i \in [m]$, with ties being broken by choosing the hypothesis with the smallest index.
There are several famous procedures, the choice of which depends on the assumptions that we are willing to make on the joint distribution of the p-values. Here, we will explore three such procedures. The first is BY \citep{benjamini_control_false_2001}, which is valid for any distribution of the p-values. The second is the procedure of \citet{su2018fdr} based on the FDR Linking Theorem. The Su method is valid under the PRDN assumption, a weaker variant of the PRDS assumption underlying BH. The third is BH itself. We will place these three methods in the context of the e-Closure Principle and show that the first two can be uniformly improved.

Since the e-Closure Principle depends on e-values, we will need to convert the input p-values to e-values. We use p-to-e calibrators for this \citep{shafer2011test}. A function $e(\mathbf{p})$ is a p-to-e calibrator if $e(\mathbf{p})$ is an e-value whenever $\mathbf{p}$ is a p-value, i.e., $\mathbb{E}_{\mathrm{P}}[e(\mathbf{p})]\leq 1$  whenever $\mathrm{P}(\mathbf{p}\leq t) \leq t$ for all $0 \leq t \leq 1$. One straightforward way to construct an FDR control procedure would be to calibrate $\mathbf{p}_1,  \ldots, \mathbf{p}_m$ to e-values and apply $\CeBH$. However, in general, this does not turn out to be the most powerful approach.

\subsection{\texorpdfstring{$\textcBY$}{closed BY}: FDR control under general dependence} \label{sec:BY}
When working with p-values, let $\rejset^{(i)}$ denote the set of indices corresponding to the $i$ smallest p-values and let $\rejset^{(0)} \coloneqq \emptyset$. \cite{benjamini_control_false_2001} proved that under general dependence, the following random set controls FDR.
$$\rejset_\alpha^{\BY} \coloneqq \rejset^{(\mathbf{r}_\alpha^\BY)}, \text{ where }\mathbf{r}_\alpha^\BY \coloneqq \max\{r \in [m]: \pvalue_{(r)} \leq \alpha r / (mh_m) \},$$
where $\mathbf{r}_\alpha^\BY$ is taken as 0 if the maximum does not exist.
Here, $h_m = \sum_{i=1}^m 1/i$ is the $m$th harmonic number.

\citet[Proposition 3]{xu_post-selection_inference_2022} formulated the following p-to-e calibrator for every $k \in \naturals$ and $\alpha \in (0, 1]$.
\begin{align}
e_k(\mathbf{p}) \coloneqq \frac{k 1\{ h_k \mathbf{p} \leq \alpha  \}   }{ \alpha  (\lceil k h_k  \mathbf{p} / \alpha  \rceil \vee 1)} \label{eq: by calibrator}
\end{align} is a p-to-e calibrator for all $\alpha \in (0, 1]$ and all $k \in \mathbb{N}$.
They showed that the eBH procedure applied to $e_m(\pvalue_1), \dots, e_m(\pvalue_m)$ is equivalent to the BY procedure. We will use this same calibrator, though in a different way, to define an e-collection.

\revise{For every nonempty $S$,} $\evalue_S$ is the average of the e-values obtained by applying $e$ to the p-values of hypotheses in $S$; however, the choice of p-to-e calibrator will now depend on $|S|$ (instead of $m$) as well as on $\alpha$.
\begin{lemma} Under $H_S$, the following is an e-value:
\begin{equation} \label{eq: e for BY}
\mathbf{e}_S = \sum_{i \in S}  \frac{ 1\{ h_{|S|} \mathbf{p}_i \leq \alpha  \}   }{ \alpha  (\lceil |S| h_{|S|} \mathbf{p}_i / \alpha  \rceil \vee 1) }
\end{equation}
\end{lemma}
\begin{proof}
We take $k=|S|$, apply the p-to-e calibrator to $\mathbf{p}_i$ for each $i \in S$, and note that the average of the resulting e-values is again an e-value.
\end{proof}
Next, we can build an FDR control method by plugging in $\mathbf{E} = (\mathbf{e}_S)_{S \in 2^{[m]}}$, for the e-values we have just defined in (\ref{eq: e for BY}), into the general e-Closure Procedure (\ref{eq: e-closure}).
Formally, we can define the simultaneous $\cBY$ (closed BY) procedure as the simultaneous discovery set resulting from the e-Closure Principle applied to \eqref{eq: e for BY}, i.e., $\rejcol_\alpha^\cBY \coloneqq \Rcal_\alpha^\FDR(\ecol)$. The $\cBY$ procedure is then one of the largest discovery sets in $\rejcol_\alpha^\cBY$, i.e.,
\begin{align}
        \rejset^\cBY_\alpha \coloneqq \rejset^{(\mathbf{r}_\alpha^\cBY)}, \text{ where }\mathbf{r}_\alpha^\cBY &\coloneqq \max\bigl(\{r \in [m]: \rejset^{(r)} \in \rejcol^\cBY_\alpha\} \cup \{0\}\bigr).
\end{align}
The simultaneous $\cBY$ procedure controls FDR under arbitrary dependence among the p-values as a result of \Cref{thm: fdr-closure}, and uniformly improves BY in both simultaneity and power, as stated in the following theorem.
\begin{theorem}
    \revise{Under arbitrary dependence among the p-values, $\rejcol_\alpha^\cBY$ controls FDR at level $\alpha$. If $m>1$, $\rejcol_\alpha^\cBY$ is a uniform improvement in simultaneity and power over $\rejset_\alpha^\BY$.}
\end{theorem}

First, we define the notion of a self-consistent discovery set at level $\alpha$ for p-values \citep{blanchard_two_simple_2008,su2018fdr}. $R$ is a self-consistent discovery set at level $\alpha$ if $\pvalue_i \leq \alpha |R| / m$ for all $i \in R$.
\begin{proof}
Let $R$ be a non-empty self-consistent discovery set at level $\alpha / h_m$. Note that $\rejset^\BY_\alpha$ is the largest self-consistent discovery set at level $\alpha / h_m$.
\begin{align}
\alpha\mathbf{e}_S
&\geq \sum_{i\in S \cap R} \frac{ 1\{ h_{|S|} \mathbf{p}_i \leq \alpha  \}   }{ \lceil |S| h_{|S|} \mathbf{p}_i / \alpha  \rceil \vee 1 }
= \sum_{i\in S \cap R} \frac{1} {\lceil |S| h_{|S|} \mathbf{p}_i / \alpha  \rceil \vee 1 }
\geq \sum_{i\in S \cap R} \frac{1} {\lceil |S| h_{|S|} |R| / m h_m  \rceil \vee 1}
\geq \frac{|S \cap R|} {|R| }.
\end{align}
This shows that $R \in \rejcol_\alpha^\cBY$. The first inequality follows from the definition $\evalue_S$ from \eqref{eq: e for BY} and the second inequality is by $|S|h_{|S|} \leq mh_m$ as $S \subseteq [m]$. The same is trivially true if $\rejset^\BY_\alpha=\emptyset$.

\revise{For strictness for every $m>1$, take ordered p-values satisfying}
\[
\revise{\pvalue_{(1)}=\cdots=\pvalue_{(m-1)}=0,
\qquad \frac{\alpha}{h_m}<\pvalue_{(m)}\leq\alpha.
}
\]
\revise{Then BY rejects the first $m-1$ hypotheses, and $\textcBY$ rejects all $m$ hypotheses: for every nonempty $S$, each zero p-value in $S$ contributes $1$ to $\alpha\evalue_S$, while if $S$ contains only the final index, its contribution is also $1$ because $\pvalue_{(m)}\leq\alpha$. Hence $\alpha\evalue_S\geq |S|/m=\FDP_S([m])$ for every $S$.}

The following example demonstrates the power improvement of $\cBY$.
\end{proof}

We now construct an example where BY fails to reject any hypothesis, while $\textcBY$ rejects all.

\begin{example}\label{ex: closed by example}
Let $m \geq 4$, \revise{and relabel hypotheses in increasing p-value order.} Then $\rejset_\alpha^\BY$ is empty, but $\rejset_\alpha^\cBY = [m]$ when the following event holds
\[
    \frac{i\alpha}{m h_m} < \revise{\mathbf{p}_{(i)}} \leq \frac{(i+1)\alpha}{m h_m} \quad \text{for } i = 1, \ldots, m-1, \quad \text{and} \quad \frac{\alpha}{h_{m - 1}} \geq \revise{\mathbf{p}_{(m)}} > \frac{\alpha}{h_m}.
\]
\end{example}
We defer a proof of this to \Cref{sec: closed by example proof}. 
The improvement of $\textcBY$ upon BY is similar in spirit to the improvement of $\CeBH$ relative to eBH.
However, $\textcBY$ will never reject hypotheses $H_i$ with $\mathbf{p}_i > \alpha$.
\ifarver{
We illustrate the aforementioned example for $m = 20$ in \Cref{fig:BY}.
\begin{figure}[!ht]
\centering
\begin{tikzpicture}[scale=.7]
\begin{axis}[
	xmin = 1,
	xmax = 20,
	xtick = {5,10,15,20},
	ymin = 0,
	ymax = 0.05,
	ytick = {0,0.01,0.02,0.03,0.04,0.05},
	yticklabels = {0,0.01,0.02,0.03,0.04,0.05},
	 ytick scale label code/.code={},
	ylabel=$p$-value,
	xlabel=rank,
        height=10cm,
	width=10cm,
		legend style={legend columns=1, at={(0.275,0.975)}}
]
 \addplot[color=blue, very thick,  mark size=2pt, mark=*, fill=blue, fill opacity=0.5, draw opacity=0.75] coordinates {
( 1,7e-04 ) ( 2,0.0014 ) ( 3,0.0021 ) ( 4,0.0028 ) ( 5,0.0035 ) ( 6,0.0042 ) ( 7,0.0049 ) ( 8,0.0056 ) ( 9,0.0063 ) ( 10,0.0069 ) ( 11,0.0076 ) ( 12,0.0083 ) ( 13,0.009 ) ( 14,0.0097 ) ( 15,0.0104 ) ( 16,0.0111 ) ( 17,0.0118 ) ( 18,0.0125 ) ( 19,0.0132 ) ( 20,0.0139 ) ( 20,1 ) ( 19,1 ) ( 18,1 ) ( 17,1 ) ( 16,1 ) ( 15,1 ) ( 14,1 ) ( 13,1 ) ( 12,1 ) ( 11,1 ) ( 10,1 ) ( 9,1 ) ( 8,1 ) ( 7,0.0077 ) ( 6,0.0063 ) ( 5,0.0053 ) ( 4,0.0044 ) ( 3,0.0044 ) ( 2,0.0038 ) ( 1,0.003 )
};
 \addlegendentry{$k=7$}
\addplot[color=red, very thick,  mark size=2pt, mark=*, fill=red, fill opacity=0.5, draw opacity=0.75] coordinates {
( 1,7e-04 ) ( 2,0.0014 ) ( 3,0.0021 ) ( 4,0.0028 ) ( 5,0.0035 ) ( 6,0.0042 ) ( 7,0.0049 ) ( 8,0.0056 ) ( 9,0.0063 ) ( 10,0.0069 ) ( 11,0.0076 ) ( 12,0.0083 ) ( 13,0.009 ) ( 14,0.0097 ) ( 15,0.0104 ) ( 16,0.0111 ) ( 17,0.0118 ) ( 18,0.0125 ) ( 19,0.0132 ) ( 20,0.0139 ) ( 20,1 ) ( 19,0.0333 ) ( 18,0.0273 ) ( 17,0.024 ) ( 16,0.0208 ) ( 15,0.0129 ) ( 14,0.0129 ) ( 13,0.0129 ) ( 12,0.0129 ) ( 11,0.0129 ) ( 10,0.0129 ) ( 9,0.0129 ) ( 8,0.0129 ) ( 7,0.0129 ) ( 6,0.0129 ) ( 5,0.0123 ) ( 4,0.0057 ) ( 3,0.0057 ) ( 2,0.0057 ) ( 1,0.0057 )
};
 \addlegendentry{$k=19$}
 \addplot[color=black, very thick,  mark size=2pt, mark=*] coordinates {
( 1,7e-04 ) ( 2,0.0014 ) ( 3,0.0021 ) ( 4,0.0028 ) ( 5,0.0035 ) ( 6,0.0042 ) ( 7,0.0049 ) ( 8,0.0056 ) ( 9,0.0063 ) ( 10,0.0069 ) ( 11,0.0076 ) ( 12,0.0083 ) ( 13,0.009 ) ( 14,0.0097 ) ( 15,0.0104 ) ( 16,0.0111 ) ( 17,0.0118 ) ( 18,0.0125 ) ( 19,0.0132 ) ( 20,0.0139 )
};
\end{axis}
\end{tikzpicture}~
\begin{tikzpicture}[scale=.7, define rgb/.code={\definecolor{mycolor}{RGB}{#1}},
                    rgb color/.style={define rgb={#1},mycolor}]
\begin{axis}[
xmin = 1,
	xmax = 20,
	xtick = {5,10,15,20},
	ymin = 0,
	ymax = 0.05,
	ytick = {0,0.01,0.02,0.03,0.04,0.05},
	yticklabels = {0,0.01,0.02,0.03,0.04,0.05},
	 ytick scale label code/.code={},
	ylabel=,
	xlabel=rank,
        height=10cm,
	width=10cm
]
\addplot[rgb color={255, 0, 0}, very thick,  mark size=2pt, mark=., fill opacity=0.5, draw opacity=0.75] coordinates {
( 1,7e-04 ) ( 2,0.0014 ) ( 3,0.0021 ) ( 4,0.0028 ) ( 5,0.0035 ) ( 6,0.0042 ) ( 7,0.0049 ) ( 8,0.0056 ) ( 9,0.0063 ) ( 10,0.0069 ) ( 11,0.0076 ) ( 12,0.0083 ) ( 13,0.009 ) ( 14,0.0097 ) ( 15,0.0104 ) ( 16,0.0111 ) ( 17,0.0118 ) ( 18,0.0125 ) ( 19,0.0132 ) ( 20,0.0139 ) ( 20,1 ) ( 19,0.0333 ) ( 18,0.0273 ) ( 17,0.024 ) ( 16,0.0208 ) ( 15,0.0129 ) ( 14,0.0129 ) ( 13,0.0129 ) ( 12,0.0129 ) ( 11,0.0129 ) ( 10,0.0129 ) ( 9,0.0129 ) ( 8,0.0129 ) ( 7,0.0129 ) ( 6,0.0129 ) ( 5,0.0123 ) ( 4,0.0057 ) ( 3,0.0057 ) ( 2,0.0057 ) ( 1,0.0057 )
};
\addplot[rgb color={255, 195, 0}, very thick,  mark size=2pt, mark=., fill opacity=0.5, draw opacity=0.75] coordinates {
( 1,7e-04 ) ( 2,0.0014 ) ( 3,0.0021 ) ( 4,0.0028 ) ( 5,0.0035 ) ( 6,0.0042 ) ( 7,0.0049 ) ( 8,0.0056 ) ( 9,0.0063 ) ( 10,0.0069 ) ( 11,0.0076 ) ( 12,0.0083 ) ( 13,0.009 ) ( 14,0.0097 ) ( 15,0.0104 ) ( 16,0.0111 ) ( 17,0.0118 ) ( 18,0.0125 ) ( 19,0.0132 ) ( 20,0.0139 ) ( 20,1 ) ( 19,1 ) ( 18,1 ) ( 17,1 ) ( 16,1 ) ( 15,0.0204 ) ( 14,0.0193 ) ( 13,0.0115 ) ( 12,0.0115 ) ( 11,0.0115 ) ( 10,0.0115 ) ( 9,0.0115 ) ( 8,0.0115 ) ( 7,0.0115 ) ( 6,0.0115 ) ( 5,0.0078 ) ( 4,0.0067 ) ( 3,0.0065 ) ( 2,0.0062 ) ( 1,0.0051 )
};
\addplot[rgb color={255, 255, 0}, very thick,  mark size=2pt, mark=., fill opacity=0.5, draw opacity=0.75] coordinates {
( 1,7e-04 ) ( 2,0.0014 ) ( 3,0.0021 ) ( 4,0.0028 ) ( 5,0.0035 ) ( 6,0.0042 ) ( 7,0.0049 ) ( 8,0.0056 ) ( 9,0.0063 ) ( 10,0.0069 ) ( 11,0.0076 ) ( 12,0.0083 ) ( 13,0.009 ) ( 14,0.0097 ) ( 15,0.0104 ) ( 16,0.0111 ) ( 17,0.0118 ) ( 18,0.0125 ) ( 19,0.0132 ) ( 20,0.0139 ) ( 20,1 ) ( 19,1 ) ( 18,1 ) ( 17,1 ) ( 16,1 ) ( 15,1 ) ( 14,1 ) ( 13,0.0184 ) ( 12,0.0128 ) ( 11,0.0128 ) ( 10,0.0128 ) ( 9,0.0128 ) ( 8,0.0124 ) ( 7,0.009 ) ( 6,0.0084 ) ( 5,0.0068 ) ( 4,0.0068 ) ( 3,0.0052 ) ( 2,0.0047 ) ( 1,0.0047 )
};
\addplot[rgb color={170, 213, 0}, very thick,  mark size=2pt, mark=., fill opacity=0.5, draw opacity=0.75] coordinates {
( 1,7e-04 ) ( 2,0.0014 ) ( 3,0.0021 ) ( 4,0.0028 ) ( 5,0.0035 ) ( 6,0.0042 ) ( 7,0.0049 ) ( 8,0.0056 ) ( 9,0.0063 ) ( 10,0.0069 ) ( 11,0.0076 ) ( 12,0.0083 ) ( 13,0.009 ) ( 14,0.0097 ) ( 15,0.0104 ) ( 16,0.0111 ) ( 17,0.0118 ) ( 18,0.0125 ) ( 19,0.0132 ) ( 20,0.0139 ) ( 20,1 ) ( 19,1 ) ( 18,1 ) ( 17,1 ) ( 16,1 ) ( 15,1 ) ( 14,1 ) ( 13,1 ) ( 12,0.0177 ) ( 11,0.0102 ) ( 10,0.0102 ) ( 9,0.0102 ) ( 8,0.0102 ) ( 7,0.0102 ) ( 6,0.0078 ) ( 5,0.0067 ) ( 4,0.0067 ) ( 3,0.0063 ) ( 2,0.0055 ) ( 1,0.0044 )
};
\addplot[rgb color={85, 170, 0}, very thick,  mark size=2pt, mark=., fill opacity=0.5, draw opacity=0.75] coordinates {
( 1,7e-04 ) ( 2,0.0014 ) ( 3,0.0021 ) ( 4,0.0028 ) ( 5,0.0035 ) ( 6,0.0042 ) ( 7,0.0049 ) ( 8,0.0056 ) ( 9,0.0063 ) ( 10,0.0069 ) ( 11,0.0076 ) ( 12,0.0083 ) ( 13,0.009 ) ( 14,0.0097 ) ( 15,0.0104 ) ( 16,0.0111 ) ( 17,0.0118 ) ( 18,0.0125 ) ( 19,0.0132 ) ( 20,0.0139 ) ( 20,1 ) ( 19,1 ) ( 18,1 ) ( 17,1 ) ( 16,1 ) ( 15,1 ) ( 14,1 ) ( 13,1 ) ( 12,1 ) ( 11,0.0171 ) ( 10,0.0083 ) ( 9,0.0083 ) ( 8,0.0083 ) ( 7,0.0083 ) ( 6,0.0083 ) ( 5,0.0083 ) ( 4,0.0067 ) ( 3,0.0049 ) ( 2,0.0045 ) ( 1,0.0045 )
};
\addplot[rgb color={0, 128, 0}, very thick,  mark size=2pt, mark=., fill opacity=0.5, draw opacity=0.75] coordinates {
( 1,7e-04 ) ( 2,0.0014 ) ( 3,0.0021 ) ( 4,0.0028 ) ( 5,0.0035 ) ( 6,0.0042 ) ( 7,0.0049 ) ( 8,0.0056 ) ( 9,0.0063 ) ( 10,0.0069 ) ( 11,0.0076 ) ( 12,0.0083 ) ( 13,0.009 ) ( 14,0.0097 ) ( 15,0.0104 ) ( 16,0.0111 ) ( 17,0.0118 ) ( 18,0.0125 ) ( 19,0.0132 ) ( 20,0.0139 ) ( 20,1 ) ( 19,1 ) ( 18,1 ) ( 17,1 ) ( 16,1 ) ( 15,1 ) ( 14,1 ) ( 13,1 ) ( 12,1 ) ( 11,1 ) ( 10,0.0151 ) ( 9,0.0115 ) ( 8,0.0098 ) ( 7,0.0087 ) ( 6,0.0072 ) ( 5,0.0045 ) ( 4,0.0045 ) ( 3,0.0045 ) ( 2,0.0045 ) ( 1,0.0045 )
};
\addplot[rgb color={0, 85, 85}, very thick,  mark size=2pt, mark=., fill opacity=0.5, draw opacity=0.75] coordinates {
( 1,7e-04 ) ( 2,0.0014 ) ( 3,0.0021 ) ( 4,0.0028 ) ( 5,0.0035 ) ( 6,0.0042 ) ( 7,0.0049 ) ( 8,0.0056 ) ( 9,0.0063 ) ( 10,0.0069 ) ( 11,0.0076 ) ( 12,0.0083 ) ( 13,0.009 ) ( 14,0.0097 ) ( 15,0.0104 ) ( 16,0.0111 ) ( 17,0.0118 ) ( 18,0.0125 ) ( 19,0.0132 ) ( 20,0.0139 ) ( 20,1 ) ( 19,1 ) ( 18,1 ) ( 17,1 ) ( 16,1 ) ( 15,1 ) ( 14,1 ) ( 13,1 ) ( 12,1 ) ( 11,1 ) ( 10,1 ) ( 9,0.0121 ) ( 8,0.0099 ) ( 7,0.0073 ) ( 6,0.007 ) ( 5,0.0065 ) ( 4,0.005 ) ( 3,0.0046 ) ( 2,0.0046 ) ( 1,0.0038 )
};
\addplot[rgb color={0, 43, 170}, very thick,  mark size=2pt, mark=., fill opacity=0.5, draw opacity=0.75] coordinates {
( 1,7e-04 ) ( 2,0.0014 ) ( 3,0.0021 ) ( 4,0.0028 ) ( 5,0.0035 ) ( 6,0.0042 ) ( 7,0.0049 ) ( 8,0.0056 ) ( 9,0.0063 ) ( 10,0.0069 ) ( 11,0.0076 ) ( 12,0.0083 ) ( 13,0.009 ) ( 14,0.0097 ) ( 15,0.0104 ) ( 16,0.0111 ) ( 17,0.0118 ) ( 18,0.0125 ) ( 19,0.0132 ) ( 20,0.0139 ) ( 20,1 ) ( 19,1 ) ( 18,1 ) ( 17,1 ) ( 16,1 ) ( 15,1 ) ( 14,1 ) ( 13,1 ) ( 12,1 ) ( 11,1 ) ( 10,1 ) ( 9,1 ) ( 8,0.0097 ) ( 7,0.0079 ) ( 6,0.0067 ) ( 5,0.0053 ) ( 4,0.0049 ) ( 3,0.0045 ) ( 2,0.0037 ) ( 1,0.0037 )
};
\addplot[rgb color={0, 0, 255}, very thick,  mark size=2pt, mark=., fill opacity=0.5, draw opacity=0.75] coordinates {
( 1,7e-04 ) ( 2,0.0014 ) ( 3,0.0021 ) ( 4,0.0028 ) ( 5,0.0035 ) ( 6,0.0042 ) ( 7,0.0049 ) ( 8,0.0056 ) ( 9,0.0063 ) ( 10,0.0069 ) ( 11,0.0076 ) ( 12,0.0083 ) ( 13,0.009 ) ( 14,0.0097 ) ( 15,0.0104 ) ( 16,0.0111 ) ( 17,0.0118 ) ( 18,0.0125 ) ( 19,0.0132 ) ( 20,0.0139 ) ( 20,1 ) ( 19,1 ) ( 18,1 ) ( 17,1 ) ( 16,1 ) ( 15,1 ) ( 14,1 ) ( 13,1 ) ( 12,1 ) ( 11,1 ) ( 10,1 ) ( 9,1 ) ( 8,1 ) ( 7,0.0077 ) ( 6,0.0063 ) ( 5,0.0053 ) ( 4,0.0044 ) ( 3,0.0044 ) ( 2,0.0038 ) ( 1,0.003 )
};
\addplot[rgb color={25, 0, 213}, very thick,  mark size=2pt, mark=., fill opacity=0.5, draw opacity=0.75] coordinates {
( 1,7e-04 ) ( 2,0.0014 ) ( 3,0.0021 ) ( 4,0.0028 ) ( 5,0.0035 ) ( 6,0.0042 ) ( 7,0.0049 ) ( 8,0.0056 ) ( 9,0.0063 ) ( 10,0.0069 ) ( 11,0.0076 ) ( 12,0.0083 ) ( 13,0.009 ) ( 14,0.0097 ) ( 15,0.0104 ) ( 16,0.0111 ) ( 17,0.0118 ) ( 18,0.0125 ) ( 19,0.0132 ) ( 20,0.0139 ) ( 20,1 ) ( 19,1 ) ( 18,1 ) ( 17,1 ) ( 16,1 ) ( 15,1 ) ( 14,1 ) ( 13,1 ) ( 12,1 ) ( 11,1 ) ( 10,1 ) ( 9,1 ) ( 8,1 ) ( 7,1 ) ( 6,0.006 ) ( 5,0.0049 ) ( 4,0.0041 ) ( 3,0.0034 ) ( 2,0.0034 ) ( 1,0.0034 )
};
\addplot[rgb color={50, 0, 172}, very thick,  mark size=2pt, mark=., fill opacity=0.5, draw opacity=0.75] coordinates {
( 1,7e-04 ) ( 2,0.0014 ) ( 3,0.0021 ) ( 4,0.0028 ) ( 5,0.0035 ) ( 6,0.0042 ) ( 7,0.0049 ) ( 8,0.0056 ) ( 9,0.0063 ) ( 10,0.0069 ) ( 11,0.0076 ) ( 12,0.0083 ) ( 13,0.009 ) ( 14,0.0097 ) ( 15,0.0104 ) ( 16,0.0111 ) ( 17,0.0118 ) ( 18,0.0125 ) ( 19,0.0132 ) ( 20,0.0139 ) ( 20,1 ) ( 19,1 ) ( 18,1 ) ( 17,1 ) ( 16,1 ) ( 15,1 ) ( 14,1 ) ( 13,1 ) ( 12,1 ) ( 11,1 ) ( 10,1 ) ( 9,1 ) ( 8,1 ) ( 7,1 ) ( 6,1 ) ( 5,0.0046 ) ( 4,0.0037 ) ( 3,0.0034 ) ( 2,0.0029 ) ( 1,0.0024 )
};
\addplot[rgb color={75, 0, 130}, very thick,  mark size=2pt, mark=., fill opacity=0.5, draw opacity=0.75] coordinates {
( 1,7e-04 ) ( 2,0.0014 ) ( 3,0.0021 ) ( 4,0.0028 ) ( 5,0.0035 ) ( 6,0.0042 ) ( 7,0.0049 ) ( 8,0.0056 ) ( 9,0.0063 ) ( 10,0.0069 ) ( 11,0.0076 ) ( 12,0.0083 ) ( 13,0.009 ) ( 14,0.0097 ) ( 15,0.0104 ) ( 16,0.0111 ) ( 17,0.0118 ) ( 18,0.0125 ) ( 19,0.0132 ) ( 20,0.0139 ) ( 20,1 ) ( 19,1 ) ( 18,1 ) ( 17,1 ) ( 16,1 ) ( 15,1 ) ( 14,1 ) ( 13,1 ) ( 12,1 ) ( 11,1 ) ( 10,1 ) ( 9,1 ) ( 8,1 ) ( 7,1 ) ( 6,1 ) ( 5,1 ) ( 4,0.0034 ) ( 3,0.0026 ) ( 2,0.0025 ) ( 1,0.0023 )
};
\addplot[rgb color={129, 43, 166}, very thick,  mark size=2pt, mark=., fill opacity=0.5, draw opacity=0.75] coordinates {
( 1,7e-04 ) ( 2,0.0014 ) ( 3,0.0021 ) ( 4,0.0028 ) ( 5,0.0035 ) ( 6,0.0042 ) ( 7,0.0049 ) ( 8,0.0056 ) ( 9,0.0063 ) ( 10,0.0069 ) ( 11,0.0076 ) ( 12,0.0083 ) ( 13,0.009 ) ( 14,0.0097 ) ( 15,0.0104 ) ( 16,0.0111 ) ( 17,0.0118 ) ( 18,0.0125 ) ( 19,0.0132 ) ( 20,0.0139 ) ( 20,1 ) ( 19,1 ) ( 18,1 ) ( 17,1 ) ( 16,1 ) ( 15,1 ) ( 14,1 ) ( 13,1 ) ( 12,1 ) ( 11,1 ) ( 10,1 ) ( 9,1 ) ( 8,1 ) ( 7,1 ) ( 6,1 ) ( 5,1 ) ( 4,1 ) ( 3,0.0024 ) ( 2,0.0018 ) ( 1,0.0017 )
};
\addplot[rgb color={184, 87, 202}, very thick,  mark size=2pt, mark=., fill opacity=0.5, draw opacity=0.75] coordinates {
( 1,7e-04 ) ( 2,0.0014 ) ( 3,0.0021 ) ( 4,0.0028 ) ( 5,0.0035 ) ( 6,0.0042 ) ( 7,0.0049 ) ( 8,0.0056 ) ( 9,0.0063 ) ( 10,0.0069 ) ( 11,0.0076 ) ( 12,0.0083 ) ( 13,0.009 ) ( 14,0.0097 ) ( 15,0.0104 ) ( 16,0.0111 ) ( 17,0.0118 ) ( 18,0.0125 ) ( 19,0.0132 ) ( 20,0.0139 ) ( 20,1 ) ( 19,1 ) ( 18,1 ) ( 17,1 ) ( 16,1 ) ( 15,1 ) ( 14,1 ) ( 13,1 ) ( 12,1 ) ( 11,1 ) ( 10,1 ) ( 9,1 ) ( 8,1 ) ( 7,1 ) ( 6,1 ) ( 5,1 ) ( 4,1 ) ( 3,1 ) ( 2,0.0015 ) ( 1,0.001 )
};
\addplot[color=black, very thick,  mark size=2pt, mark=.] coordinates {
( 1,7e-04 ) ( 2,0.0014 ) ( 3,0.0021 ) ( 4,0.0028 ) ( 5,0.0035 ) ( 6,0.0042 ) ( 7,0.0049 ) ( 8,0.0056 ) ( 9,0.0063 ) ( 10,0.0069 ) ( 11,0.0076 ) ( 12,0.0083 ) ( 13,0.009 ) ( 14,0.0097 ) ( 15,0.0104 ) ( 16,0.0111 ) ( 17,0.0118 ) ( 18,0.0125 ) ( 19,0.0132 ) ( 20,0.0139 )
};
\end{axis}
\end{tikzpicture}
 \caption{\revise{Regions of ordered p-value sequences, for $m=20$ and $\alpha=0.05$.
The $x$-axis gives the rank and the $y$-axis gives the ordered p-values $\pvalue_{(1)}\leq \cdots \leq \pvalue_{(m)}$.
The black curve is the BY rejection boundary, i.e., an ordered sequence that lies entirely above the black curve leads BY to reject no hypotheses.
The colored regions show ordered sequences for which $\textcBY$ rejects at least $k=2, 3, \ldots, 13,15,19$ hypotheses while BY rejects none.
These regions were constructed by greedily choosing the largest $\pvalue_{(m)}, \ldots, \pvalue_{(1)}$, in that order, to achieve the minimal desired local e-values.}} \label{fig:BY}
\end{figure}
 }{
Further examples of improvements of $\cBY$ over BY are given in \Cref{sec: p-value power improvements}.
}

\begin{remark}
\citet{blanchard_two_simple_2008} introduced a reshaped BH procedure (we call it BR in the following) for FDR control under arbitrary dependence that generalizes the BY procedure. \citet{wang_false_discovery_2022} showed that BR is a particular instance of eBH. While they used a slightly different proof technique, their claim can be shown using the p-to-e calibrator
    $$e(\mathbf{p}) \coloneqq \sum_{i=1}^m \frac{m}{\alpha i}  1\left\{\frac{\alpha \beta(i-1)}{m}<\mathbf{p}\leq \frac{\alpha \beta(i)}{m}\right\},$$
    where $\beta$ is a so-called shape function, meaning $\beta(i)=\int_0^i x \, d\nu(x)$, $i\in [m]$, $\beta(0)=0$, for some probability measure $\nu$.
    Therefore, calibrating p-values into e-values with this calibrator and plugging the calibrated e-values into $\textCeBH$ leads to a uniform improvement in simultaneity and power over the BR procedure. One could also adjust the calibrator for each $S\subseteq [m]$ according to $|S|$, as we have done for the BY procedure in \eqref{eq: e for BY}. \revise{By this adaptation, we mean using $|S|$ instead of $m$ when applying the reshaping construction to define the local e-value $\evalue_S$, since $|S|$ is the number of hypotheses relevant to $H_S$.}
\end{remark}

\begin{remark}
We can also uniformly improve the BY procedure in simultaneity and power via applying $\CeBH$ to the e-values $(e_m(\pvalue_1), \ldots, e_m(\pvalue_m))$ constructed through the p-to-e calibrator \eqref{eq: by calibrator} with $k=m$. This procedure is different from $\cBY$, and neither procedure dominates the other in general. We provide an explicit example where $\CeBH$ applied to the calibrated p-values is more powerful than $\cBY$ in \Cref{sec: cebh calibrated vs cby example}. However, in practice, we note that $\cBY$ is more powerful, as we shall see in the results of \Cref{tab:by_real_data_discoveries_full}.
\end{remark}

\subsection{$\textcSu$: FDR control under the PRDN assumption} \label{sec:Su}

\citet{su2018fdr} investigated FDR control by BH under the PRDN assumption. PRDN is a weaker variant of the PRDS assumption that is sufficient for BH to be valid \citep{benjamini_control_false_2001}. The p-values $\mathbf{p}_1\ldots, \mathbf{p}_m$ satisfy PRDS in $M$ if, for every increasing set $A \subseteq \mathbb{R}^m$, for every $\mathrm{P} \in M$, and for every $i \in N_\mathrm{P}$, we have that
\begin{equation} \label{eq: PRDS}
    \mathrm{P}((\mathbf{p}_1, \ldots, \mathbf{p}_m) \in A \mid \mathbf{p}_i \leq t) \text{ is weakly increasing in }t.
\end{equation}
PRDS, therefore, is an assumption both on the p-values of true and false hypotheses. PRDN, in contrast, is an analogous assumption only on the p-values of true hypotheses. The p-values $\mathbf{p}_1\ldots, \mathbf{p}_m$ satisfy PRDN in $M$ if, for every $\mathrm{P} \in M$, for every increasing set $A \subseteq \mathbb{R}^{|N_\mathrm{P}|}$, and for every $i \in N_\mathrm{P}$, we have that
$
\mathrm{P}((\mathbf{p}_j)_{j \in N_\mathrm{P}} \in A \mid \mathbf{p}_i \leq t) \text{ is weakly increasing in }t.
$
As argued by \cite{su2018fdr}, PRDN is a more attractive assumption than PRDS, since generally we do not want to assume anything about the distribution of p-values of false hypotheses.

Where PRDS is the assumption under which BH was proven \citep{benjamini_control_false_2001}, PRDN is the sufficient assumption of the Simes test \citep{simes1986improved, su2018fdr}:
\[
\mathbf{p}_S = \min_{1\leq i \leq |S|} \frac{|S|\pvalue_{i: S}}i,
\]
is a bona-fide p-value for testing $H_S$, where $\mathbf{p}_{i:S}$ is the $i$th smallest p-value among $\pvalue_i$, $i \in S$. \citet{su2018fdr} showed that there is a close connection between the BH procedure and the Simes test.

Moreover, \citet{su2018fdr} proved that when BH is applied at level $\alpha'$ under the assumption of PRDN, rather than PRDS, it achieves an FDR of at most $\alpha'(1 + \log(1 / \alpha'))$ rather than $\alpha'$. Solving
\[
\alpha = \alpha'(1 + \log(1 / \alpha')),
\]
we obtain $\alpha' = -\alpha/w(-\alpha/e)$,\footnote{Though we use $e$ both for p-to-e calibrators (as a function) and for $\exp(1)$ (as a constant), this should lead to no confusion.} where $w$ is the $-1$ branch of the Lambert W function. We let $\ell_\alpha \coloneqq -w(-\alpha / e)$.
This, in turn, implies that the BH procedure applied at the more conservative
level $\alpha / \ell_\alpha$, rather than $\alpha$, controls FDR under PRDN. We call this the Su procedure; let $\rejset^\Su_\alpha$ denote this specific discovery set.
Numerically, for $\alpha = 0.01, 0.05, 0.1$, we have $1 /\ell_\alpha = 0.131, 0.174, 0.205$, respectively.

In this section we will improve the Su procedure uniformly using the e-Closure Principle. First, we will define a useful p-to-e calibrator, which will allow us to calibrate the Simes p-value.

\begin{lemma}\label{lemma: su-calibrator}
    $e(x) \coloneqq (\ell_\alpha x \vee \alpha)^{-1}$ is a p-to-e calibrator for all $\alpha \in (0,1]$.
\end{lemma}
\begin{corollary}
For all $\mathrm{P} \in H_S$, the following is an e-value under PRDN.
\begin{equation} \label{eq: e Su}
    \mathbf{e}_S =  (\mathbf{p}_S \ell_\alpha \vee \alpha)^{-1}
\end{equation}
\end{corollary}

We can apply the e-Closure procedure from the e-collection formed through (\ref{eq: e Su}) to get the simultaneous discovery set $\rejcol^\cSu_\alpha \coloneqq \Rcal_\alpha^\FDR(\mathbf{E})$. One of the largest discovery sets in $\cSu$ is
\begin{align}
    \rejset^\cSu_\alpha \coloneqq \rejset^{(\mathbf{r}_\alpha^\cSu)}, \text{ where }\mathbf{r}_\alpha^\cSu &\coloneqq \max\bigl(\{r \in [m]: \rejset^{(r)} \in \rejcol^\cSu_\alpha\} \cup \{0\}\bigr).
\end{align}
\begin{theorem}\label{thm: su-improvement}
    \revise{Under PRDN, $\rejcol^\cSu_\alpha$ controls FDR at level $\alpha$. For $\alpha \in (0, 1)$ and $m>1$, $\rejcol^\cSu_\alpha$ uniformly improves $\rejset^\Su_\alpha$ in both power and simultenaity.}
\end{theorem}

\ifarver{
    We provide an example where the $\cSu$\ procedure makes multiple discoveries but the Su procedure only makes a single one, and illustrate it in \Cref{fig:Su}.
    \begin{example}\label{ex: closed su example}
Let $m \geq 3$. Then, we will have $\rejset^\Su_\alpha = [1]$, but $\rejset^\cSu_\alpha = [m]$ under the following event:   
\[
    \pvalue_1 = \frac{ \alpha}{m\ell_\alpha}, \quad \pvalue_2 = \cdots = \pvalue_m = \frac{ \alpha}{\ell_\alpha} + \delta, \quad \revise{0 < \delta \leq \min\left\{\frac{\alpha}{(m-1)\ell_\alpha},1-\frac{\alpha}{\ell_\alpha}\right\}.}
\]
\end{example}
\begin{proof}
That Su rejects $[1]$ only is immediate from the definition:  $p_i > \alpha i / (m\ell_\alpha)$ for $i = 2, \ldots, m$. $\textcSu$ rejects $[m]$, since we have $\mathbf{e}_S = \alpha^{-1}$ if $1 \in S$. If $1 \not\in S$, the constraint on $\delta$ ensures $\alpha e_S = \alpha / (\alpha + \ell_\alpha \delta) \geq (m-1)/m \geq |S|/m$.\end{proof}
    \begin{figure}[!ht]
\centering
\begin{tikzpicture}[scale=.7]
\begin{axis}[
	xmin = 1,
	xmax = 20,
	xtick = {5,10,15,20},
	ymin = 0,
	ymax = 0.1,
	ytick = {0,0.02,0.04,0.06,0.08,0.1},
	yticklabels = {0,0.02,0.04,0.06,0.08,0.1},
	 ytick scale label code/.code={},
	ylabel=$p$-value,
	xlabel=rank,
        height=10cm,
	width=10cm,
		legend style={legend columns=1, at={(0.275,0.975)}}
]
 \addplot[color=blue, very thick,  mark size=2pt, mark=*, fill=blue, fill opacity=0.5, draw opacity=0.75] coordinates {
( 1,0 ) ( 2,0 ) ( 3,0 ) ( 4,0 ) ( 5,0 ) ( 6,0 ) ( 7,0.003 ) ( 8,0.0035 ) ( 9,0.0039 ) ( 10,0.0044 ) ( 11,0.0048 ) ( 12,0.0052 ) ( 13,0.0057 ) ( 14,0.0061 ) ( 15,0.0065 ) ( 16,0.007 ) ( 17,0.0074 ) ( 18,0.0078 ) ( 19,0.0083 ) ( 20,0.0087 ) ( 20,1 ) ( 19,1 ) ( 18,1 ) ( 17,1 ) ( 16,1 ) ( 15,1 ) ( 14,1 ) ( 13,1 ) ( 12,1 ) ( 11,1 ) ( 10,1 ) ( 9,1 ) ( 8,1 ) ( 7,0.0044 ) ( 6,0.002 ) ( 5,0.002 ) ( 4,0.002 ) ( 3,0.002 ) ( 2,0.002 ) ( 1,0.002 )
};
 \addlegendentry{$k=7$}
 \addplot[color=red, very thick,  mark size=2pt, mark=*, fill=red, fill opacity=0.5, draw opacity=0.75] coordinates {
( 1,0 ) ( 2,0 ) ( 3,0 ) ( 4,0 ) ( 5,0 ) ( 6,0 ) ( 7,0 ) ( 8,0 ) ( 9,0 ) ( 10,0 ) ( 11,0 ) ( 12,0 ) ( 13,0 ) ( 14,0 ) ( 15,0 ) ( 16,0 ) ( 17,0 ) ( 18,0 ) ( 19,0.0083 ) ( 20,0.0087 ) ( 20,1 ) ( 19,0.0827 ) ( 18,0.0276 ) ( 17,0.0276 ) ( 16,0.0083 ) ( 15,0.0083 ) ( 14,0.0083 ) ( 13,0.0083 ) ( 12,0.0083 ) ( 11,0.0083 ) ( 10,0.0083 ) ( 9,0.0083 ) ( 8,0.0083 ) ( 7,0.0083 ) ( 6,0.0083 ) ( 5,0.0083 ) ( 4,6e-04 ) ( 3,6e-04 ) ( 2,6e-04 ) ( 1,6e-04 )
};
\addlegendentry{$k=19$}
\addplot[black, dashed, very thick, domain=0:20] {0.05*x/20};
\addlegendentry{BH}
\end{axis}
\end{tikzpicture}~
\begin{tikzpicture}[scale=.7, define rgb/.code={\definecolor{mycolor}{RGB}{#1}},
                    rgb color/.style={define rgb={#1},mycolor}]
\begin{axis}[
xmin = 1,
	xmax = 20,
	xtick = {5,10,15,20},
	ymin = 0,
	ymax = 0.1,
	ytick = {0,0.02,0.04,0.06,0.08,0.1},
	yticklabels = {0,0.02,0.04,0.06,0.08,0.1},
	 ytick scale label code/.code={},
	ylabel=,
	xlabel=rank,
        height=10cm,
	width=10cm
]
\addplot[rgb color={129, 43, 166}, very thick,  mark size=2pt, mark=., fill opacity=0.5, draw opacity=0.75] coordinates {
( 1,0 ) ( 2,0 ) ( 3,0.0013 ) ( 4,0.0017 ) ( 5,0.0022 ) ( 6,0.0026 ) ( 7,0.003 ) ( 8,0.0035 ) ( 9,0.0039 ) ( 10,0.0044 ) ( 11,0.0048 ) ( 12,0.0052 ) ( 13,0.0057 ) ( 14,0.0061 ) ( 15,0.0065 ) ( 16,0.007 ) ( 17,0.0074 ) ( 18,0.0078 ) ( 19,0.0083 ) ( 20,0.0087 ) ( 20,1 ) ( 19,1 ) ( 18,1 ) ( 17,1 ) ( 16,1 ) ( 15,1 ) ( 14,1 ) ( 13,1 ) ( 12,1 ) ( 11,1 ) ( 10,1 ) ( 9,1 ) ( 8,1 ) ( 7,1 ) ( 6,1 ) ( 5,1 ) ( 4,1 ) ( 3,0.0015 ) ( 2,7e-04 ) ( 1,7e-04 )
};
\addplot[rgb color={50, 0, 172}, very thick,  mark size=2pt, mark=., fill opacity=0.5, draw opacity=0.75] coordinates {
( 1,0 ) ( 2,0 ) ( 3,0 ) ( 4,0 ) ( 5,0.0022 ) ( 6,0.0026 ) ( 7,0.003 ) ( 8,0.0035 ) ( 9,0.0039 ) ( 10,0.0044 ) ( 11,0.0048 ) ( 12,0.0052 ) ( 13,0.0057 ) ( 14,0.0061 ) ( 15,0.0065 ) ( 16,0.007 ) ( 17,0.0074 ) ( 18,0.0078 ) ( 19,0.0083 ) ( 20,0.0087 ) ( 20,1 ) ( 19,1 ) ( 18,1 ) ( 17,1 ) ( 16,1 ) ( 15,1 ) ( 14,1 ) ( 13,1 ) ( 12,1 ) ( 11,1 ) ( 10,1 ) ( 9,1 ) ( 8,1 ) ( 7,1 ) ( 6,1 ) ( 5,0.0027 ) ( 4,0.0013 ) ( 3,0.0013 ) ( 2,0.0013 ) ( 1,0.0013 )
};
\addplot[rgb color={0, 0, 255}, very thick,  mark size=2pt, mark=., fill opacity=0.5, draw opacity=0.75] coordinates {
( 1,0 ) ( 2,0 ) ( 3,0 ) ( 4,0 ) ( 5,0 ) ( 6,0 ) ( 7,0.003 ) ( 8,0.0035 ) ( 9,0.0039 ) ( 10,0.0044 ) ( 11,0.0048 ) ( 12,0.0052 ) ( 13,0.0057 ) ( 14,0.0061 ) ( 15,0.0065 ) ( 16,0.007 ) ( 17,0.0074 ) ( 18,0.0078 ) ( 19,0.0083 ) ( 20,0.0087 ) ( 20,1 ) ( 19,1 ) ( 18,1 ) ( 17,1 ) ( 16,1 ) ( 15,1 ) ( 14,1 ) ( 13,1 ) ( 12,1 ) ( 11,1 ) ( 10,1 ) ( 9,1 ) ( 8,1 ) ( 7,0.0044 ) ( 6,0.002 ) ( 5,0.002 ) ( 4,0.002 ) ( 3,0.002 ) ( 2,0.002 ) ( 1,0.002 )
};
\addplot[rgb color={0, 85, 85}, very thick,  mark size=2pt, mark=., fill opacity=0.5, draw opacity=0.75] coordinates {
( 1,0 ) ( 2,0 ) ( 3,0 ) ( 4,0 ) ( 5,0 ) ( 6,0 ) ( 7,0 ) ( 8,0 ) ( 9,0.0039 ) ( 10,0.0044 ) ( 11,0.0048 ) ( 12,0.0052 ) ( 13,0.0057 ) ( 14,0.0061 ) ( 15,0.0065 ) ( 16,0.007 ) ( 17,0.0074 ) ( 18,0.0078 ) ( 19,0.0083 ) ( 20,0.0087 ) ( 20,1 ) ( 19,1 ) ( 18,1 ) ( 17,1 ) ( 16,1 ) ( 15,1 ) ( 14,1 ) ( 13,1 ) ( 12,1 ) ( 11,1 ) ( 10,1 ) ( 9,0.0065 ) ( 8,0.003 ) ( 7,0.003 ) ( 6,0.003 ) ( 5,0.003 ) ( 4,0.003 ) ( 3,0.003 ) ( 2,0.003 ) ( 1,0.003 )
};
\addplot[rgb color={85, 170, 0}, very thick,  mark size=2pt, mark=., fill opacity=0.5, draw opacity=0.75] coordinates {
( 1,0 ) ( 2,0 ) ( 3,0 ) ( 4,0 ) ( 5,0 ) ( 6,0 ) ( 7,0 ) ( 8,0 ) ( 9,0 ) ( 10,0 ) ( 11,0.0048 ) ( 12,0.0052 ) ( 13,0.0057 ) ( 14,0.0061 ) ( 15,0.0065 ) ( 16,0.007 ) ( 17,0.0074 ) ( 18,0.0078 ) ( 19,0.0083 ) ( 20,0.0087 ) ( 20,1 ) ( 19,1 ) ( 18,1 ) ( 17,1 ) ( 16,1 ) ( 15,1 ) ( 14,1 ) ( 13,1 ) ( 12,1 ) ( 11,0.0096 ) ( 10,0.0044 ) ( 9,0.0044 ) ( 8,0.0044 ) ( 7,0.0044 ) ( 6,0.0044 ) ( 5,0.0044 ) ( 4,0.0044 ) ( 3,0.0044 ) ( 2,0.0044 ) ( 1,0.0044 )
};
\addplot[rgb color={255, 255, 0}, very thick,  mark size=2pt, mark=., fill opacity=0.5, draw opacity=0.75] coordinates {
( 1,0 ) ( 2,0 ) ( 3,0 ) ( 4,0 ) ( 5,0 ) ( 6,0 ) ( 7,0 ) ( 8,0 ) ( 9,0 ) ( 10,0 ) ( 11,0 ) ( 12,0 ) ( 13,0.0057 ) ( 14,0.0061 ) ( 15,0.0065 ) ( 16,0.007 ) ( 17,0.0074 ) ( 18,0.0078 ) ( 19,0.0083 ) ( 20,0.0087 ) ( 20,1 ) ( 19,1 ) ( 18,1 ) ( 17,1 ) ( 16,1 ) ( 15,1 ) ( 14,1 ) ( 13,0.0141 ) ( 12,0.0063 ) ( 11,0.0063 ) ( 10,0.0063 ) ( 9,0.0063 ) ( 8,0.0063 ) ( 7,0.0063 ) ( 6,0.0063 ) ( 5,0.0063 ) ( 4,7e-04 ) ( 3,7e-04 ) ( 2,7e-04 ) ( 1,7e-04 )
};
\addplot[rgb color={255, 195, 0}, very thick,  mark size=2pt, mark=., fill opacity=0.5, draw opacity=0.75] coordinates {
( 1,0 ) ( 2,0 ) ( 3,0 ) ( 4,0 ) ( 5,0 ) ( 6,0 ) ( 7,0 ) ( 8,0 ) ( 9,0 ) ( 10,0 ) ( 11,0 ) ( 12,0 ) ( 13,0 ) ( 14,0 ) ( 15,0.0065 ) ( 16,0.007 ) ( 17,0.0074 ) ( 18,0.0078 ) ( 19,0.0083 ) ( 20,0.0087 ) ( 20,1 ) ( 19,1 ) ( 18,1 ) ( 17,1 ) ( 16,1 ) ( 15,0.0218 ) ( 14,0.0093 ) ( 13,0.0093 ) ( 12,0.0093 ) ( 11,0.0093 ) ( 10,0.0093 ) ( 9,0.0093 ) ( 8,0.0013 ) ( 7,0.0013 ) ( 6,0.0013 ) ( 5,0.0013 ) ( 4,0.0013 ) ( 3,0.0013 ) ( 2,0.0013 ) ( 1,0.0013 )
};
\addplot[rgb color={255, 110, 0}, very thick,  mark size=2pt, mark=., fill opacity=0.5, draw opacity=0.75] coordinates {
( 1,0 ) ( 2,0 ) ( 3,0 ) ( 4,0 ) ( 5,0 ) ( 6,0 ) ( 7,0 ) ( 8,0 ) ( 9,0 ) ( 10,0 ) ( 11,0 ) ( 12,0 ) ( 13,0 ) ( 14,0 ) ( 15,0 ) ( 16,0 ) ( 17,0 ) ( 18,0 ) ( 19,0.0083 ) ( 20,0.0087 ) ( 20,1 ) ( 19,0.0827 ) ( 18,0.0276 ) ( 17,0.0276 ) ( 16,0.0083 ) ( 15,0.0083 ) ( 14,0.0083 ) ( 13,0.0083 ) ( 12,0.0083 ) ( 11,0.0083 ) ( 10,0.0083 ) ( 9,0.0083 ) ( 8,0.0083 ) ( 7,0.0083 ) ( 6,0.0083 ) ( 5,0.0083 ) ( 4,6e-04 ) ( 3,6e-04 ) ( 2,6e-04 ) ( 1,6e-04 )
};
\addplot[rgb color={255, 0, 0}, very thick,  mark size=2pt, mark=., fill opacity=0.5, draw opacity=0.75] coordinates {
( 1,0 ) ( 2,0 ) ( 3,0 ) ( 4,0 ) ( 5,0 ) ( 6,0 ) ( 7,0 ) ( 8,0 ) ( 9,0 ) ( 10,0 ) ( 11,0 ) ( 12,0 ) ( 13,0 ) ( 14,0 ) ( 15,0 ) ( 16,0 ) ( 17,0 ) ( 18,0 ) ( 19,0.0083 ) ( 20,0.0087 ) ( 20,1 ) ( 19,0.0827 ) ( 18,0.0276 ) ( 17,0.0276 ) ( 16,0.0083 ) ( 15,0.0083 ) ( 14,0.0083 ) ( 13,0.0083 ) ( 12,0.0083 ) ( 11,0.0083 ) ( 10,0.0083 ) ( 9,0.0083 ) ( 8,0.0083 ) ( 7,0.0083 ) ( 6,0.0083 ) ( 5,0.0083 ) ( 4,6e-04 ) ( 3,6e-04 ) ( 2,6e-04 ) ( 1,6e-04 )
};

\end{axis}
\end{tikzpicture}
 \caption{\revise{Regions of ordered p-value sequences, for $m=20$ and $\alpha=0.05$.
The $x$-axis gives the rank and the $y$-axis gives the ordered p-values $p_{(1)}\leq \cdots \leq p_{(m)}$.
The colored regions show ordered sequences for which $\textcSu$ rejects at least $k=3, 5, \ldots,19$ hypotheses while Su rejects none.
These regions were constructed by greedily choosing the largest $p_{(m)}, \ldots, p_{(1)}$, in that order, to achieve the minimal desired local e-values.
The black dashed line is the rejection threshold of the BH procedure, i.e., the BH procedure rejects the $r$ smallest p-values if $r$ is the largest rank such that $p_{(r)}$ is below this line.
The upper bound on the region of p-value sequences where the $\textcSu$ procedure contains p-value sequences that make 19 or more rejections is higher than the BH threshold, which means that the BH procedure makes strictly fewer rejections for those sequences of p-values.}} \label{fig:Su}
\end{figure}
     
}{
We provide concrete examples of improvement in \Cref{sec: p-value power improvements}, where we give an event in which Su rejects only one hypothesis, but $\textcSu$ contains the discovery set that rejects all, and an example where $\cSu$ makes more discoveries than BH.
}

\ifarver{
\section{Benjamini-Hochberg and its adaptive variants}\label{sec: bh and variants}

The seminal FDR control method by \citet{benjamini_controlling_false_1995} remains the most popular by far. It rejects the following discovery set:
$$
\rejset_\alpha^{\BH} \coloneqq \rejset^{(\mathbf{r}_\alpha^\BH)}, \text{ where }\mathbf{r}_\alpha^\BH \coloneqq \max\{r \in [m]: \pvalue_{(r)} \leq \alpha r / m \},$$
\revise{where $\mathbf{r}_\alpha^\BH$ is taken as 0 if the maximum does not exist. However, it has some puzzling aspects. First, BH is not guaranteed under arbitrary dependence. Its classical finite-sample validity result holds under PRDS \eqref{eq: PRDS}, which constrains the joint distribution of the p-values of true and false hypotheses, although other sufficient conditions are known \citep{goeman_uniform_improvement_2026}. Such an assumption is unusual among multiple testing procedures, which generally make assumptions only on the distribution of true hypotheses, and \citet{su2018fdr} made a forceful argument against it. Second, \BH\ is known to be inadmissible under PRDS. The minimally adaptive BH (MABH) procedure constructed by \citet{solari2017minimally} gives a small uniform improvement; more recently, \citet{goeman_uniform_improvement_2026} constructed a closed BH procedure that uniformly improves BH under PRDS and a weaker sufficient condition. MABH rejects}
\begin{align}
    \mathbf{r}_\alpha^{\MABH} \coloneqq \max\left\{i \in [m]: i^{-1}\cdot \pvalue_{(i)} \leq \frac{\alpha}{m - 1}\right\}, \rejset_\alpha^{\MABH} \coloneqq \begin{cases}
        \emptyset & \text{if }\rejset_\alpha^\BH = \emptyset;\\
        \rejset^{(\mathbf{r}_\alpha^{\MABH})} & \text{otherwise.}
    \end{cases}
\end{align}
For the \BH\ procedure, we can define the following local e-value
\begin{equation} \label{eq: BH local e-value}
\evalue_S = \frac{1}{ |S|} \sum_{i \in S} \widetilde{\evalue}_i \text{ where } \widetilde{\evalue}_i =  \frac{m}{\alpha |\rejset_\alpha^\BH|} \mathbf{1}\Big\{\pvalue_i \leq \frac{\alpha |\rejset_\alpha^\BH|}{m} \Big\}.
\end{equation}
\revise{\citet{li_note_e-values_2025} showed that $\widetilde\evalue_i$ is a compound e-value under PRDS and that \eBH\ applied to these compound e-values reconstructs \BH, but they did not directly show it is an e-value as well. The validity of $\widetilde\evalue_i$, and hence of $\evalue_S$, can also be seen directly from the dependency control condition of \citet{blanchard_two_simple_2008}. For each null $i$, take $U=\pvalue_i$, $V=|\rejset_\alpha^\BH|$, and the linear shape function $\beta(v)=v$. Proposition 3.6 of \citet{blanchard_two_simple_2008} gives
\[
\expect_\pdist\left[\frac{\mathbf{1}\{\pvalue_i\leq \alpha |\rejset_\alpha^\BH|/m\}}{|\rejset_\alpha^\BH|}\right]\leq \frac{\alpha}{m}
\]
under PRDS, with the ratio defined as zero when $|\rejset_\alpha^\BH|=0$. Multiplying by $m/\alpha$ proves that $\widetilde\evalue_i$ is an e-value. Averaging over $i\in S$ then proves that $\evalue_S$ is an e-value for $H_S$.}
\revise{We then define $\rejcol_\alpha^{\cBH}$ as the e-Closure procedure resulting from this e-collection.}

The e-collection derived from \eqref{eq: BH local e-value} depends on all $m$ p-values through $|\rejset_\alpha^\BH|$ instead of only on p-values pertaining to hypotheses $S$ (as is the case for the e-collections of $\cBY$ and $\cSu$). Thus, it is clear that an assumption on the joint distribution of p-values of true and false hypotheses is needed.

We can also define the corresponding local e-values for the minimally adaptive BH procedure as follows.
\[
\evalue_{[m]} = 1\{\rejset^\BH_\alpha \neq \emptyset\}/\alpha, \qquad
\evalue_S = \frac{m-1}{ |S|} \sum_{i \in S} \frac{1}{\alpha |\rejset_\alpha^\MABH|} \mathbf{1}\Big\{\pvalue_i \leq \frac{\alpha |\rejset_\alpha^\MABH|}{m-1} \Big\}\text{ for }S \subset [m]\label{eq: mabh local evalue}
\]
\begin{proposition} \label{thm: e for MABH}
    \revise{Under PRDS, $\ecol = (\evalue_S)_{S\in 2^{[m]}}$ with e-values defined in \eqref{eq: mabh local evalue} is a valid e-collection.}
\end{proposition}
We define $\rejcol_\alpha^{\cMABH}$ as the e-Closure procedure resulting from the e-collection comprising these local e-values.

\revise{This approach generalizes to other adaptive versions of the BH procedure. Unlike MABH, the validity claim below assumes that the p-values $(\pvalue_1, \dots, \pvalue_m)$ are independent and that the estimator satisfies the stated monotonicity and moment conditions.} For any $i \in N_\pdist$, define
\begin{align}
\hat{\boldsymbol{\pi}}_0^i \coloneqq \hat\pi(\pvalue_1, \dots, \pvalue_{i - 1}, 0, \pvalue_{i + 1}, \dots, \pvalue_m), \qquad \expect_\pdist\left(\frac{1}{\hat{\boldsymbol{\pi}}_0^i}\right) \leq \frac{m}{|N_\pdist|}, \label{eq: pi0 condition}
\end{align}
where $\hat\pi$ is an estimator of the null proportion and is also required to be coordinatewise nondecreasing.
Let $\hat{\boldsymbol{\pi}}_0 \coloneqq \max_{i \in [m]} \hat{\boldsymbol{\pi}}_0^i$. Now, we define the adaptive BH discovery set as
\begin{align}
    \mathbf{r}_\alpha^{\adaBH} \coloneqq \max \left(\left\{i \in [m]: i^{-1} \cdot \pvalue_{(i)} \leq \frac{\alpha}{\hat{\boldsymbol{\pi}}_0m}\right\} \cup \{0\}\right), \qquad \rejset_\alpha^{\adaBH} \coloneqq \rejset^{(\mathbf{r}_\alpha^\adaBH)}.
\end{align}
Storey's procedure \citep{storey_direct_approach_2002a, storey2004strong, benjamini2006adaptive}, which sets $\hat\pi(p_1, \dots, p_m) = (1 + \sum_{i=1}^m \mathbf{1}\{p_i > \lambda\}) / ((1 - \lambda)m)$ for a fixed tuning parameter $\lambda \in (0, 1)$, is an instance of this adaptive BH procedure. Methods of this type have been well studied because they can estimate the null proportion and increase power when the number of non-null hypotheses is substantial \citep{benjamini2006adaptive,sarkar_methods_controlling_2008,blanchard2009adaptive,dohler_unified_class_2023,gao_adaptive_null_2025}.

We can define the e-collection for improving the adaptive BH procedure. We now formulate the compound e-values and resulting local e-values that correspond to an adaptive procedure.
\begin{align}
    \mathbf{e}_S = \frac{1}{m}\sum_{i \in S}\widetilde\evalue_i \text{ where }\widetilde\evalue_i = \frac{m}{\alpha (|\rejset_\alpha^{\adaBH}| \vee 1)} \mathbf{1}\Big\{\pvalue_i \leq \frac{\alpha (|\rejset_\alpha^{\adaBH}| \vee 1)}{m \hat{\boldsymbol{\pi}}_0} \Big\}, \label{eq: adaBH local e-value}
\end{align}
\revise{The validity of these compound e-values follows directly from the argument used to prove Theorem 11 of \citet{blanchard2009adaptive}. Write $G(\mathbf{p})=1/\hat{\boldsymbol{\pi}}_0(\mathbf{p})$. For each null $i$, condition on the other p-values and apply their Lemma 27 with $U=\pvalue_i$, $g(U)=|\rejset_\alpha^{\adaBH}(\pvalue_1,\ldots,\pvalue_{i-1},U,\pvalue_{i+1},\ldots,\pvalue_m)|\vee 1$, and $c=\alpha G(\mathbf{p}_{0,i})/m$, where $\mathbf{p}_{0,i}$ replaces $\pvalue_i$ by zero. Monotonicity of $G$ in each of its arguments yields
\[
\expect_\pdist[\widetilde\evalue_i]\leq \expect_\pdist[G(\mathbf{p}_{0,i})]\leq \frac{m}{|N_\pdist|}.
\]
Consequently, under $H_S$, $\expect_\pdist[\mathbf{e}_S]\leq |S|/|N_\pdist|\leq 1$, proving that the local e-values are valid. Equivalently, Theorem 11 establishes FDR control for this adaptive BH procedure, after which \eqref{eq: universal compound evalue} gives the same compound e-values.}
The adaptive closed BH procedure, $\cadaBH$, is then defined as $\rejcol^{\cadaBH}_\alpha \coloneqq \Rcal^\FDR_\alpha(\mathbf{E})$ for the e-collection based on \eqref{eq: adaBH local e-value}. \revise{The theorem below shows that this e-collection exactly reconstructs adaptive BH, with no additional nonempty discovery sets.}
\revise{Unfortunately, the closed \BH\ procedure $\cBH$ resulting from \eqref{eq: BH local e-value} never gives a larger rejected set. Moreover, it gives no nontrivial simultaneity unless $\rejset_\alpha^{\BH}= [m]$. On that full-rejection branch, the simultaneous collection contains every singleton; by \Cref{thm: FWER}, the global selector formed by taking the union of all available singletons controls FWER. This is an unconditional guarantee for that selector, not a claim that BH conditionally controls FWER given the event $\rejset_\alpha^{\BH}=[m]$. In contrast, the corresponding closed MABH and adaptive BH procedures have no such exception: they exactly reconstruct their original procedures even when all hypotheses are rejected.} \begin{theorem}\label{thm: BH-no-simultaneity}
\revise{Assume PRDS for the BH and MABH claims. For the adaptive BH claim, assume independent p-values and that $\hat{\boldsymbol\pi}_0$ satisfies \eqref{eq: pi0 condition} and is coordinatewise nondecreasing. For the BH procedure, if $\rejset_\alpha^\BH\neq [m]$, then $\rejcol_\alpha^{\cBH}=\{\emptyset,\rejset_\alpha^\BH\}$. Otherwise, if $\rejset_\alpha^\BH=[m]$, then $\rejcol_\alpha^{\cBH}=2^{[m]}$. For minimally adaptive BH and adaptive BH, we always have $\rejcol_\alpha^{\cMABH}=\{\emptyset,\rejset_\alpha^\MABH\}$ and $\rejcol_\alpha^{\cadaBH}=\{\emptyset,\rejset_\alpha^\adaBH\}$.}
\end{theorem}
\begin{proof}
    First, consider the case where $\rejset = \rejset^\BH_\alpha$ and $\rejcol = \rejcol^\cBH_\alpha$. We can rewrite each local e-value as follows:
\begin{eqnarray*}
\mathbf{e}_S &=&  \frac{m |\rejset \cap S|}{\alpha |S|(|\rejset| \vee 1)}.
\end{eqnarray*}
If
$\rejset =[m]$, then for all $S$, $ \mathbf{e}_S = 1/ \alpha$, which is always greater than or equal to $\textnormal{FDP}_S(R)/\alpha$ for every $R$ and $S$, thus  $\rejcol =2^{[m]}$.

\revise{If $\rejset=\emptyset$, all the local e-values are zero and the conclusion is immediate.} Suppose $\rejset \neq \emptyset$ and $\rejset \neq [m]$.
We show that for every non-empty $R \neq \rejset$, there exists an $S$ such that
\begin{eqnarray*}
\frac{m|\rejset \cap S|}{ |S||\rejset|}&<& \frac{|R \cap S|}{ |R|}.
\end{eqnarray*}
Let $R \subset \rejset$. Then
\begin{align}
    0 < (m - |\rejset|)(|\rejset|- |R|).
\end{align}
This implies that
\begin{align}
    m|\rejset| + |R||\rejset| - |\rejset|^2 = (|R| + m - |\rejset|)|\rejset| > m|R|.
\end{align}
Now, we take $S = R \cup ([m] \setminus \rejset)$.
\begin{align}
    \frac{m|\rejset \cap S|}{|S||\rejset|} = \frac{m |R|}{(|R|+m-|\rejset|) |\rejset|}  < 1 = \frac{|R \cap S|}{|R|}.
\end{align}

Let $R \not\subseteq \rejset$, that is, $R$ contains some  $i \notin \rejset$, and
take $S=[m] \setminus \rejset$. Then,
\begin{eqnarray*}
    \frac{m|\rejset \cap S|}{|S||\rejset|} = 0<   \frac{|\{i\}|}{|R|} \leq  \frac{|R \cap ( [m] \setminus \rejset )|}{|R|}.\\
\end{eqnarray*}
Thus, we have shown the desired result for $\rejset_\alpha^\BH$ and $\rejcol_\alpha^\cBH$.

\revise{Now let $\rejset=\rejset_\alpha^\MABH$ and $\rejcol=\rejcol_\alpha^{\cMABH}$. The case $\rejset=\emptyset$ is immediate, so suppose $\rejset\neq\emptyset$. The original set $\rejset$ belongs to $\rejcol$: for $S=[m]$ this follows from $\evalue_{[m]}=1/\alpha$, while for every proper $S$ it follows from $(m-1)/|S|\geq1$. Consider any nonempty $R\neq\rejset$. If $|R|>|\rejset|$, choose $j\in\rejset$ and set $S=[m]\setminus\{j\}$. Then
\[
\alpha\evalue_S=\frac{|\rejset|-1}{|\rejset|}
<1-\frac{1}{|R|}
\leq\frac{|R\cap S|}{|R|}.
\]
If instead $|R|\leq|\rejset|$, choose $i\in\rejset\setminus R$ and again set $S=[m]\setminus\{i\}$. In this case,
\[
\alpha\evalue_S=\frac{|\rejset|-1}{|\rejset|}<1=\frac{|R\cap S|}{|R|}.
\]
Thus, for every nonempty $R\neq\rejset$ there exists some $S$ s.t. $\alpha \evalue_S < \FDP_S(R)$, including when $\rejset=[m]$, and hence $\rejcol_\alpha^{\cMABH}=\{\emptyset,\rejset_\alpha^\MABH\}$.}

\revise{Finally, let $\rejset=\rejset_\alpha^{\adaBH}$. The indicator in \eqref{eq: adaBH local e-value} equals one precisely for the hypotheses in $\rejset$. 
\[
\alpha\evalue_S=\frac{|\rejset\cap S|}{|\rejset|\vee1}.
\]
Otherwise, consider a nonempty $R\neq\rejset$. If $R\not\subseteq\rejset$, choose $i\in R\setminus\rejset$ and $S=\{i\}$, yielding $\alpha\evalue_S=0<|R\cap S|/|R|$. If $R\subsetneq\rejset$, take $S=R$, yielding $\alpha\evalue_S=|R|/|\rejset|<1=|R\cap S|/|R|$. Therefore $\rejcol_\alpha^{\cadaBH}=\{\emptyset,\rejset_\alpha^\adaBH\}$, including when $\rejset_\alpha^\adaBH=[m]$.}
\end{proof}

\revise{The particular e-collections considered in this section therefore do not substantially improve \BH\ or its adaptive variants. \citet{goeman_uniform_improvement_2026}, however, develops a different e-Closure construction that does uniformly improve over \BH\ in ways that are different from minimally adaptive \BH.}
 }{
\begin{remark}
\revise{The e-Closure Principle can be applied to the Benjamini-Hochberg procedure and its adaptive variants (e.g., MABH and Storey's procedure). However, the resulting closed procedures do not improve upon the original methods. For MABH and adaptive BH, the closed procedure always yields only $\rejcol=\{\emptyset,\rejset\}$, even when $\rejset=[m]$. For BH, the same conclusion holds when $\rejset\neq[m]$, whereas $\rejset=[m]$ yields $\rejcol=2^{[m]}$. See \Cref{sec: bh and variants} for details.}
\revise{However, a different choice of e-collection can improve the BH procedure: \citet{goeman_uniform_improvement_2026} constructs a closed BH procedure that uniformly improves BH under PRDS and under a weaker sufficient condition.}
\end{remark}
}

\section{Post hoc control in multiple testing} \label{sec: flexible}

Classical multiple testing requires the researcher to prespecify the error rate to be controlled and the nominal level at which to control this error rate. The procedure then returns a single random set that has such control. This set-up is rigorous but highly inflexible. \citet{goeman2011multiple} argued that there are many situations in which the researcher might want to deviate from this set-up. For example, the rejected set may be too large for planned follow-up experiments, and the researcher may prefer to have a smaller rejection set that guarantees the same error rate or possibly a more stringent one. In other cases, the rejected set may be too small, or a subset of the rejected set may be of interest.

Error control according to \Cref{def: single ER}, however, does not allow any deviation from the rejected set. This has been noted in the context of FDR control by \citet{finner2001false}, who showed how researchers can, intentionally or not, ``cheat with FDR.'' Reducing the set $\mathbf{R}$ while retaining its prespecified form, e.g., rejecting a set of the $\mathbf{s} < \mathbf{r}$ smallest p-values in BH, compromises FDR control if $\mathbf{s}$ is chosen post hoc \citep{katsevich2020simultaneous}. Deviations from the set $\mathbf{R}$ of the $\mathbf{r}$ smallest p-values have been investigated in prior work, but such methods usually require the researcher to commit to an algorithm for reducing or changing $\mathbf{R}$ before seeing the data \citep{lei2018adapt, katsevich2023filtering}. Some methods allow interactive post hoc amendments to this prespecified algorithm based on progressively revealing parts of the data \citep{lei2021general, katsevich2020simultaneous}. As far as we are aware, no FDR controlling methods have so far been proposed that allow researchers to reduce or amend $\mathbf{R}$ after looking at all of the data. Notably, methods controlling tail probabilities of the FDP, rather than the FDR, obtain some post hoc flexibility directly from the Closure Principle \citep{goeman2011multiple, goeman2021only}.

The e-Closure Principle extends the possibilities for post hoc flexibility in the rejected set to general error rates including FDR. \revise{Here, post hoc flexibility is the operational interpretation of simultaneity: simultaneous control gives a collection of valid choices, and the researcher may choose from that collection after seeing the data.} It also allows post hoc flexibility in error loss, and in some cases even the nominal level. We will explore these properties of e-Closure in this section.

\subsection{Post hoc choice of discovery set} \label{sec: post hoc sets}

Simultaneous control of error rates over several candidate sets as offered by \Cref{def: simultaneous ER} is useful whenever information external to the p-values may influence the conclusions of an experiment. In such cases it is often difficult for researchers to specify exactly, before seeing the data, how such external information will be used, and post hoc flexibility is desirable. Such post hoc flexibility, argued for by \citet{goeman2011multiple}, was thus far only available with FDP control. The e-Closure Principle extends this to other error rates, most importantly FDR. We give two brief use cases from the literature for such flexibility.

In bioinformatics, a larger rejected set is not always preferable, since there is limited budget for follow-up experiments. Moreover, significance is not the only consideration for rejection, and researchers prefer less significant results with larger effect size estimates. A popular methodology, called the `Volcano Plot', reduces the initial FDR-significant set by discarding findings with low effect size. Done naively, such discarding may compromise FDR control \citep{ebrahimpoor2021inflated}, and such a reduction retains FDR control as long as the reduced set is one of the candidate discovery sets via the e-Closure Principle.

In neuroimaging, hypotheses correspond to voxels (3d generalizations of pixels) in the brain, and the significant set $\mathbf{R}$ often splits naturally into several brain areas (``clusters''). In such situations, FDR control on the total set $\mathbf{R}$ is not very informative, since researchers will interpret their results in terms of the clusters \citep{rosenblatt2018all}. Simultaneous FDR control using the e-Closure Principle allows researchers to claim that these clusters have FDR at most $\alpha$ if and only if these clusters are among the candidate discovery sets. Similar considerations apply in bioinformatics, where hypotheses correspond to molecular markers, which can be meaningfully grouped into sets called pathways \citep{ebrahimpoor2020simultaneous}.
\paragraph{A simultaneous discovery version of the knockoff filter.} We show that simultaneity can be achieved in the knockoff filter \citep{barber_controlling_2015,candes_panning_gold_2018}. A key component of knockoff procedures is that the procedure forms a test statistic $\mathbf{w}_i$ for each $i \in [m]$ such that the sign of $\mathbf{w}_i$ is independent of all other statistics and has an equal chance of being in $\{\pm 1\}$ if $i \in N_\pdist$. The knockoff filter is defined by:
\begin{align}
\revise{\mathbf{c}_\alpha^\Kn \coloneqq \min\left\{c\in\{|\mathbf w_i|:|\mathbf w_i|>0\}: \frac{1+\sum_{i \in [m]}\ind\{\mathbf{w}_i \leq -c\}}{\sum_{i \in [m]}\ind\{\mathbf{w}_i \geq c\}} \leq \alpha \right\}, \qquad \rejset^{\Kn}_\alpha \coloneqq \{i \in [m]: \mathbf{w}_i \geq \mathbf{c}_\alpha^\Kn\},}
\end{align}
\revise{with the convention that the minimum of the empty set is $+\infty$. Assume the knockoff statistics satisfy the usual null sign-flip property. Using the coin-flip argument underlying Lemma 1 of \citet{barber_supplement_2015} and eq. (8) of \citet{ren2024derandomised}, we define the following local e-values}:
\begin{align}
    \mathbf{e}_S = \frac{\sum_{i \in S} \mathbf{1}\{\mathbf{w}_i \geq \mathbf{c}_\alpha^\Kn\}}{1 + \sum_{j \in S} \mathbf{1}\{\mathbf{w}_j \leq -\mathbf{c}_\alpha^\Kn\}}. \label{eq: knockoff local e-value}
\end{align} We refer to the resulting simultaneous discovery procedure as $\cKnockoffs$, i.e., $\rejcol_\alpha^{\cKn}  \coloneqq \Rcal_\alpha^\FDR(\mathbf{E})$. 
\begin{theorem}\label{thm: simultaneity improvement}
    \revise{Under the knockoff null sign-flip property, the set $\rejcol_\alpha^{\cKn}$ is given by}
    \begin{align}
        \rejcol_\alpha^{\cKn}=\left\{R\subseteq [m]: \mathbf{w}_i\geq \mathbf{c}_\alpha^\Kn \text{ for all } i\in R \text{ and } \frac{|R|}{1 + \sum_{j \in [m]} \mathbf{1}\{\mathbf{w}_j \leq -\mathbf{c}_\alpha^\Kn\}}\geq 1/\alpha \right\} \cup \{\emptyset\}.\label{eq: knockoff_char}
    \end{align}
\end{theorem}
Hence, $\rejcol_\alpha^{\cKn}$ is a uniform improvement in simultaneity over $\rejset_\alpha^{\Kn}$, although it contains no strict superset of $\rejset_\alpha^{\Kn}$.
To illustrate a concrete example where simultaneity is improved by $\rejcol_\alpha^{\cKn}$, consider the following instance: let $m = 6$ and $\alpha = 0.4$.
\revise{Now, let $\mathbf{w}=(6,5,4,3,-2,-1)$. Then $\mathbf c_\alpha^\Kn=3$ and the knockoffs discovery set is $\rejset_\alpha^\Kn = [4]$.} In addition, $\rejcol_\alpha^{\cKn}$ also includes all subsets of size 3 of $[4]$. Thus, closed knockoffs can provide simultaneity improvements in such settings.

\subsection{Post hoc choice of loss}\label{sec: post hoc loss}

In the classical setting of multiple testing, researchers have to choose the error rate to be controlled before seeing the data. This may lead to disappointment if the number or strength of the signals differs from what was expected. For example, a researcher may set out to control FWER but obtain no discoveries and, in hindsight, prefer FDR control. Conversely, a researcher who controls FDR and obtains hundreds of discoveries might prefer a smaller number of discoveries with an FWER guarantee. In other situations, researchers may want two tiers of discoveries simultaneously: an inner, more certain set of FWER discoveries and an outer, more tentative FDR significant set. Such situations call for simultaneity or post hoc flexibility in the choice of error rate.

Simultaneous control of FDR (for some outer set) and FWER (for some inner set) is actually allowed as a direct consequence of \Cref{def: simultaneous ER} of simultaneous FDR control, and the relationship between the two error rates. To see this, let $\rejcol$ control FDR, and consider the singleton sets in $\rejcol$. Since singleton sets have an FDP of either 0 or 1, there is no distinction between FDR and FWER for such sets.
The researcher may now choose to reject the union of all these singleton sets and control FWER for this set, in addition to FDR. This switch from FDR control to FWER control may be made fully post hoc, after observing all the data and the resulting rejected collection $\boldsymbol{\mathcal{R}}$. It is made explicit in Theorem \ref{thm: FWER}.

\begin{theorem} \label{thm: FWER}
Suppose $\boldsymbol{\mathcal{R}}$ has simultaneous control of the FDR at level $\alpha$. Define $
\mathbf{R} = \{i \in [m]\colon \{i\} \in \boldsymbol{\mathcal{R}}\}$.
Then $\mathbf{R}$ controls FWER at level $\alpha$, that is, for every $\mathrm{P} \in M$, $
\mathrm{P}(\mathbf{R} \cap N_\mathrm{P} = \emptyset) \geq 1-\alpha$.
\end{theorem}
\Cref{thm: FWER} is a special case of the more general result that the e-Closure Principle allows post hoc choice of the error metric. Let $\mathcal{F}$ be the set of all error functions under consideration.
\begin{definition}
    A collection of discovery sets $(\rejcol^{\ER_\F}_\alpha)_{\F \in \mathcal{F}}$ provides \emph{simultaneous error rate control} over $\mathcal{F}$ if
\begin{align}
    \expect_{\pdist}\left[\sup_{\F\in \mathcal{F}} \max_{R \in \rejcol^{\ER_\F}_\alpha}\  \F_{N_\pdist}(R)\right]
    \leq \alpha \text{ for all }\pdist \in M
    \label{eq:sup-fcond}
\end{align}
\end{definition}

A collection of discovery sets with simultaneous error rate control over $\mathcal{F}$ allows one to choose any $\boldsymbol{\F}\in \mathcal{F}$ and a corresponding rejection set $\rejset_\alpha^{\ER_{\boldsymbol{\F}}}\in \rejcol_\alpha^{\ER_{\boldsymbol{\F}}}$ post hoc, while ensuring valid error rate control. \revise{This is an unconditional post hoc guarantee over the data-dependent choice of $\boldsymbol{\F}$ and $\rejset_\alpha^{\ER_{\boldsymbol{\F}}}$; it should not be interpreted as error control conditional on the realized choice of loss.} It turns out that the e-Closure Principle is necessary and sufficient for such simultaneous error rate control.
\begin{theorem}[The e-Closure Principle for simultaneous error rate control]\label{thm: multiloss}
    For any collection of nonnegative losses $ \mathcal{F}$ and collection of discovery sets $(\rejcol_\alpha^{\ER_\F})_{\F \in \mathcal{F}}$ with simultaneous error rate control over $\mathcal{F}$, there exists 
    an e-collection $\mathbf{E}$ and resulting collection of simultaneous discovery sets $(\Rcal^{\ER_\F}_\alpha(\mathbf{E}))_{\F \in \mathcal{F}}$ such that $\rejcol^{\ER_\F}_\alpha \subseteq \Rcal^{\ER_{\F}}_\alpha(\mathbf{E})$ for all $\F \in \mathcal{F}$.
Further, $(\Rcal^{\ER_\F}_\alpha(\mathbf{E}))_{\F \in \mathcal{F}}$ provides simultaneous error rate control over $\mathcal{F}$ for every e-collection $\mathbf{E}$.

\end{theorem}
\begin{example}[Post hoc choice of the FDX threshold]
\revise{False discovery exceedance control at level $\gamma$, also called $\gamma$-FDX control, has been a well studied problem in the multiple testing literature, with early contributions by \citet{korn_controlling_number_2004} and
\citet{van_augmentation_2004}
more recent developments including \citet{lehmann_romano_generalizations_2005,romano_shaikh_stepdown_2006,guo_romano_generalized_2007,romano_wolf_generalized_2007,guo_he_sarkar_fdp_2014,delattre_roquain_new_2015}. While such methods have generally been constructed via FWER controlling procedures, \citet{katsevich2020simultaneous} showed that $\gamma$-FDX guarantees could also be produced through pathwise FDR controlling procedures. Here, we show that simultaneous (over $\gamma$) $\gamma$-FDX control is a special case of simultaneous error control. For a fixed $\gamma\in[0,1)$, define the loss $\F^\gamma_N(R)=\ind\{\FDP_N(R)>\gamma\}$.}

\revise{In particular, the usual fixed $\gamma$-FDX guarantee is recovered by taking $\mathcal{F}$ to contain a single value of $\gamma$; taking $\alpha=1/2$ yields the corresponding median-FDP interpretation (e.g., as is studied by \citet{hemerik}), connecting this example to quantile-FDP control \citep{genovese2006exceedance}.}

\revise{Otherwise, setting $\mathcal{F} = \{\F^\gamma: \gamma \in (0, 1)\}$ permits simultaneous error rate control over $\mathcal{F}$ and enables the researcher to choose $\gamma$ post hoc while retaining the aggregate guarantee in \eqref{eq:sup-fcond}, which recovers the simultaneous control of \citep{goeman2011multiple,goeman2021only}. Formally, if we have a family of discovery sets $(\rejset^\gamma)_{\gamma \in (0, 1)}$ such that $\rejset^\gamma \in \rejcol^{\ER_{\F^\gamma}}_\alpha$ for all $\gamma \in (0, 1)$, then \eqref{eq:sup-fcond} implies the following uniform guarantee:
\begin{align}
\prob(\text{exists }\gamma \in (0, 1) \text{ s.t. }\FDP_{N_\pdist}(\rejset^\gamma) > \gamma) \leq \alpha.
\end{align}}
\end{example}
\begin{example}[Deciding between e-Holm and $\CeBH$ post hoc\label{example:Holm}]
    Let $\evalue_S$ be defined as the mean e-value in \Cref{sec: e-combining}. Let our collection of losses be $\Fcal = \{\FDP, \F\}$ where we define $\F_S(R)=\ind\{|S\cap R|> 0\}$ as the loss for FWER control. Note that $\rejset^{\textnormal{e-Holm}}_\alpha \coloneqq R^\FWER_\alpha(\ecol)$ is precisely the e-Holm procedure \citep{vovk_e-values_calibration_2021, hartog_family-wise_error_2025}. Consequently, we could decide based on the data to control the FDR with the $\CeBH$ procedure or to control the FWER with the e-Holm procedure. \revise{An example of such a rule is to choose a fixed threshold $r^* \in [m]$, and we choose to only use the $\CeBH$ discovery set if $|\rejset^\CeBH_\alpha| \geq r^*$, and the e-Holm discovery set otherwise. The guarantee we would get is as follows:
    \begin{align}
        \expect_\pdist\left[\FDP_{N_\pdist}(\rejset^\CeBH_\alpha) \cdot \ind\{|\rejset^\CeBH_\alpha| \geq r^*\} + \F_{N_\pdist}(\rejset_\alpha^{\textnormal{eHolm}})\cdot \ind\{|\rejset^\CeBH_\alpha| < r^*\}\right] \leq \alpha.
    \end{align}}
    \revise{A practitioner may want post hoc control under this decision rule because FDR may not be meaningful for small discovery sets, yet the size of the data driven discovery set is not known in advance. Here, the practitioner controls FWER for small discovery sets and FDR for large discovery sets. This provides a meaningful guarantee for distributions that yield only discovery sets smaller than $r^*$ (e.g., in a strong, sparse regime): in this case, the procedure guarantees FWER control. In every case, FDR is controlled because FWER control implies FDR control. Thus, a post hoc choice of loss allows the practitioner to adapt the error guarantee to the unknown underlying distribution.}
    
    \revise{A related multi-resolution procedure appears in \citet{hemerik_multiresolution_2025} for FWER and $\gamma$-FDX control that generalizes the well-known maxT procedure \citep{westfall_young_resampling_1993} for FWER control. For a fixed $\gamma$, when the maxT procedure rejects fewer than $\gamma^{-1}$ hypotheses, it coincides with the multi-resolution procedure and this multi-resolution procedure retains the same FWER control. When at least $\gamma^{-1}$ hypotheses are rejected by maxT, the multi-resolution procedure can make additional rejections while retaining $\gamma$-FDX control. Indeed, a promising direction for future work is to investigate what e-Closure based procedures can be derived for hybrid $
\gamma$-FDX and FWER control as well.}
\end{example}

We can also utilize our result attaining control for multiple error rates to recover the Closure Principle for simultaneous FDP control \citep{genovese2006exceedance,goeman2011multiple, goeman2021only}.
A random function $\mathbf{q}:2^{[m]}\to [0,1]$ is said to provide simultaneous FDP control, if for all $\mathrm{P} \in M$,
\begin{align}
            \pdist(\mathbf{q}(R) \geq \FDP_{N_\mathrm{P}}(R) \text{ for all }R\in 2^{[m]}) \geq 1 - \alpha.
\end{align}
Similar to \citet{goeman2021only}, we will focus on simultaneous true discovery guarantee in the following, which is equivalent to simultaneous FDP control but easier to handle. In contrast to $\mathbf{q}$, a simultaneous true discovery procedure $\mathbf{d}:2^{[m]}\to \mathbb{N}_0$ provides a simultaneous lower bound for the number of true discoveries in $R$, $|R\setminus N_\mathrm{P}| $, over all $R\in 2^{[m]}$.

Let $\boldsymbol{\Phi}$ be a family of local intersection tests.
\citet{genovese2006exceedance} and \citet{goeman2011multiple} showed that the
closed procedure $\mathrm{d}_{\boldsymbol{\Phi}}:2^{[m]}\to \mathbb{N}_0$ defined by
    \begin{align}
        \mathrm{d}_{\boldsymbol{\Phi}}(R)
                     \coloneqq \min\{|R\setminus S|:  S \in 2^{[m]}, \boldsymbol{\phi}_S=0\} \label{eq:closed_simult_FDP}
    \end{align}
    provides simultaneous true discovery guarantee.
We now show how this procedure can be reconstructed by our e-Closure Principle.  Let $\mathbf{E} = (\mathbf{e}_S)_{S \in 2^{[m]}}$ be an e-collection and $\F^d$, $d\in [m]_0 \coloneqq \{0\} \cup [m]$, be the error function defined by $\F_N^d(R)\coloneqq 1\{|R\setminus N|<d\}$. For a fixed $R\in 2^{[m]}$, we define $\mathrm{d}_\mathbf{E}$ as the largest $d\in [m]_0$ such that $R$ is contained in $\Rcal^{\ER_{\F^d}}_\alpha (\mathbf{E})$;
    \begin{align}
        \mathrm{d}_\mathbf{E}(R)\coloneqq \max\left\{d\in [m]_0: \mathbf{e}_S \geq \frac{\F_S^d(R)}{\alpha} \quad \forall S\in 2^{[m]}\right\}. \label{eq:e-closed_simult_FDP}
    \end{align}

\begin{theorem}\label{thm:simult_fdp}

    The procedure $\mathrm{d}_\mathbf{E}$ defined by \eqref{eq:e-closed_simult_FDP} provides simultaneous true discovery guarantee,
    $$
    \pdist(\mathrm{d}_\mathbf{E}(R) \leq |R\setminus N_\mathrm{P}| \text{ for all }R\in 2^{[m]}) \geq 1 - \alpha \quad \forall \mathrm{P} \in M.
    $$
    Furthermore, if $\mathbf{e}_S = \boldsymbol{\phi}_S/\alpha$, $S\in 2^{[m]}$, for some family of level $\alpha$ local tests $\boldsymbol{\Phi}$, then $\mathrm{d}_\mathbf{E}(R)=\mathrm{d}_{\boldsymbol{\Phi}}(R)$ for all $R\in 2^{[m]}$, where $\mathrm{d}_{\boldsymbol{\Phi}}$ is given by \eqref{eq:closed_simult_FDP}.
\end{theorem}
Since \citet{goeman2021only} proved that every admissible procedure with simultaneous true discovery guarantee can be constructed by \eqref{eq:closed_simult_FDP}, \Cref{thm:simult_fdp} immediately implies that the same also holds for our procedure \eqref{eq:e-closed_simult_FDP} based on the e-Closure Principle.

\begin{remark}
    For particular choices of uncountable $\mathcal{F}$, the expectation in \eqref{eq:sup-fcond} might not be well-defined. One could avoid this by modifying the definition of simultaneous error rate control to require that  $${\sup_{\F\in \mathcal{F}} \max_{R \in \Rcal^{\ER_\F}(x)}\  \F_{S}(R)\leq U_{S}(x)} \quad \text{for all possible data }x \text{ and for each }S\subseteq [m], $$ where $\boldsymbol{U}_S=U_S(X)$, $S\subseteq [m]$, are nonnegative random variables with $\expect_{\pdist}\left[ \boldsymbol{U}_{N_{\pdist}} \right]\leq \alpha$. Alternatively, one could make weak measure-theoretic assumptions on the error functions and data (see, e.g., \citet{chugg2025admissibility}). With this, it is also possible to extend the equivalence of simultaneous and post hoc control (Lemma~\ref{lemma:simult_posthoc}) to error functions. However, such technical subtleties are beyond the scope of this work.
\end{remark}

\subsection{Post hoc choice of nominal level} \label{sec: alpha}
Theorem \ref{thm: FWER} gives researchers the exciting option of switching from one error rate to another after seeing the data. In this section, we shall see that, for certain e-Closure procedures, there is the additional option of adapting the nominal error rate $\alpha$ to the data. To achieve this, we will extend \Cref{def: simultaneous ER} for simultaneous error rate control further, building on the post hoc validity definitions of \cite{grunwald_neyman-pearson_e-values_2024} and \cite{koning2023post}. This lets us select the \emph{nominal level} $\alpha >0$ in a post hoc fashion and retain a validity guarantee.

\begin{definition}[Simultaneous $\ER_\F$ and $\alpha$ control] \label{def: simultaneous ER alpha}
    Let $\mathcal{F}$ be a family of loss functions. For every $\alpha >0$ and $\F \in \mathcal{F}$, let $\boldsymbol{\mathcal{R}}_\alpha^{\ER_\F} \subseteq 2^{[m]}$ be a simultaneous discovery set. Then $\rejcol^{\mathcal{F}} \coloneqq (\boldsymbol{\mathcal{R}}_\alpha^{\ER_\F})_{\F \in \mathcal{F}, \alpha >0}$, ensures simultaneous control over the family of error functions $\mathcal{F}$ and nominal levels $\alpha>0$ if, for every $\mathrm{P} \in M$,
\[
\expect_\mathrm{P}\bigg( \sup_{\alpha >0} \sup_{\F \in \mathcal{F}} \max_{R \in \boldsymbol{\mathcal{R}}_\alpha^{\ER_\F}} \frac{\F_{N_\pdist}(R)}{\alpha}\bigg)   \leq 1.
\] where we let $0 / 0 = 0$ and $x / 0 = \infty$ for any positive $x$.
\end{definition}
\revise{For uncountable loss families or level sets, we assume the displayed supremum is measurable (for example through joint measurability and separability). 
}

Where $\ER_\F$ control in the simultaneous sense according to \Cref{def: simultaneous ER} allows the researcher to choose $R \in \boldsymbol{\mathcal{R}}_\alpha$ freely for a pre-chosen $\alpha$, $\ER_\F$ control according to \Cref{def: simultaneous ER alpha} implies that the researcher may choose freely a post hoc $R \in \boldsymbol{\Rcal_\alpha^{\ER_\F}}$ for each value of $\alpha$ (and for each choice of $\F$). Note that we do not directly bound the expectation of the loss by $\alpha$ anymore in the post hoc definition --- however, when $\alpha$ is fixed to be a single value, the standard error control of \Cref{def: simultaneous ER} is equivalent to the post hoc definition in \Cref{def: simultaneous ER alpha}. Furthermore, FDR control with post hoc $\alpha$ implies FDR control with data-dependent $\boldsymbol{\alpha}$.
\begin{corollary}
    Say we have a collection of simultaneous discovery sets $(\rejcol^{\ER_\F}_\alpha)_{\F \in \mathcal{F}, \alpha >0}$. \revise{Let $\mathbf{f}$ be a data-dependent loss function in $\mathcal{F}$, let $\boldsymbol{\alpha}>0$ be a data-dependent nominal level, and let $\mathbf R$ be a measurable selector satisfying $\mathbf R\in\rejcol^{\ER_{\mathbf f}}_{\boldsymbol\alpha}$ almost surely. These objects can be arbitrarily dependent with each other and with $(\rejcol^{\ER_\F}_\alpha)_{\F \in \mathcal{F}, \alpha >0}$. If $(\rejcol^{\ER_\F}_\alpha)_{\F \in \mathcal{F}, \alpha >0}$ has simultaneous control in the sense of \Cref{def: simultaneous ER alpha}, then, for all $\pdist \in M$,}
    \begin{align}
        \revise{\expect_\pdist\left(\frac{\mathbf{f}_{N_\pdist}(\mathbf R)}{\boldsymbol{\alpha}}\right) \leq 1.}
    \end{align} 
\end{corollary}

FDR control with post hoc $\alpha$, though seemingly more complicated, follows more or less directly from the e-Closure Principle. It only adds the additional requirement that the e-collection $\mathbf{E}$ does not depend on $\alpha$, since the resulting objects would no longer be valid e-values. \revise{For example, for our $\cBY$ calibrator in \eqref{eq: by calibrator}, one must fix a single $\alpha$ level in the calibrator, separate from the post hoc nominal level $\boldsymbol{\alpha}$. We cannot allow the calibrator's $\alpha$ level to vary in a data-dependent fashion, since the resulting random variable would no longer be a valid e-value.} The proof is similar to that of \Cref{thm: e-closure}. \Cref{thm: alpha} extends the corresponding result of \citet[Theorem 9.10]{ramdas_hypothesis_testing_2025} that $\alpha$ can be chosen post hoc in eBH to general e-Closure procedures and simultaneous FDR control.

\begin{theorem}[The e-Closure Principle for general simultaneous error control] \label{thm: alpha}
Suppose $\mathbf{E}$ does not depend on $\alpha$. Then $(\Rcal_\alpha^\F(\mathbf{E}))_{\F \in \mathcal{F}, \alpha >0}$ guarantees simultaneous control over $ \mathcal{F}$ and $\alpha >0$ (\Cref{def: simultaneous ER alpha}).
Furthermore for every collection $(\rejcol_\alpha^\F)_{\F \in \mathcal{F}, \alpha >0}$ with simultaneous control, there exists an e-collection $\mathbf{E}$ (that does not depend on $\alpha$) such that $\rejcol^{\ER_\F}_\alpha \subseteq \Rcal^{\ER_\F}_\alpha(\mathbf{E})$ for all $\F \in \mathcal{F}$ and $\alpha >0$.
\end{theorem}

\begin{proof}
Choose any $\mathrm{P} \in M$. Then,
$$
\expect_\mathrm{P}\left( \sup_{\alpha >0} \sup_{\F \in \mathcal{F}} \max_{R \in \Rcal^{\ER_\F}_\alpha(\mathbf{E})} \frac{\F_{N_\pdist}(R)}{\alpha}\right) \leq
\expect_\mathrm{P} (\mathbf{e}_{N_\mathrm{P}}) \leq 1.
$$
This shows the validity of the post hoc controlling collection defined from an e-collection.
Now, we can define the following e-values for each $S\subseteq [m]$ from any $\rejcol^\mathcal{F}$:
\begin{align}
    \evalue_S = \sup_{\alpha >0} \sup_{\F \in \mathcal{F}} \max_{R \in \rejcol_\alpha^{\ER_\F}} \frac{\F_{S}(R)}{\alpha}. \label{eq:post-hoc-closed-evalue}
\end{align}

We can then use these e-values to define a post hoc valid procedure again as the following application of the e-Closure procedure: $\Rcal^{\mathcal{F}}(\mathbf{E}) \coloneqq (\Rcal^{\ER_\F}_\alpha(\mathbf{E}))_{\F \in \mathcal{F}, \alpha >0}$.
The e-collection with e-values defined in \eqref{eq:post-hoc-closed-evalue} is valid via the fact that $\rejcol^\mathcal{F}$ is post hoc valid as defined in \Cref{def: simultaneous ER alpha}. Then by definition of $\Rcal^{\ER_\F}_\alpha(\mathbf{E})$, we know for any $R \in \rejcol^{\ER_\F}_\alpha$, we have $R \in \Rcal^{\ER_\F}_\alpha(\mathbf{E})$ as well, and we get our desired dominance result.
\end{proof}
\begin{remark}
\revise{\Cref{thm: alpha} could still be used with methods such as $\cBY$ whose calibrator depends on a fixed level: denote that calibration level by $\alpha_0$, and keep it fixed while the nominal error-control level $\alpha$ is chosen post hoc.}
\end{remark}
Thus, in this section, we have shown that the e-Closure Principle not only increases power but also allows greater variety in the choice of discovery set and error control guarantees.

\section{Restricted combinations} \label{sec: Shaffer}

Thus far, we have always constructed an e-collection simply by constructing e-values for intersection hypotheses $H_S$. This is sufficient for making an e-collection, but not necessary (see \Cref{rem: partitoning}). We can instead make e-values for all partitioning hypotheses $\bar{H}_S \coloneqq \{\pdist \in M: N_\pdist = S\}$ \citep{finner2002partitioning}. If each $\mathbf{e}_S$ is an e-value for $\bar{H}_S$, then in particular $\mathbf{e}_{N_\pdist}$ is a valid e-value, since $\bar{H}_{N_\pdist}$ is true for every $\pdist$. Since $\bar{H}_S \subseteq H_S$, there is a potential power gain in switching from $H_S$ to $\bar{H}_S$.

One situation in which this power gain might materialize is in the context of restricted combinations
\citep{shaffer1986modified}.
For example, suppose we have $m=3$ and $H_1\colon \theta_A=\theta_B$, $H_2\colon \theta_A=\theta_C$ and $H_3\colon \theta_B=\theta_C$. Then we have $H_{\{1,3\}}: \theta_A=\theta_B=\theta_C$, but
$\bar H_{\{1,3\}} = \{\mathrm{P} \in M\colon \theta_A=\theta_B=\theta_C  \textrm{ and } \theta_A \neq \theta_C\} = \emptyset$. \revise{Consequently, validity imposes no constraint on $\mathbf e_{\{1,3\}}$: its e-Closure constraint may be omitted, or one may use an arbitrarily large constant. Another way to look at this is that there is no $\pdist$ for which $N_\pdist = \{1,3\}$.}

Taking restricted combinations into account can therefore lead to substantial gains in power. \revise{For example, if we test the $m=6$ pairwise comparisons among four parameters $\theta_A, \theta_B, \theta_C, \theta_D$, attainable true null patterns correspond to the $B_4=15$ equivalence relations on four labels. Thus, among all $2^6=64$ subsets, 15 partitioning hypotheses are nonempty and 49 are empty; excluding the empty true null set leaves 14 attainable and 49 impossible subsets among the 63 nonempty subsets. Suppose that the hypothesis $H_1: \theta_A=\theta_B$ has e-value $4\alpha^{-1}$, $H_2: \theta_A=\theta_C$ and $H_3: \theta_A=\theta_D$ both have e-values of $\alpha^{-1}$, while the other three hypotheses have e-values of 0. Then $\CeBH$ would only reject $\mathbf{R} = \{1\}$ if restricted combinations were ignored, but could additionally reject $\mathbf{R} = \{1,2,3\}$ (and all its subsets) when they are used.}
\begin{remark}
    \revise{When $\F_N(R)$ is nondecreasing in $N$ under set inclusion, every simultaneously valid procedure admits a reconstructing e-collection whose $\evalue_S$ is valid for the intersection hypothesis $H_S$. Thus intersection valid local e-values suffice for completeness. Partition valid local e-values can nevertheless be strictly more powerful when some null configurations are impossible. Without monotonicity of the loss, validity of $\evalue_{N_\pdist}$ is the necessary and sufficient condition.} \label{rem: partitoning}
\end{remark}
\section{Computation} \label{sec: computation}
For an arbitrary e-collection $\mathbf{E}$, checking whether $R \in \Rcal_\alpha^\FDR(\mathbf{E})$ involves computing and using exponentially many e-values. However, in special cases, most notably for $\CeBH$, $\textcBY$ and $\textcSu$ methods proposed above, computation reduces to polynomial time.

It is easy to check that the local e-values of $\textcBY$ and $\textcSu$ have the property that they are weakly decreasing in the per-hypothesis p-values. We can exploit this in the spirit of the shortcuts for closed testing by \citet{dobriban2020fast}. To check for a set $R$ whether $\alpha\mathbf{e}_S \geq |R \cap S|/|R|$, it suffices to check, for each $1\leq a \leq |R|$ and each $0 \leq b \leq m-|R|$, whether the condition holds for the set $S$ consisting of the indices of $a$ largest p-values in $R$ and of the $b$ largest p-values in $[m]\setminus R$. This takes $O(m^2)$ evaluations of $\mathbf{e}_{S}$. Then, finding the largest set $[r]$ that is rejected by the method takes $O(m^3)$ such evaluations, or $O(m^4)$ time. Precalculation of sums of e-values can reduce this to $O(m^3)$ for $\cBY$.

For $\CeBH$ we can do the calculations even faster. We will illustrate this for sets $R$ of the form $[r]$. Define
\[
\mathbf{s}_k \coloneqq \sum_{i=1}^k \mathbf{e}_i, \text{ and }\mathbf{g}(a,r,b)  \coloneqq \mathbf{s}_m - \mathbf{s}_b + \mathbf{s}_r - \mathbf{s}_a - \frac{(m-b+r-a)(r-a)}{r\alpha} .
\]
We have $\alpha\mathbf{e}_S \geq |R \cap S|/|R|$, for all $S \in 2^{[m]}$, if and only if, for all $a \in \{0, \dots, r - 1\}$ and $b \in \{r, \dots, m\}$, $\mathbf{g}(a,r,b) \geq 0$.
It is easily checked that $\mathbf{g}(a,r,b)$ is convex in $b$. Therefore, for fixed $r$ and $a$, the minimum of $\mathbf{g}$ can be found in $O(\log m) $ time. We can therefore check whether $R \in \Rcal_\alpha^\FDR(\mathbf{E})$ in $O(m\log m)$ time, and find the largest such $R$ in $O(m^2\log m)$ time by trying all values of $r$ from $m$ downward. In practice, if $r$ is substantially larger than the size of the largest rejected $R$, one quickly finds $a,b$ for which $\mathbf{g}(a,r,b)<0$, so that computation time for finding the largest rejected $R$ is often closer to $m\log m$ in practice. Checking whether $R\in \Rcal_\alpha^\FDR(\mathbf{E})$ for sets $R$ not of the form $[r]$ can also be done in $O(m\log m)$ time following essentially the same reasoning as for sets of the form $[r]$, adapting the definition of $\mathbf{g}$ as appropriate.

\revise{The $\CeBH$, $\cBY$ and $\cSu$ methods have been efficiently implemented in the \texttt{eClosure} R package, available on CRAN, and the \texttt{eclosure} Python package on PyPI.}

\section{Additional real data applications} \label{sec: real data}

To better understand the relationship between $\cBY$ and other possible FDR controlling procedures, we produce in \Cref{tab:by_real_data_discoveries_full} a more comprehensive version of \Cref{tab:by_real_data_discoveries} that also includes the number of discoveries made by the standard BH procedure and the $\CeBH$ procedure applied to p-values calibrated via \eqref{eq: by calibrator}, which we denote by ``'$\CeBH$ + cal.''. We can see that the BH procedure still makes many more discoveries than all other procedures. Further, we do see that $\cBY$ practically is more powerful than $\CeBH$ with calibration.

\begin{table}[htbp]\ifarver{}{\fontsize{10}{12}\selectfont}
\caption{\revise{Number of discoveries made by existing procedures, \citet{benjamini_control_false_2001} (BY) and \citet{benjamini_controlling_false_1995} (BH), and by the e-Closure-derived procedures $\cBY$ and $\CeBH$ applied to p-values calibrated via \eqref{eq: by calibrator}, for the same datasets as \Cref{tab:by_real_data_discoveries}. The BH procedure is used as a benchmark here, as it is invalid under arbitrary dependence (unlike the 3 other procedures), and requires strong assumptions on positive dependence among the p-values.}
\label{tab:by_real_data_discoveries_full}
}
    \centering
\begin{tabular}{lccccccccc}
\toprule
Dataset & $m$ & \multicolumn{4}{c}{$\alpha = 0.05$} & \multicolumn{4}{c}{$\alpha = 0.1$} \\ 
\cmidrule(lr){3-6} \cmidrule(lr){7-10}
 &  & BY & $\CeBH$ + cal. & $\cBY$ & BH & BY & $\CeBH$ + cal. & $\cBY$ & BH \\
\midrule
APSAC & 15 & 3 & 3 & 3 & 4 & 3 & 4 & 5 & 9 \\
NAEP & 34 & 6 & 7 & 8 & 11 & 8 & 11 & 11 & 12 \\
PADJUST & 50 & 12 & 14 & 15 & 20 & 17 & 20 & 20 & 21 \\
PVALUES & 4,289 & 129 & 144 & 145 & 767 & 225 & 270 & 275 & 1,139 \\
VANDEVIJVER & 4,919 & 614 & 672 & 677 & 1,340 & 779 & 859 & 866 & 1,728 \\
GOLUB & 7,128 & 617 & 640 & 648 & 1,249 & 743 & 791 & 799 & 1,605 \\
\bottomrule
\end{tabular}
\end{table}
In addition to the results for BY and $\cBY$ in \Cref{tab:by_real_data_discoveries}, we also compare the eBH procedure and the $\CeBH$ procedure on a real cryptocurrency dataset.
To evaluate the efficacy of eBH, \citet{wang_false_discovery_2022} introduced a cryptocurrency dataset and e-values for monthly coin returns. The return data ranged from 2015 to 2021, and the e-values were constructed to test whether each coin's return over the period was positive. Three scenarios are considered, corresponding to the coins with the 426 (all coins), 200, and 100 highest market capitalizations. Three strategies are also considered for generating the e-values, corresponding to three investment strategies: (1) Buy and Hold, (2) Rebalance 30\%, and (3) Rebalance 50\%.
The resulting number of discoveries for each procedure is shown in \Cref{tab:ebh_real_data_discoveries}. Thus, we can see that $\CeBH$ always makes at least as many discoveries as eBH (as guaranteed), and often substantially more, especially at higher values of $\alpha$.

Our results on real datasets empirically confirm the theoretical improvements of these closed procedures over their respective standard counterparts.
\begin{table}[htbp]\ifarver{}{\fontsize{10}{12}\selectfont}
\centering
\caption{Number of discoveries for eBH and $\CeBH$ made on monthly return data of cryptocurrency coins introduced in \citet{wang_false_discovery_2022}.}
\label{tab:ebh_real_data_discoveries}
\begin{tabular}{lcccccc}
\toprule
Strategy & $m$ & \multicolumn{2}{c}{$\alpha = 0.05$} & \multicolumn{2}{c}{$\alpha = 0.2$} \\
\cmidrule(lr){3-4} \cmidrule(lr){5-6}
&  & eBH & $\CeBH$ & eBH & $\CeBH$ \\
\midrule
Buy and Hold & 100 & 0 & 1 & 19 & 34 \\
& 200 & 4 & 7 & 21 & 52 \\
& 426 & 3 & 9 & 18 & 68 \\
Rebalance 30\% & 100 & 6 & 11 & 21 & 64 \\
& 200 & 12 & 28 & 44 & 105 \\
& 426 & 21 & 34 & 55 & 124 \\
Rebalance 50\% & 100 & 13 & 29 & 40 & 65 \\
& 200 & 28 & 55 & 65 & 108 \\
& 426 & 42 & 62 & 83 & 140 \\
\bottomrule
\end{tabular}
\end{table}
\section{Discussion}
The e-Closure Principle resolves the disparity between FWER and FDP tail probability methods on one side and FDR control methods on the other. It translates the Closure Principle for these methods directly to the FDR context, bringing many of its benefits into the realm of FDR.

The e-Closure Principle can be used to propose new methods and improve existing ones. In this paper, we have mostly focused on improving existing methods, in order to showcase the power of the principle. However, we believe that its greatest strength is in the development of novel methods. To develop an FDR control method, a researcher only needs to decide how to aggregate the evidence against a partitioning or intersection hypothesis into an e-value. After making that choice, the only remaining problem to be solved is computational. \revise{Designing how to aggregate such evidence into an e-value can require extra work. For example, when one already has access to an e-value for each hypothesis, the e-Closure Principle immediately provides a method for multiple testing that is robust to arbitrary dependence by taking the average of the e-values in each intersection hypothesis. On the other hand, when one has per-hypothesis p-values, it may seem an attractive option to convert such p-values to per-hypothesis e-values and apply the e-Closure Principle. However, for the p-value based methods $\textcBY$ and $\textcSu$, we have seen further improvement by deriving more powerful local e-values directly for the intersection hypotheses.}

The e-Closure Principle also brings post hoc flexibility to FDR control on a scale that was previously only known in FDP control. Rather than only a single rejection set, researchers have many rejection sets to choose from, and they may use all the data to decide which one to report, while still retaining FDR control. If signal is very strong, an attractive option is to switch to FWER control post hoc, or for methods in which the e-values do not depend on $\alpha$, to adjust the target FDR level to match the amount of signal in the data. Finally, e-Closure is not restricted to the classical error rate, but allows methods to be formulated for novel error rates, as long as they can be written as expected loss.

There are several ways in which this work may be extended. We assumed a finite number of hypotheses, and the e-Closure Principle could also be used to derive procedures for infinite hypotheses (i.e., in online multiple testing).There are also many more FDR controlling methods than we have considered in this paper, but we believe that some of them can be improved, or at least made simultaneous, using e-Closure. Finding polynomial time algorithms for some such improvements will be challenging, but the literature on advanced shortcuts in closed testing may help. For example, the branch and bound algorithm may be as useful in e-Closure as it is in closed testing \citep{vesely2023permutation}.
 
\paragraph{Acknowledgments} We thank Aur\`ele Mingam for pointing out an error with our theorem relating to appyling e-Closure to BH variants and correcting our proof.
\paragraph{Declaration of funding} Rianne de Heide's work was supported by funding from NWO Veni grant number VI.Veni.222.018. Lasse Fischer acknowledges funding by the Deutsche Forschungsgemeinschaft (DFG, German Research Foundation) – Project number 281474342/GRK2224/2. Aaditya Ramdas acknowledges funding from NSF grant DMS-2310718. \bibliographystyle{plainnat}

\bibliography{ref}

\begin{thebibliography}{94}
\providecommand{\natexlab}[1]{#1}
\providecommand{\url}[1]{\texttt{#1}}
\expandafter\ifx\csname urlstyle\endcsname\relax
  \providecommand{\doi}[1]{doi: #1}\else
  \providecommand{\doi}{doi: \begingroup \urlstyle{rm}\Url}\fi

\bibitem[Barber and Candès(2015{\natexlab{a}})]{barber_controlling_2015}
Rina~Foygel Barber and Emmanuel~J. Candès.
\newblock Controlling the false discovery rate via knockoffs.
\newblock \emph{The Annals of Statistics}, 43\penalty0 (5):\penalty0
  2055--2085, 2015{\natexlab{a}}.

\bibitem[Barber and Candès(2015{\natexlab{b}})]{barber_supplement_2015}
Rina~Foygel Barber and Emmanuel~J. Candès.
\newblock Supplement to ``controlling the false discovery rate via knockoffs''.
\newblock Supplementary material for The Annals of Statistics,
  2015{\natexlab{b}}.

\bibitem[Benjamini and Hochberg(1995)]{benjamini_controlling_false_1995}
Yoav Benjamini and Yosef Hochberg.
\newblock Controlling the false discovery rate: A practical and powerful
  approach to multiple testing.
\newblock \emph{Journal of the Royal Statistical Society. Series B
  (Methodological)}, 57\penalty0 (1):\penalty0 289--300, 1995.

\bibitem[Benjamini and Hochberg(2000)]{benjamini2000adaptive}
Yoav Benjamini and Yosef Hochberg.
\newblock On the adaptive control of the false discovery rate in multiple
  testing with independent statistics.
\newblock \emph{Journal of Educational and Behavioral Statistics}, 25\penalty0
  (1):\penalty0 60--83, 2000.

\bibitem[Benjamini and Yekutieli(2001)]{benjamini_control_false_2001}
Yoav Benjamini and Daniel Yekutieli.
\newblock The control of the false discovery rate in multiple testing under
  dependency.
\newblock \emph{The Annals of Statistics}, 29\penalty0 (4):\penalty0
  1165--1188, 2001.

\bibitem[Benjamini et~al.(2006)Benjamini, Krieger, and
  Yekutieli]{benjamini2006adaptive}
Yoav Benjamini, Abba~M Krieger, and Daniel Yekutieli.
\newblock Adaptive linear step-up procedures that control the false discovery
  rate.
\newblock \emph{Biometrika}, 93\penalty0 (3):\penalty0 491--507, 2006.

\bibitem[Blanchard and Roquain(2008)]{blanchard_two_simple_2008}
Gilles Blanchard and Etienne Roquain.
\newblock Two simple sufficient conditions for {{FDR}} control.
\newblock \emph{Electronic Journal of Statistics}, 2:\penalty0 963--992, 2008.

\bibitem[Blanchard and Roquain(2009)]{blanchard2009adaptive}
Gilles Blanchard and Etienne Roquain.
\newblock Adaptive false discovery rate control under independence and
  dependence.
\newblock \emph{Journal of Machine Learning Research}, 10\penalty0 (12), 2009.

\bibitem[{Blier-Wong} and Wang(2024)]{blier-wong_improved_thresholds_2024}
Christopher {Blier-Wong} and Ruodu Wang.
\newblock Improved thresholds for e-values.
\newblock arXiv:2408.11307, 2024.

\bibitem[Bretz et~al.(2009)Bretz, Maurer, Brannath, and
  Posch]{bretz_graphical_approach_2009}
Frank Bretz, Willi Maurer, Werner Brannath, and Martin Posch.
\newblock A graphical approach to sequentially rejective multiple test
  procedures.
\newblock \emph{Statistics in Medicine}, 28\penalty0 (4):\penalty0 586--604,
  2009.

\bibitem[Cand{\`e}s et~al.(2018)Cand{\`e}s, Fan, Janson, and
  Lv]{candes_panning_gold_2018}
Emmanuel Cand{\`e}s, Yingying Fan, Lucas Janson, and Jinchi Lv.
\newblock Panning for gold: `{{Model-X}}' knockoffs for high dimensional
  controlled variable selection.
\newblock \emph{Journal of the Royal Statistical Society: Series B (Statistical
  Methodology)}, 80\penalty0 (3):\penalty0 551--577, 2018.

\bibitem[Chugg et~al.(2025)Chugg, Lardy, Ramdas, and
  Gr{\"u}nwald]{chugg2025admissibility}
Ben Chugg, Tyron Lardy, Aaditya Ramdas, and Peter Gr{\"u}nwald.
\newblock On admissibility in post-hoc hypothesis testing.
\newblock arXiv:2508.00770, 2025.

\bibitem[Clerico(2026)]{clerico_simple_geometric_2026}
Eugenio Clerico.
\newblock A simple geometric proof for the characterisation of {E}-merging
  functions.
\newblock \emph{Statistics \& Probability Letters}, 236:\penalty0 110750, 2026.
\newblock \doi{10.1016/j.spl.2026.110750}.

\bibitem[Cui et~al.(2021)Cui, Dickhaus, Ding, and
  Hsu]{cui_handbook_multiple_2021}
Xinping Cui, Thorsten Dickhaus, Ying Ding, and Jason~C. Hsu.
\newblock \emph{Handbook of {{Multiple Comparisons}}}.
\newblock CRC Press, 2021.

\bibitem[Dandapanthula and Ramdas(2025)]{dandapanthula_multiple_testing_2025}
Sanjit Dandapanthula and Aaditya Ramdas.
\newblock Multiple testing in multi-stream sequential change detection.
\newblock arXiv:2501.04130, 2025.

\bibitem[Delattre and Roquain(2015)]{delattre_roquain_new_2015}
Sylvain Delattre and Etienne Roquain.
\newblock New procedures controlling the false discovery proportion via
  {Romano--Wolf's} heuristic.
\newblock \emph{The Annals of Statistics}, 43\penalty0 (3):\penalty0
  1141--1177, 2015.

\bibitem[Dobriban(2020)]{dobriban2020fast}
Edgar Dobriban.
\newblock Fast closed testing for exchangeable local tests.
\newblock \emph{Biometrika}, 107\penalty0 (3):\penalty0 761--768, 2020.

\bibitem[D{\"o}hler and Meah(2023)]{dohler_unified_class_2023}
Sebastian D{\"o}hler and Iqraa Meah.
\newblock A unified class of null proportion estimators with plug-in {{FDR}}
  control.
\newblock arXiv:2307.13557, 2023.

\bibitem[Ebrahimpoor and Goeman(2021)]{ebrahimpoor2021inflated}
Mitra Ebrahimpoor and Jelle~J Goeman.
\newblock Inflated false discovery rate due to volcano plots: problem and
  solutions.
\newblock \emph{Briefings in Bioinformatics}, 22\penalty0 (5), 2021.

\bibitem[Ebrahimpoor et~al.(2020)Ebrahimpoor, Spitali, Hettne, Tsonaka, and
  Goeman]{ebrahimpoor2020simultaneous}
Mitra Ebrahimpoor, Pietro Spitali, Kristina Hettne, Roula Tsonaka, and Jelle
  Goeman.
\newblock Simultaneous enrichment analysis of all possible gene-sets: unifying
  self-contained and competitive methods.
\newblock \emph{Briefings in Bioinformatics}, 21\penalty0 (4), 2020.

\bibitem[Efron(2010)]{efron_large-scale_inference_2010}
Bradley Efron.
\newblock \emph{Large-{{Scale Inference}}: {{Empirical Bayes Methods}} for
  {{Estimation}}, {{Testing}}, and {{Prediction}}}.
\newblock Institute of {{Mathematical Statistics Monographs}}. Cambridge
  University Press, Cambridge, 2010.

\bibitem[Efron and Hastie(2016)]{efron_computer_age_2021}
Bradley Efron and Trevor Hastie.
\newblock \emph{Computer {{Age Statistical Inference}}: {{Algorithms}},
  {{Evidence}}, and {{Data Science}}}.
\newblock Cambridge University Press, 2016.

\bibitem[Finner and Strassburger(2002)]{finner2002partitioning}
H~Finner and K~Strassburger.
\newblock The partitioning principle: a powerful tool in multiple decision
  theory.
\newblock \emph{The Annals of Statistics}, 30\penalty0 (4):\penalty0
  1194--1213, 2002.

\bibitem[Finner and Roters(2001)]{finner2001false}
Helmut Finner and Markus Roters.
\newblock On the false discovery rate and expected type {I} errors.
\newblock \emph{Biometrical Journal}, 43\penalty0 (8):\penalty0 985--1005,
  2001.

\bibitem[Fischer and Ramdas(2024)]{fischer2024admissible}
Lasse Fischer and Aaditya Ramdas.
\newblock Admissible online closed testing must employ e-values.
\newblock arXiv:2407.15733, 2024.

\bibitem[Fischer et~al.(2024)Fischer, Xu, and Ramdas]{fischer2024online}
Lasse Fischer, Ziyu Xu, and Aaditya Ramdas.
\newblock An online generalization of the (e-){Benjamini-Hochberg} procedure.
\newblock arXiv:2407.20683, 2024.

\bibitem[Foster and Stine(2008)]{foster_a-investing_procedure_2008}
Dean~P. Foster and Robert~A. Stine.
\newblock {$\alpha$}-investing: A procedure for sequential control of expected
  false discoveries.
\newblock \emph{Journal of the Royal Statistical Society: Series B (Statistical
  Methodology)}, 70\penalty0 (2):\penalty0 429--444, 2008.

\bibitem[Gao(2025)]{gao_adaptive_null_2025}
Zijun Gao.
\newblock An adaptive null proportion estimator for false discovery rate
  control.
\newblock \emph{Biometrika}, 112\penalty0 (1), 2025.

\bibitem[Genovese and Wasserman(2006)]{genovese2006exceedance}
Christopher~R Genovese and Larry Wasserman.
\newblock Exceedance control of the false discovery proportion.
\newblock \emph{Journal of the American Statistical Association}, 101\penalty0
  (476):\penalty0 1408--1417, 2006.

\bibitem[Goeman(2026)]{goeman_uniform_improvement_2026}
Jelle Goeman.
\newblock A uniform improvement of the {Benjamini-Hochberg} procedure via
  e-closure.
\newblock arXiv:2606.01854, 2026.

\bibitem[Goeman et~al.(2025)Goeman, de~Heide, and
  Solari]{goeman_e-partitioning_principle_2025}
Jelle Goeman, Rianne de~Heide, and Aldo Solari.
\newblock The e-partitioning principle of false discovery rate control.
\newblock arXiv:2504.15946, 2025.

\bibitem[Goeman and Solari(2011)]{goeman2011multiple}
Jelle~J Goeman and Aldo Solari.
\newblock Multiple testing for exploratory research.
\newblock \emph{Statistical Science}, 26\penalty0 (4):\penalty0 584--597, 2011.

\bibitem[Goeman and Solari(2014)]{goeman2014multiple}
Jelle~J Goeman and Aldo Solari.
\newblock Multiple hypothesis testing in genomics.
\newblock \emph{Statistics in Medicine}, 33\penalty0 (11):\penalty0 1946--1978,
  2014.

\bibitem[Goeman and Solari(2024)]{goeman2024selection}
Jelle~J Goeman and Aldo Solari.
\newblock On selection and conditioning in multiple testing and selective
  inference.
\newblock \emph{Biometrika}, 111\penalty0 (2):\penalty0 393--416, 2024.

\bibitem[Goeman et~al.(2021)Goeman, Hemerik, and Solari]{goeman2021only}
Jelle~J Goeman, Jesse Hemerik, and Aldo Solari.
\newblock Only closed testing procedures are admissible for controlling false
  discovery proportions.
\newblock \emph{The Annals of Statistics}, 49\penalty0 (2):\penalty0
  1218--1238, 2021.

\bibitem[Gr\"unwald et~al.(2024)Gr\"unwald, de~Heide, and
  Koolen]{grunwals2024safe}
Peter Gr\"unwald, Rianne de~Heide, and Wouter~M. Koolen.
\newblock Safe testing.
\newblock \emph{Journal of the Royal Statistical Society: Series B (Statistical
  Methodology)}, 86:\penalty0 1163–1171, 2024.

\bibitem[Gr{\"u}nwald(2023)]{grunwald_e-posterior_2023}
Peter~D. Gr{\"u}nwald.
\newblock The e-posterior.
\newblock \emph{Philosophical Transactions of the Royal Society A:
  Mathematical, Physical and Engineering Sciences}, 381\penalty0 (2247), 2023.

\bibitem[Gr{\"u}nwald(2024)]{grunwald_neyman-pearson_e-values_2024}
Peter~D. Gr{\"u}nwald.
\newblock Beyond {{Neyman}}--{{Pearson}}: {{E-values}} enable hypothesis
  testing with a data-driven alpha.
\newblock \emph{Proceedings of the National Academy of Sciences}, 121\penalty0
  (39), 2024.

\bibitem[Guo and Romano(2007)]{guo_romano_generalized_2007}
Wenge Guo and Joseph~P. Romano.
\newblock A generalized sidak--holm procedure and control of generalized error
  rates under independence.
\newblock \emph{Statistical Applications in Genetics and Molecular Biology},
  6\penalty0 (1):\penalty0 Article 3, 2007.

\bibitem[Guo et~al.(2014)Guo, He, and Sarkar]{guo_he_sarkar_fdp_2014}
Wenge Guo, Li~He, and Sanat~K. Sarkar.
\newblock Further results on controlling the false discovery proportion.
\newblock arXiv:1406.0266, 2014.

\bibitem[Hartog and Lei(2025)]{hartog_family-wise_error_2025}
Will Hartog and Lihua Lei.
\newblock Family-wise error rate control with e-values.
\newblock arXiv:2501.09015, 2025.

\bibitem[Hemerik(2025)]{hemerik_multiresolution_2025}
Jesse Hemerik.
\newblock Resampling-based multi-resolution false discovery exceedance control.
\newblock arXiv:2509.02376, 2025.

\bibitem[Hemerik et~al.(2024)Hemerik, Solari, and Goeman]{hemerik}
Jesse Hemerik, Aldo Solari, and Jelle~J. Goeman.
\newblock Flexible control of the median of the false discovery proportion.
\newblock \emph{Biometrika}, 111\penalty0 (4):\penalty0 1129--1150, 2024.

\bibitem[Henning and Westfall(2015)]{henning_closed_testing_2015}
Kevin S.~S. Henning and Peter~H. Westfall.
\newblock Closed testing in pharmaceutical research: Historical and recent
  developments.
\newblock \emph{Statistics in Biopharmaceutical Research}, 7\penalty0
  (2):\penalty0 126--147, 2015.

\bibitem[Holm(1979)]{holm_simple_sequentially_1979}
Sture Holm.
\newblock A {{Simple Sequentially Rejective Multiple Test Procedure}}.
\newblock \emph{Scandinavian Journal of Statistics}, 6\penalty0 (2):\penalty0
  65--70, 1979.

\bibitem[Ignatiadis et~al.(2024)Ignatiadis, Wang, and
  Ramdas]{ignatiadis2024values}
Nikolaos Ignatiadis, Ruodu Wang, and Aaditya Ramdas.
\newblock E-values as unnormalized weights in multiple testing.
\newblock \emph{Biometrika}, 111\penalty0 (2):\penalty0 417--439, 2024.

\bibitem[Ignatiadis et~al.(2025)Ignatiadis, Wang, and
  Ramdas]{ignatiadis_asymptotic_compound_2025}
Nikolaos Ignatiadis, Ruodu Wang, and Aaditya Ramdas.
\newblock Asymptotic and compound e-values: Multiple testing and empirical
  {{Bayes}}.
\newblock arXiv:2409.19812, 2025.

\bibitem[Katsevich and Ramdas(2020)]{katsevich2020simultaneous}
Eugene Katsevich and Aaditya Ramdas.
\newblock Simultaneous high-probability bounds on the false discovery
  proportion in structured, regression and online settings.
\newblock \emph{The Annals of Statistics}, 48\penalty0 (6):\penalty0
  3465--3487, 2020.

\bibitem[Katsevich et~al.(2023)Katsevich, Sabatti, and
  Bogomolov]{katsevich2023filtering}
Eugene Katsevich, Chiara Sabatti, and Marina Bogomolov.
\newblock Filtering the rejection set while preserving false discovery rate
  control.
\newblock \emph{Journal of the American Statistical Association}, 118\penalty0
  (541):\penalty0 165--176, 2023.

\bibitem[Kechris(1995)]{kechris_classical_1995}
Alexander~S. Kechris.
\newblock \emph{Classical Descriptive Set Theory}, volume 156 of \emph{Graduate
  Texts in Mathematics}.
\newblock Springer-Verlag, New York, 1995.

\bibitem[Koning(2023)]{koning2023post}
Nick~W Koning.
\newblock Post-hoc $\alpha $ hypothesis testing and the post-hoc $p$-value.
\newblock arXiv:2312.08040, 2023.

\bibitem[Korn et~al.(2004)Korn, Troendle, McShane, and
  Simon]{korn_controlling_number_2004}
Edward~L. Korn, James~F. Troendle, Lisa~M. McShane, and Richard Simon.
\newblock Controlling the number of false discoveries: Application to
  high-dimensional genomic data.
\newblock \emph{Journal of Statistical Planning and Inference}, 124\penalty0
  (2):\penalty0 379--398, 2004.

\bibitem[Lee and Ren(2024)]{lee_boosting_e-bh_2024a}
Junu Lee and Zhimei Ren.
\newblock Boosting e-{{BH}} via conditional calibration.
\newblock arXiv:2404.17562, 2024.

\bibitem[Lee et~al.(2025)Lee, Popov, and Ren]{lee_full-conformal_novelty_2025}
Junu Lee, Ilia Popov, and Zhimei Ren.
\newblock Full-conformal novelty detection: {{A}} powerful and non-random
  approach.
\newblock arXiv:2501.02703, 2025.

\bibitem[Lee and Ren(2025)]{lee_selection_hierarchical_2025}
Yonghoon Lee and Zhimei Ren.
\newblock Selection from {{Hierarchical Data}} with {{Conformal E-values}}.
\newblock arXiv:2501.02514, 2025.

\bibitem[Lehmann and
  Romano(2005{\natexlab{a}})]{lehmann_romano_generalizations_2005}
E.~L. Lehmann and Joseph~P. Romano.
\newblock Generalizations of the familywise error rate.
\newblock \emph{The Annals of Statistics}, 33\penalty0 (3):\penalty0
  1138--1154, 2005{\natexlab{a}}.

\bibitem[Lehmann and Romano(2005{\natexlab{b}})]{lehmann2005testing}
Erich~Leo Lehmann and Joseph~P. Romano.
\newblock \emph{Testing Statistical Hypotheses}.
\newblock Springer, 2005{\natexlab{b}}.

\bibitem[Lei and Fithian(2018)]{lei2018adapt}
Lihua Lei and William Fithian.
\newblock {AdaPT}: an interactive procedure for multiple testing with side
  information.
\newblock \emph{Journal of the Royal Statistical Society: Series B (Statistical
  Methodology)}, 80\penalty0 (4):\penalty0 649--679, 2018.

\bibitem[Lei et~al.(2021)Lei, Ramdas, and Fithian]{lei2021general}
Lihua Lei, Aaditya Ramdas, and William Fithian.
\newblock A general interactive framework for false discovery rate control
  under structural constraints.
\newblock \emph{Biometrika}, 108\penalty0 (2):\penalty0 253--267, 2021.

\bibitem[Li and Zhang(2025)]{li_note_e-values_2025}
Guanxun Li and Xianyang Zhang.
\newblock A note on e-values and multiple testing.
\newblock \emph{Biometrika}, 112\penalty0 (1), 2025.

\bibitem[Malek et~al.(2017)Malek, Katariya, Chow, and
  Ghavamzadeh]{malek_sequential_multiple_2017a}
Alan Malek, Sumeet Katariya, Yinlam Chow, and Mohammad Ghavamzadeh.
\newblock Sequential multiple hypothesis testing with type {{I}} error control.
\newblock In \emph{{{International Conference}} on {{Artificial Intelligence}}
  and {{Statistics}}}, 2017.

\bibitem[Marcus et~al.(1976)Marcus, Peritz, and Gabriel]{marcus1976closed}
Ruth Marcus, Eric Peritz, and K.~R. Gabriel.
\newblock On closed testing procedures with special reference to ordered
  analysis of variance.
\newblock \emph{Biometrika}, 63\penalty0 (3):\penalty0 655--660, 1976.

\bibitem[Maurer et~al.(2023)Maurer, Bretz, and Xun]{maurer_optimal_test_2023b}
Willi Maurer, Frank Bretz, and Xiaolei Xun.
\newblock Optimal test procedures for multiple hypotheses controlling the
  familywise expected loss.
\newblock \emph{Biometrics}, 79\penalty0 (4):\penalty0 2781--2793, 2023.

\bibitem[Ramdas and Wang(2025)]{ramdas_hypothesis_testing_2025}
Aaditya Ramdas and Ruodu Wang.
\newblock Hypothesis testing with e-values.
\newblock \emph{Foundation and Trends{\textregistered} in Statistics}, 2025.

\bibitem[Ramdas et~al.(2019)Ramdas, Barber, Wainwright, and
  Jordan]{ramdas_unified_treatment_2019}
Aaditya~K. Ramdas, Rina~F. Barber, Martin~J. Wainwright, and Michael~I. Jordan.
\newblock A unified treatment of multiple testing with prior knowledge using
  the p-filter.
\newblock \emph{The Annals of Statistics}, 47\penalty0 (5):\penalty0
  2790--2821, 2019.

\bibitem[Ren and Barber(2024)]{ren2024derandomised}
Zhimei Ren and Rina~Foygel Barber.
\newblock Derandomised knockoffs: leveraging e-values for false discovery rate
  control.
\newblock \emph{Journal of the Royal Statistical Society: Series B (Statistical
  Methodology)}, 86\penalty0 (1):\penalty0 122--154, 2024.

\bibitem[Romano and Shaikh(2006)]{romano_shaikh_stepdown_2006}
Joseph~P. Romano and Azeem~M. Shaikh.
\newblock On stepdown control of the false discovery proportion.
\newblock \emph{The Annals of Statistics}, 34\penalty0 (4):\penalty0
  1850--1873, 2006.

\bibitem[Romano and Wolf(2007)]{romano_wolf_generalized_2007}
Joseph~P. Romano and Michael Wolf.
\newblock Control of generalized error rates in multiple testing.
\newblock \emph{The Annals of Statistics}, 35\penalty0 (4):\penalty0
  1378--1408, 2007.

\bibitem[Rosenblatt et~al.(2018)Rosenblatt, Finos, Weeda, Solari, and
  Goeman]{rosenblatt2018all}
Jonathan~D Rosenblatt, Livio Finos, Wouter~D Weeda, Aldo Solari, and Jelle~J
  Goeman.
\newblock All-resolutions inference for brain imaging.
\newblock \emph{NeuroImage}, 181:\penalty0 786--796, 2018.

\bibitem[Sarkar(2008)]{sarkar_methods_controlling_2008}
Sanat~K. Sarkar.
\newblock On methods controlling the false discovery rate.
\newblock \emph{Sankhy$\bar{a}$: The Indian Journal of Statistics, Series A},
  70\penalty0 (2):\penalty0 135--168, 2008.

\bibitem[Shafer(2021)]{shafer_testing_betting_2021}
Glenn Shafer.
\newblock Testing by betting: {{A}} strategy for statistical and scientific
  communication.
\newblock \emph{Journal of the Royal Statistical Society: Series A (Statistics
  in Society)}, 184\penalty0 (2):\penalty0 407--431, 2021.

\bibitem[Shafer et~al.(2011)Shafer, Shen, Vereshchagin, and
  Vovk]{shafer2011test}
Glenn Shafer, Alexander Shen, Nikolai Vereshchagin, and Vladimir Vovk.
\newblock Test martingales, {Bayes} factors and p-values.
\newblock \emph{Statistical Science}, 26\penalty0 (1):\penalty0 84--101, 2011.

\bibitem[Shaffer(1986)]{shaffer1986modified}
Juliet~Popper Shaffer.
\newblock Modified sequentially rejective multiple test procedures.
\newblock \emph{Journal of the American Statistical Association}, 81\penalty0
  (395):\penalty0 826--831, 1986.

\bibitem[Simes(1986)]{simes1986improved}
R~John Simes.
\newblock An improved {Bonferroni} procedure for multiple tests of
  significance.
\newblock \emph{Biometrika}, 73\penalty0 (3):\penalty0 751--754, 1986.

\bibitem[Solari and Goeman(2017)]{solari2017minimally}
Aldo Solari and Jelle~J Goeman.
\newblock Minimally adaptive {BH}: A tiny but uniform improvement of the
  procedure of {Benjamini} and {Hochberg}.
\newblock \emph{Biometrical Journal}, 59\penalty0 (4):\penalty0 776--780, 2017.

\bibitem[Sonnemann and
  Finner(1988)]{sonnemann_vollstaendigkeitssaetze_fuer_1988}
E.~Sonnemann and H.~Finner.
\newblock Vollst{\"a}ndigkeitss{\"a}tze f{\"u}r multiple testprobleme.
\newblock In \emph{{Multiple Hypothesenpr{\"u}fung / Multiple Hypotheses
  Testing}}, volume~70, pages 121--135. Springer, 1988.

\bibitem[Sonnemann(1982)]{sonnemann1982allgemeine}
Eckart Sonnemann.
\newblock \emph{Allgemeine L{\"o}sungen multipler Testprobleme}.
\newblock Universit{\"a}t Bern. Institut f{\"u}r Mathematische Statistik und
  Versicherungslehre, 1982.

\bibitem[Sonnemann(2008)]{sonnemann_general_solutions_2008}
Eckart Sonnemann.
\newblock General solutions to multiple testing problems.
\newblock \emph{Biometrical Journal}, 50\penalty0 (5):\penalty0 641--656, 2008.

\bibitem[Storey(2002)]{storey_direct_approach_2002a}
John~D. Storey.
\newblock A direct approach to false discovery rates.
\newblock \emph{Journal of the Royal Statistical Society: Series B (Statistical
  Methodology)}, 64\penalty0 (3):\penalty0 479--498, 2002.

\bibitem[Storey et~al.(2004)Storey, Taylor, and Siegmund]{storey2004strong}
John~D Storey, Jonathan~E Taylor, and David Siegmund.
\newblock Strong control, conservative point estimation and simultaneous
  conservative consistency of false discovery rates: a unified approach.
\newblock \emph{Journal of the Royal Statistical Society: Series B (Statistical
  Methodology)}, 66\penalty0 (1):\penalty0 187--205, 2004.

\bibitem[Su(2018)]{su2018fdr}
Weijie~J Su.
\newblock The {FDR}-linking theorem.
\newblock arXiv:1812.08965, 2018.

\bibitem[Sun and Wang(2026)]{sun_admissibility_complete_2026}
Liulei Sun and Ruodu Wang.
\newblock Admissibility and complete classes for false discovery rate control
  with {E}-values, 2026.
\newblock URL \url{https://arxiv.org/abs/2607.14380}.

\bibitem[Van~der Laan et~al.(2004)Van~der Laan, Dudoit, and
  Pollard]{van_augmentation_2004}
Mark~J Van~der Laan, Sandrine Dudoit, and Katherine~S Pollard.
\newblock Augmentation procedures for control of the generalized family-wise
  error rate and tail probabilities for the proportion of false positives.
\newblock \emph{Statistical Applications in Genetics \& Molecular Biology},
  3\penalty0 (1), 2004.

\bibitem[Vesely et~al.(2023)Vesely, Finos, and Goeman]{vesely2023permutation}
Anna Vesely, Livio Finos, and Jelle~J Goeman.
\newblock Permutation-based true discovery guarantee by sum tests.
\newblock \emph{Journal of the Royal Statistical Society: Series B (Statistical
  Methodology)}, 85\penalty0 (3):\penalty0 664--683, 2023.

\bibitem[Vovk and Wang(2021)]{vovk_e-values_calibration_2021}
Vladimir Vovk and Ruodu Wang.
\newblock E-values: {{Calibration}}, combination and applications.
\newblock \emph{The Annals of Statistics}, 49\penalty0 (3):\penalty0
  1736--1754, 2021.

\bibitem[Vovk and Wang(2023)]{vovk_confidence_discoveries_2023a}
Vladimir Vovk and Ruodu Wang.
\newblock Confidence and discoveries with e-values.
\newblock \emph{Statistical Science}, 38\penalty0 (2):\penalty0 329--354, 2023.

\bibitem[Vovk et~al.(2022)Vovk, Wang, and Wang]{vovk2022admissible}
Vladimir Vovk, Bin Wang, and Ruodu Wang.
\newblock Admissible ways of merging p-values under arbitrary dependence.
\newblock \emph{The Annals of Statistics}, 50\penalty0 (1):\penalty0 351--375,
  2022.

\bibitem[Wang(2025)]{wang_only_admissible_2025}
Ruodu Wang.
\newblock The only admissible way of merging arbitrary e-values.
\newblock \emph{Biometrika}, 2025.

\bibitem[Wang and Ramdas(2022)]{wang_false_discovery_2022}
Ruodu Wang and Aaditya Ramdas.
\newblock False discovery rate control with e-values.
\newblock \emph{Journal of the Royal Statistical Society: Series B (Statistical
  Methodology)}, 84\penalty0 (3):\penalty0 822--852, 2022.

\bibitem[Westfall and Young(1993)]{westfall_young_resampling_1993}
Peter~H. Westfall and S.~Stanley Young.
\newblock \emph{Resampling-Based Multiple Testing: Examples and Methods for
  P-Value Adjustment}.
\newblock Wiley, New York, 1993.

\bibitem[Wiens and Dmitrienko(2005)]{wiens_fallback_procedure_2005}
Brian~L. Wiens and Alexei Dmitrienko.
\newblock The fallback procedure for evaluating a single family of hypotheses.
\newblock \emph{Journal of Biopharmaceutical Statistics}, 15\penalty0
  (6):\penalty0 929--942, 2005.

\bibitem[Xu and Ramdas(2023)]{xu_more_powerful_2023}
Ziyu Xu and Aaditya Ramdas.
\newblock More powerful multiple testing under dependence via randomization.
\newblock arXiv:2305.11126, 2023.

\bibitem[Xu et~al.(2024)Xu, Wang, and Ramdas]{xu_post-selection_inference_2022}
Ziyu Xu, Ruodu Wang, and Aaditya Ramdas.
\newblock Post-selection inference for e-value based confidence intervals.
\newblock \emph{Electronic Journal of Statistics}, 18\penalty0 (1):\penalty0
  2292--2338, 2024.

\bibitem[Xu et~al.(2025)Xu, Fischer, and Ramdas]{xu_bringing_closure_2025}
Ziyu Xu, Lasse Fischer, and Aaditya Ramdas.
\newblock Bringing closure to {{FDR}} control: Beating the
  e-{{Benjamini-Hochberg}} procedure.
\newblock arXiv:2504.11759, 2025.

\end{thebibliography}

\appendix

\section*{Appendix}
We organize the appendix as follows. 
In \Cref{sec: extensions}, we discuss several extensions to the $\textCeBH$ procedure based on modern developments in e-value based multiple testing, including randomization, boosting, and alternative e-merging functions.
\revise{We include a regularity-qualified representation result for admissible multiple testing procedures on e-values in \Cref{sec: admissibility} and, in the separate \Cref{sec: ordered closure sets}, formalize the ordering principle for symmetric local e-values.}
\Cref{sec:related-work} covers related work on closed testing with e-values and boosting the power of the eBH procedure. 
In \Cref{sec: more MT}, we show how our framework can be used to control more general multiple testing metrics such as mFDR and pFDR. 
\revise{\Cref{sec: minimally adaptive eBH} shows, for $m>4$, that $\textCeBH$ is a uniform improvement in both simultaneity and power over minimally adaptive eBH.}
\ifarver{
\Cref{sec: cebh calibrated vs cby example} provides a concrete example where $\textCeBH$\ with calibration is more powerful than $\cBY$.
}{
\Cref{sec: power examples} provides concrete examples illustrating differences in power among e-Closure based procedures.
}
\ifarver{}
{
\Cref{sec: bh and variants} presents the e-Closure formulation for the Benjamini-Hochberg procedure and its adaptive variants, \revise{showing that the particular e-collections considered there do not improve these procedures in terms of power or simultaneity.}
}
Lastly, in \Cref{sec: deferred proofs}, we provide deferred proofs for results stated in the main text.

\section{Extensions to $\textCeBH$}\label{sec: extensions}

We discuss some extensions to the $\textCeBH$ procedure. These extensions primarily highlight the fact that our procedure retains the advantages of e-value procedures. Notably, two of the improvements we discuss, randomization and boosting, apply to any method based on the e-Closure Principle, not just the methods in this section that directly use input e-values.

\paragraph{Randomization.} \citet{xu_more_powerful_2023} introduced stochastic rounding of e-values to improve e-value based hypothesis tests and the eBH procedure through randomization. We can apply stochastic rounding individually to each intersection e-value $\evalue_S$ for $S \subseteq [m]$, as follows. This technique can improve any $\ER_\F$-controlling procedure through its formulation as an e-Closure procedure; here, we illustrate a simple application to the $\CeBH$ procedure.
\revise{Define data-dependent lower support points}
\begin{align}
\revise{a_S \coloneqq \max_{R \in \rejcol^{\CeBH}_\alpha}\ \frac{\FDP_S(R)}{\alpha},
    \qquad b\coloneqq\alpha^{-1},\qquad \evalue_S^\star\coloneqq \evalue_S\wedge b}
\end{align}
\revise{for each $S \in 2^{[m]}$. Thus $0\leq a_S\leq b$. For $a_S<b$, define the stochastically rounded e-value}
\begin{align}
    \bar{\evalue}_S =
    \revise{a_S + (b-a_S) \cdot \mathbf{1}\left\{U \leq \frac{\evalue_S^\star-a_S}{b-a_S}\right\},}
\end{align}
\revise{and set $\bar\evalue_S=b$ when $a_S=b$. Here $U$ is uniform on $[0,1]$ and independent of everything else. Since membership of the original e-Closure collection gives $a_S\leq \evalue_S$, the displayed probability lies in $[0,1]$, and}
\[
\revise{\expect[\bar\evalue_S\mid\mathbf E]=\evalue_S^\star\leq\evalue_S.}
\]
\revise{Hence $\bar\evalue_{N_\pdist}$ is an e-value for every $\pdist\in M$.}

\revise{Note that $\bar{\evalue}_S \geq a_S$ almost surely, which means that $\Rcal_\alpha^\FDR(\bar{\mathbf{E}}) \supseteq \rejcol^\CeBH_\alpha$, where $\bar{\mathbf{E}} \coloneqq (\bar{\evalue}_S)_{S \subseteq [m]}$; that is, stochastic rounding can only increase the collection of valid discovery sets. Strict enlargement occurs on any event where the rounded local e-values make an additional e-Closure inequality hold, and whether such an event has positive probability depends on the true distribution.}

\paragraph{Boosting.} \citet{wang_false_discovery_2022} observed that one can increase the power of an e-value when testing against a finite set of rejection thresholds. In the most general setting of assuming arbitrary dependence, if one knows the marginal distribution of an e-value under the null, one can calculate a boosting factor $b$ for an e-value $E$ as follows $$b \coloneqq \max\ \{b' \in [1, \infty): \expect_\pdist[T(b' E)] \leq 1\},$$ where the expectation is taken under the known null distribution $\pdist \in M$. Here, $T$ is a truncation function that takes the largest element in a predetermined set that is less than or equal to the input.

In our setting, one simple way to boost e-values is to define a different truncation function $T_A$ for each $A \subseteq [m]$ as follows: $$T_A(x) \coloneqq \max\ \left( \{r / (\alpha k) : r / (\alpha k) \leq x, k \in [m], r \in [k \wedge |A|]\} \cup \{0\} \right).$$
This truncation function differs from those used to boost the eBH procedure and is derived by restricting the support of $\evalue_S$ to the values that $\FDP_S(R) / \alpha$ can take. \revise{For an intersection null, define}
\[
\revise{b_S\coloneqq\sup\left\{b\geq1:\sup_{\pdist\in H_S}\expect_\pdist[T_S(b\evalue_S)]\leq1\right\}.}
\]
\revise{Then $T_S(b_S \evalue_S)$ is a valid e-value (provided the supremum is attained; otherwise use any feasible $b<b_S$) and is at least as powerful as $\evalue_S$ in the sense that $\FDP_S(R) / \alpha \leq \evalue_S$ implies $\FDP_S(R) / \alpha \leq T_S(b_S \evalue_S)$ for all $R \subseteq [m]$. When a simple null distribution is completely known, the inner supremum reduces to expectation under that single distribution.}

Using the above boosted e-values immediately yields that the boosted $\textCeBH$ procedure dominates the original.
Other boosting techniques may be possible given more dependence assumptions or knowledge about the joint distribution of the entire e-collection $(\evalue_S)_{S \subseteq [m]}$.

\paragraph{Alternative e-merging functions.} While we used the average e-merging function in our formulation of the $\textCeBH$ procedure, one can use other e-merging functions when stronger dependence assumptions, i.e., independence or sequential dependence, are made about the e-values.
For example, if we assume all null e-values are independent, one can let the e-values for each intersection hypothesis be the product of the base hypothesis e-values, i.e.,
\begin{align}
    \evalue_S = \prod\limits_{i \in S} \evalue_i.
\end{align} One can then proceed to define an alternative $\textCeBH$ procedure for the product e-merging function in a similar fashion as we did in \Cref{sec: e-combining}, along with the corresponding algorithm of Section \ref{sec: computation}, by using the logarithm to turn the products into sums when computing e-values for the intersection hypotheses.

\section{\revise{Admissibility under arbitrary dependence}}\label{sec: admissibility}
\revise{For this section, call a deterministic simultaneous procedure $\rejcol(x)$ \emph{admissible} if there is no other valid procedure $\rejcol'(x)$ with $\rejcol(x)\subseteq\rejcol'(x)$ for every input $x$ and strict inclusion for at least one input.
}

\begin{proposition}\label{thm: admissible e-values}
\revise{Let $M$ be the class of all probability distributions on $[0,\infty)^m$, and, for each $S\subseteq[m]$, let $\bar H_S$ be the set of distributions under which $\expect_\pdist[\evalue_i]\leq1$ for $i\in S$ and $\expect_\pdist[\evalue_j]>1$ for $j\notin S$. Suppose an admissible deterministic $\ER_\F$-controlling procedure is represented by the e-Closure discovery set \eqref{eq: e-closure} from Borel local e-values}
\[
\revise{\evalue_S=F_S((\evalue_i)_{i\in S},(\evalue_j)_{j\notin S}).}
\]
\revise{Assume, for every nonempty $S$, that the envelope}
\[
\revise{\psi_S((x_i)_{i\in S})\coloneqq\sup_{(x_j)_{j\notin S}}
F_S((x_i)_{i\in S},(x_j)_{j\notin S})}
\]
\revise{is Borel measurable. Then the procedure can be represented by an e-collection (with $\evalue'_\emptyset=0$) where for every nonempty $S$,}
\[
\revise{\evalue'_S=\sum_{i\in S}w_{i,S}\evalue_i+c_S,}
\]
\revise{where $w_{i,S},c_S\geq0$ and $\sum_{i\in S}w_{i,S}+c_S\leq1$.}
\end{proposition}

\begin{proof}
\revise{Fix a nonempty $S$. Because $\evalue_S$ is an e-value under every distribution in $\bar H_S$, we first show that $\psi_S$ is finite. If $\psi_S(x_S)=\infty$ at some finite $x_S$, choose $q\in(0,1)$ such that $qx_i\leq1$ for every $i\in S$, and choose $x_{-S}$ with $F_S(x_S,x_{-S})>1/q$. Put probability $q$ on $(x_S,x_{-S})$ and probability $1-q$ on $(0,z_{-S})$, where the coordinates of $z_{-S}$ are sufficiently large that their expectations exceed one. This distribution belongs to $\bar H_S$, but}
\[
\revise{\expect_\pdist[F_S((\evalue_i)_{i\in S},(\evalue_j)_{j\notin S})]
\geq qF_S(x_S,x_{-S})>1,}
\]
\revise{a contradiction. (When $S=[m]$, the second atom may simply be placed at the origin.)}

\revise{It remains to show that $\psi_S$ is a generalized e-merging function. Let $Q$ be any distribution of $(\evalue_i)_{i\in S}$ satisfying $\expect_Q[\evalue_i]\leq1$ for every $i\in S$. For $\varepsilon\in(0,1)$, the Borel set}
\[
\revise{A_\varepsilon
\coloneqq
\{(x_S,x_{-S}):F_S(x_S,x_{-S})>\psi_S(x_S)-\varepsilon\}}
\]
\revise{has a nonempty section at every $x_S$. A standard measurable-selection theorem (the Jankov--von Neumann theorem; see \citet[Theorem~18.1]{kechris_classical_1995}) therefore lets us choose an $\varepsilon$-near-maximizer $g_\varepsilon(x_S)$ such that $(x_S,g_\varepsilon(x_S))\in A_\varepsilon$, with $g_\varepsilon$ universally measurable. Since $Q$ is arbitrary, this choice is measurable under the completion of $Q$, which is sufficient for the expectation argument below. If $S\neq[m]$, set}
\[
\revise{\delta_\varepsilon(x_S)\coloneqq
\frac{\varepsilon}{1+\psi_S(x_S)},
\qquad
L_\varepsilon(x_S)\coloneqq\frac{2}{\delta_\varepsilon(x_S)}.}
\]
\revise{Construct a distribution $\pdist_\varepsilon$ by drawing $x_S\sim Q$ and, conditionally on $x_S$, setting $x_{-S}=g_\varepsilon(x_S)$ with probability $1-\delta_\varepsilon(x_S)$ and setting every coordinate of $x_{-S}$ equal to $L_\varepsilon(x_S)$ otherwise. The marginal distribution of $x_S$ remains $Q$, while each outside coordinate has expectation at least $\expect_Q[\delta_\varepsilon L_\varepsilon]=2$; hence $\pdist_\varepsilon\in\bar H_S$. Nonnegativity and the selector property give}
\[
\begin{aligned}
\revise{1}
&\revise{\geq \expect_{\pdist_\varepsilon}\!\left[
F_S((\evalue_i)_{i\in S},(\evalue_j)_{j\notin S})\right]} \\
&\revise{\geq \expect_Q\!\left[
\left(1-\delta_\varepsilon((\evalue_i)_{i\in S})\right)
\left(\psi_S((\evalue_i)_{i\in S})-\varepsilon\right)\right]} \\
&\revise{\geq \expect_Q[\psi_S((\evalue_i)_{i\in S})]-2\varepsilon.}
\end{aligned}
\]
\revise{For $S=[m]$, the same conclusion follows directly because $\psi_S=F_S$. Letting $\varepsilon\downarrow0$ shows that $\psi_S((\evalue_i)_{i\in S})$ is an e-value for every arbitrarily dependent collection of e-values indexed by $S$. Thus $\psi_S$ is a Borel-measurable generalized e-merging function.}

\revise{By construction, $F_S(x)\leq\psi_S(x_S)$. Crucially, no monotonicity of $\psi_S$ is required: \citet[Theorem 1]{clerico_simple_geometric_2026} proves that every Borel-measurable generalized e-merging function is dominated by a weighted affine e-merging function. Hence there are weights $w_{i,S}\geq0$ with $\sum_{i\in S}w_{i,S}\leq1$; setting $c_S=1-\sum_{i\in S}w_{i,S}$ gives}
\[
\revise{F_S(x)\leq\psi_S(x_S)\leq\sum_{i\in S}w_{i,S}x_i+c_S.}
\]
\revise{The upper bound is a valid e-value under the corresponding partitioning null and therefore defines an e-collection whose e-Closure contains the original procedure pointwise. Admissibility rules out strict containment, so the e-Closure procedure derived from the e-collection represents the same procedure.}
\end{proof}

\revise{For FDR specifically, \citet[Theorem 5.1]{sun_admissibility_complete_2026} prove the stronger complete-class result that every simultaneous FDR-controlling procedure based on arbitrary input e-values is strongly dominated by an admissible weighted-mean closed eBH procedure. Their FDR-specific theorem complements the general-loss representation above.}

\section{\revise{Ordering symmetric local e-values}}\label{sec: ordered closure sets}

\revise{Call a loss function permutation-invariant if
\[
\F_{\pi(N)}(\pi(R))=\F_N(R)
\]
for every permutation $\pi$ of $[m]$. Equivalently, there is a function $\phi$ such that
\[
\F_N(R)=\phi\bigl(|N|,|R|,|N\cap R|\bigr).
\]
For the ordering result, we additionally assume that $\phi(n,r,k)$ is nondecreasing in $k$ for fixed $n,r$: with the numbers of nulls and rejections fixed, declaring more null hypotheses to be discoveries cannot reduce the loss.}

\revise{Separately, suppose the local e-values are generated symmetrically from base e-values,
\[
\evalue_S=g_{|S|}\bigl((\evalue_i)_{i\in S}\bigr),
\]
where each $g_s$ is symmetric and coordinatewise nondecreasing. The p-value analogue takes
\[
\evalue_S=g_{|S|}\bigl((\pvalue_i)_{i\in S}\bigr)
\]
with each $g_s$ symmetric and coordinatewise nonincreasing.}

\begin{lemma}\label{lem: ordered closure sets}
\revise{Let $\F$ satisfy the permutation-invariance and monotonicity conditions above. If a set $R$ of size $r$ belongs to $\Rcal_\alpha^{\ER_\F}(\mathbf E)$, then the set of the $r$ largest base e-values also belongs to $\Rcal_\alpha^{\ER_\F}(\mathbf E)$ under the e-value condition above. Under the analogous p-value condition, the same holds for the set of the $r$ smallest p-values. Ties may be resolved by any fixed deterministic rule.}
\end{lemma}

\begin{proof}
\revise{We give the e-value argument; the p-value argument is identical with the inequalities reversed. Suppose $i\in R$, $j\notin R$, and $\evalue_j\geq\evalue_i$, and define $R'=(R\setminus\{i\})\cup\{j\}$. We show that $R'\in\Rcal_\alpha^{\ER_\F}(\mathbf E)$. For a set $S$ containing both $i,j$ or neither, its e-Closure inequality is unchanged. If $i\in S$ and $j\notin S$, then $|S\cap R'|\leq|S\cap R|$, so monotonicity of $\phi$ gives $\F_S(R')\leq\F_S(R)$ and the inequality remains valid. Finally, if $j\in S$ and $i\notin S$, let $S'=(S\setminus\{j\})\cup\{i\}$. Permutation invariance gives
\[
\F_S(R')=\F_{S'}(R),
\qquad\text{while}\qquad
\evalue_S\geq\evalue_{S'}
\]
by symmetry and coordinatewise monotonicity of the local e-values. Hence the closure inequality for $(S,R')$ follows from that for $(S',R)$. Repeating such exchanges produces the set of the $r$ largest e-values without losing closure consistency.}
\end{proof}

\section{Related work}\label{sec:related-work}

\revise{We first note the connection of the e-Closure Principle to \citet{grunwald_neyman-pearson_e-values_2024}, and then briefly describe the two lines of research and prior art that are most relevant to our current work. 
}
Aside from these, there is a substantial literature using e-values and the eBH procedure for FDR control --- \citet{ramdas_hypothesis_testing_2025} performs a thorough overview of such developments.

\revise{\paragraph{Multiple testing as a generalized Neyman-Pearson problem.} To see that the general multiple testing problem in \Cref{sec: MT} can be viewed as a composite generalized Neyman-Pearson (GNP) problem \citep{grunwald_neyman-pearson_e-values_2024}, we observe that we can let $\F$ be our loss function, and $2^{[m]}$ is both the parameter and the action space for our GNP problem.
Then, we can view the e-Closure Procedure as an application of \citet[Proposition 3]{grunwald_neyman-pearson_e-values_2024}, and outputting the largest discovery set in $\Rcal_\alpha^{\ER_\F}(\mathbf{E})$ akin to being the maximally compatible decision rule for the multiple testing problem w.r.t.\ the e-collection derived in \Cref{thm: e-closure} for FDR control. However, admissibility results for maximally compatible decision rules from \citet{grunwald_neyman-pearson_e-values_2024} do not directly apply to our setting since their assumption on action ``richness'' (continuity of loss --- for FDP or other discrete multiple testing error metrics  --- w.r.t.\ the action space) is violated in multiple testing, as there is only a finite set of values that the FDP can take. Moreover, we are also concerned with simultaneity to enable flexibility against multiple utility functions rather than optimization w.r.t.\ a single predetermined utility function in the multiple testing setting. 
}

\paragraph{Closed testing and e-values.} There is a vast amount of literature for FWER control and simultaneous (over all potential discovery sets) FDP bounds using closed testing --- see overviews presented in \citet[Chapters 9 and 14]{cui_handbook_multiple_2021} and \citet{henning_closed_testing_2015}. We will point out three specific works that are particularly relevant to ours since they also utilize e-values in the context of closed testing.
The first is \citet{vovk_confidence_discoveries_2023a}, who formulate a method for FWER control via the arithmetic mean e-merging function, and consequently also simultaneous FDP bounds. \citet{fischer2024admissible} show that e-values and their validity under optional continuation are essential for constructing simultaneous FDP bounds in online multiple testing, where a stream of hypotheses is tested sequentially.
\citet{hartog_family-wise_error_2025} extend the benefit of using arithmetic mean e-merging functions to graphical approaches for FWER control \citep{bretz_graphical_approach_2009}, and derive computationally efficient algorithms for e-value versions of the Fallback \citep{wiens_fallback_procedure_2005} and \citet{holm_simple_sequentially_1979} procedures.
Our work proposes new procedures for the FDR metric
which has not been connected to intersection hypothesis testing or e-merging functions before, and we show existing procedures (such as e-Holm) for FWER or tail bounds on the FDP can be formulated using the generalized Closure Principle we develop in \Cref{sec: principle}.

\paragraph{Boosting the power of eBH.} Another pertinent line of work develops approaches to boosting the eBH procedure so that it makes more discoveries. This work generally falls into two areas: (1) boosting with no assumptions on the e-value distributions and (2) boosting e-values using explicit knowledge of the underlying marginal or joint distribution.

For the first area, we have already discussed the connection between our work and the stochastic rounding framework of \citet{xu_more_powerful_2023} in \Cref{sec: extensions}.The second area includes two general types of assumptions.
\begin{itemize}
    \item \emph{Marginal assumptions:} Along with introducing the eBH procedure, \citet{wang_false_discovery_2022} also outlined how one can boost e-values if one knows the explicit marginal distribution of each e-value, and we discuss a simple variant for boosting the $\textCeBH$ procedure in \Cref{sec: extensions}. \citet{blier-wong_improved_thresholds_2024} relaxed the necessary knowledge on marginal distributions for boosting by deriving boosting thresholds for certain nonparametric classes of e-values.
    \item \emph{Conditional assumptions:} \citet{lee_boosting_e-bh_2024a,lee_full-conformal_novelty_2025} improve the power of eBH in the setting where one knows the distribution of all $m$ e-values when conditioned on a sufficient statistic calculated from the data under a singular null (i.e., a specific hypothesis is null, and all other hypotheses may be either null or non-null). This can also be viewed as a framework for deriving compound e-values from e-values when their joint distribution is known. Boosting factors can be calculated for each e-value via exact knowledge of these conditional distributions.
\end{itemize}
\revise{Further, conformal e-values provide a concrete application of e-value-based FDR control. \citet{lee_full-conformal_novelty_2025} construct full-conformal e-values for novelty detection and combine them with conditionally calibrated eBH, while \citet{lee_selection_hierarchical_2025} construct subsampling and hierarchical conformal e-values for selection from hierarchical data and apply eBH to control the FDR. These constructions provide natural settings in which e-Closure may yield additional simultaneity guarantees.}
Developments in this second area are orthogonal to our contribution. These boosting methods all require additional knowledge about the underlying distributions. On the other hand, the $\textCeBH$ procedure requires no assumptions on the e-values, and can be applied with valid FDR control in all situations where the eBH procedure can be applied. Potentially, these boosting methods can be used in conjunction with $\textCeBH$ to further improve its power.

\section{More general multiple testing metrics}\label{sec: more MT}

\revise{Ratio metrics can be brought into the framework by a nonnegative reduction, after relaxing the convenient convention $\F_N(\emptyset)=0$. We say that a discovery set $\rejset$ controls the ratio metric $\textnormal{RER}_{\F, \G, \eta}$ at level $\alpha$ if}
\begin{align}
    \frac{\expect_\pdist(\F_{N_\pdist}(\rejset))}{\expect_\pdist(\G_{N_\pdist}(\rejset)) + \eta}
    \revise{\leq\alpha}
\end{align} \revise{for all $\pdist \in M$, where $\F_N$ and $\G_N$ are nonnegative, $0\leq\G_N(R)\leq G<\infty$, and $\eta\geq0$; when $\eta=0$, the claim is restricted to distributions for which the denominator is positive. Note that mFDR is a ratio metric: $\textnormal{mFDR}_\eta$ for $\eta>0$ \citep{foster_a-investing_procedure_2008} has $\F_N(R) = |N \cap R|$ and $\G_N(R)=|R|$ (with $G=m$). More recently, \citet{dandapanthula_multiple_testing_2025} proposed the notion of error-over-patience (EOP) risk metrics of the form}
\begin{align}
    \frac{\expect_\pdist(\F_{N_\pdist}(\rejset_\tau))}{\expect_\pdist(\tau)}
\end{align} for a user-chosen stopping time $\tau$ and sequence of discovery sets $(\rejset_t)_{t \in \naturals}$. \revise{On a bounded horizon $\tau\leq T$, one may let $\G_N(R,t)=t-1$, $G=T-1$, and $\eta=1$. Extending the reduction to an unbounded stopping time is outside the bounded formulation used here and is left to future work.}

\revise{We now give the nonnegative reduction.}
\begin{proposition}
\revise{For any random discovery set $\rejset$, define the nonnegative, $\alpha$-dependent loss}
\[
\revise{\F^{\mathrm{RER},\alpha}_N(R)
\coloneqq\frac{\F_N(R)+\alpha\{G-\G_N(R)\}}{G+\eta}.}
\]
\revise{Then}
\[
\revise{\frac{\expect_\pdist[\F_{N_\pdist}(\rejset)]}
{\expect_\pdist[\G_{N_\pdist}(\rejset)]+\eta}\leq\alpha
\quad\Longleftrightarrow\quad
\expect_\pdist[\F^{\mathrm{RER},\alpha}_{N_\pdist}(\rejset)]\leq\alpha.
}
\]
\end{proposition}
\revise{The proof is immediate after multiplying both inequalities by their nonnegative denominators. The reduction may give $\F^{\mathrm{RER},\alpha}_N(\emptyset)>0$ (as is the case for mFDR), so this application uses the e-Closure proof without the simplifying empty-set convention; none of the inequalities in \Cref{thm: e-closure} require that convention for validity.}
\revise{For ordinary pFDR, take $\F_N(R)=\FDP_N(R)$, $\G_N(R)=\ind\{R\neq\emptyset\}$, $G=1$, and $\eta=0$. The transformed loss is}
\[
\revise{\FDP_N(R)+\alpha\ind\{R=\emptyset\}.}
\]
\revise{Thus pFDR is covered by this explicitly stated extension whenever $\pdist(R\neq\emptyset)>0$; it is not a member of the $\eta>0$ ratio class.}

\section{$\textCeBH$ uniformly improves minimally adaptive eBH}\label{sec: minimally adaptive eBH}

The minimally adaptive eBH procedure of \citet{ignatiadis2024values} makes the following discoveries:
\begin{align}
    \mathbf{r}_{\text{ma-}\eBH} &\revise{\coloneqq \max\bigl(\{r \in [m]: \evalue_{(r)} \geq (m - 1)/(\alpha r)\}\cup\{0\}\bigr)},\\
    \rejset_\alpha^{\text{ma-}\eBH} &\coloneqq \begin{cases}
    \rejset^{(\mathbf{r}_{\text{ma-}\eBH})} & \text{if } \revise{m^{-1}\sum_{i \in [m]}\evalue_i \geq \alpha^{-1}} \\
    \emptyset & \text{otherwise}
    \end{cases}
\end{align}
Note that this uniformly dominates the eBH procedure since the rejection thresholds are lower, and the global null test of average e-values also dominates the Simes e-value that is implicitly used in the eBH procedure. However, we will show that this procedure is also dominated by $\CeBH$ procedure.

\begin{theorem} \label{thm: minimally adaptive eBH}
\revise{When $m>4$, $\rejcol^\CeBH_\alpha$ is a uniform improvement in both simultaneity and power over $\rejset_\alpha^{\textnormal{ma-}\eBH}$.}
\end{theorem}
\begin{proof}
\revise{Abbreviate $\mathbf{R} = \rejset_\alpha^{\textnormal{ma-}\eBH}$. If $\mathbf{R} \neq \emptyset$, then when $S = [m]$, we have $\evalue_{[m]} \geq \alpha^{-1} = \FDP_{[m]}(\rejset)/\alpha$ by the global gate. If $S\cap\mathbf R=\emptyset$, the e-Closure inequality is immediate. For all remaining $S\subsetneq[m]$, we have}
\[
\frac{1}{|S|} \sum_{i \in S} \mathbf{e}_i
\geq
\frac{|\mathbf{R} \cap S|}{|S|} \cdot \min_{i \in  \mathbf{R} \cap S}\mathbf{e}_{i}
\geq
\frac{|\mathbf{R} \cap S|}{|S|} \cdot \min_{i \in  \mathbf{R}}\mathbf{e}_{i}
\geq
\frac{|\mathbf{R} \cap S|}{|S|} \cdot \frac{m - 1}{\alpha |\mathbf{R}|}
\geq
\frac{\FDP_S(\mathbf{R})}{\alpha}
\]
which shows that $\mathbf{R} \in \rejcol_\alpha^\CeBH$. The last inequality is because we are only considering $S$ s.t.\ $|S| \leq m - 1$. Showing this bound is trivially true if $\mathbf{R}=\emptyset$.

To show actual improvement, let $m>4$ and consider the event that $\mathbf{e}_1=(m-2)/\alpha$, $\mathbf{e}_2 = \evalue_3 = 1/\alpha$, $\mathbf{e}_4 = \ldots = \mathbf{e}_m = 0$. Then $\mathbf{R}_\alpha^{\textnormal{ma-}\eBH} = \emptyset$, but $1 \in  \rejset^\CeBH_\alpha$, because $|S|^{-1} \sum_{i \in S} \mathbf{e}_i \geq \alpha^{-1}$ whenever $1 \in S$.
\end{proof}
\ifarver{}
{
\section{Concrete examples of interest}\label{sec: power examples}
In this section, we collect examples in which the methods in our paper demonstrate improvements or other notable behavior.

\subsection{Power improvements in p-value based methods}\label{sec: p-value power improvements}
We illustrate two examples from \Cref{sec: p-combining} in which the e-Closure based procedures improve upon their original counterparts. \Cref{fig:BY} demonstrates the improvement in the rejection thresholds of $\cBY$ over BY as the number of discoveries increases. Similarly, \Cref{fig:Su} illustrates the corresponding improvement of $\cSu$ over Su.
\begin{figure}[!ht]
\centering
\begin{tikzpicture}[scale=.7]
\begin{axis}[
	xmin = 1,
	xmax = 20,
	xtick = {5,10,15,20},
	ymin = 0,
	ymax = 0.05,
	ytick = {0,0.01,0.02,0.03,0.04,0.05},
	yticklabels = {0,0.01,0.02,0.03,0.04,0.05},
	 ytick scale label code/.code={},
	ylabel=$p$-value,
	xlabel=rank,
        height=10cm,
	width=10cm,
		legend style={legend columns=1, at={(0.275,0.975)}}
]
 \addplot[color=blue, very thick,  mark size=2pt, mark=*, fill=blue, fill opacity=0.5, draw opacity=0.75] coordinates {
( 1,7e-04 ) ( 2,0.0014 ) ( 3,0.0021 ) ( 4,0.0028 ) ( 5,0.0035 ) ( 6,0.0042 ) ( 7,0.0049 ) ( 8,0.0056 ) ( 9,0.0063 ) ( 10,0.0069 ) ( 11,0.0076 ) ( 12,0.0083 ) ( 13,0.009 ) ( 14,0.0097 ) ( 15,0.0104 ) ( 16,0.0111 ) ( 17,0.0118 ) ( 18,0.0125 ) ( 19,0.0132 ) ( 20,0.0139 ) ( 20,1 ) ( 19,1 ) ( 18,1 ) ( 17,1 ) ( 16,1 ) ( 15,1 ) ( 14,1 ) ( 13,1 ) ( 12,1 ) ( 11,1 ) ( 10,1 ) ( 9,1 ) ( 8,1 ) ( 7,0.0077 ) ( 6,0.0063 ) ( 5,0.0053 ) ( 4,0.0044 ) ( 3,0.0044 ) ( 2,0.0038 ) ( 1,0.003 )
};
 \addlegendentry{$k=7$}
\addplot[color=red, very thick,  mark size=2pt, mark=*, fill=red, fill opacity=0.5, draw opacity=0.75] coordinates {
( 1,7e-04 ) ( 2,0.0014 ) ( 3,0.0021 ) ( 4,0.0028 ) ( 5,0.0035 ) ( 6,0.0042 ) ( 7,0.0049 ) ( 8,0.0056 ) ( 9,0.0063 ) ( 10,0.0069 ) ( 11,0.0076 ) ( 12,0.0083 ) ( 13,0.009 ) ( 14,0.0097 ) ( 15,0.0104 ) ( 16,0.0111 ) ( 17,0.0118 ) ( 18,0.0125 ) ( 19,0.0132 ) ( 20,0.0139 ) ( 20,1 ) ( 19,0.0333 ) ( 18,0.0273 ) ( 17,0.024 ) ( 16,0.0208 ) ( 15,0.0129 ) ( 14,0.0129 ) ( 13,0.0129 ) ( 12,0.0129 ) ( 11,0.0129 ) ( 10,0.0129 ) ( 9,0.0129 ) ( 8,0.0129 ) ( 7,0.0129 ) ( 6,0.0129 ) ( 5,0.0123 ) ( 4,0.0057 ) ( 3,0.0057 ) ( 2,0.0057 ) ( 1,0.0057 )
};
 \addlegendentry{$k=19$}
 \addplot[color=black, very thick,  mark size=2pt, mark=*] coordinates {
( 1,7e-04 ) ( 2,0.0014 ) ( 3,0.0021 ) ( 4,0.0028 ) ( 5,0.0035 ) ( 6,0.0042 ) ( 7,0.0049 ) ( 8,0.0056 ) ( 9,0.0063 ) ( 10,0.0069 ) ( 11,0.0076 ) ( 12,0.0083 ) ( 13,0.009 ) ( 14,0.0097 ) ( 15,0.0104 ) ( 16,0.0111 ) ( 17,0.0118 ) ( 18,0.0125 ) ( 19,0.0132 ) ( 20,0.0139 )
};
\end{axis}
\end{tikzpicture}~
\begin{tikzpicture}[scale=.7, define rgb/.code={\definecolor{mycolor}{RGB}{#1}},
                    rgb color/.style={define rgb={#1},mycolor}]
\begin{axis}[
xmin = 1,
	xmax = 20,
	xtick = {5,10,15,20},
	ymin = 0,
	ymax = 0.05,
	ytick = {0,0.01,0.02,0.03,0.04,0.05},
	yticklabels = {0,0.01,0.02,0.03,0.04,0.05},
	 ytick scale label code/.code={},
	ylabel=,
	xlabel=rank,
        height=10cm,
	width=10cm
]
\addplot[rgb color={255, 0, 0}, very thick,  mark size=2pt, mark=., fill opacity=0.5, draw opacity=0.75] coordinates {
( 1,7e-04 ) ( 2,0.0014 ) ( 3,0.0021 ) ( 4,0.0028 ) ( 5,0.0035 ) ( 6,0.0042 ) ( 7,0.0049 ) ( 8,0.0056 ) ( 9,0.0063 ) ( 10,0.0069 ) ( 11,0.0076 ) ( 12,0.0083 ) ( 13,0.009 ) ( 14,0.0097 ) ( 15,0.0104 ) ( 16,0.0111 ) ( 17,0.0118 ) ( 18,0.0125 ) ( 19,0.0132 ) ( 20,0.0139 ) ( 20,1 ) ( 19,0.0333 ) ( 18,0.0273 ) ( 17,0.024 ) ( 16,0.0208 ) ( 15,0.0129 ) ( 14,0.0129 ) ( 13,0.0129 ) ( 12,0.0129 ) ( 11,0.0129 ) ( 10,0.0129 ) ( 9,0.0129 ) ( 8,0.0129 ) ( 7,0.0129 ) ( 6,0.0129 ) ( 5,0.0123 ) ( 4,0.0057 ) ( 3,0.0057 ) ( 2,0.0057 ) ( 1,0.0057 )
};
\addplot[rgb color={255, 195, 0}, very thick,  mark size=2pt, mark=., fill opacity=0.5, draw opacity=0.75] coordinates {
( 1,7e-04 ) ( 2,0.0014 ) ( 3,0.0021 ) ( 4,0.0028 ) ( 5,0.0035 ) ( 6,0.0042 ) ( 7,0.0049 ) ( 8,0.0056 ) ( 9,0.0063 ) ( 10,0.0069 ) ( 11,0.0076 ) ( 12,0.0083 ) ( 13,0.009 ) ( 14,0.0097 ) ( 15,0.0104 ) ( 16,0.0111 ) ( 17,0.0118 ) ( 18,0.0125 ) ( 19,0.0132 ) ( 20,0.0139 ) ( 20,1 ) ( 19,1 ) ( 18,1 ) ( 17,1 ) ( 16,1 ) ( 15,0.0204 ) ( 14,0.0193 ) ( 13,0.0115 ) ( 12,0.0115 ) ( 11,0.0115 ) ( 10,0.0115 ) ( 9,0.0115 ) ( 8,0.0115 ) ( 7,0.0115 ) ( 6,0.0115 ) ( 5,0.0078 ) ( 4,0.0067 ) ( 3,0.0065 ) ( 2,0.0062 ) ( 1,0.0051 )
};
\addplot[rgb color={255, 255, 0}, very thick,  mark size=2pt, mark=., fill opacity=0.5, draw opacity=0.75] coordinates {
( 1,7e-04 ) ( 2,0.0014 ) ( 3,0.0021 ) ( 4,0.0028 ) ( 5,0.0035 ) ( 6,0.0042 ) ( 7,0.0049 ) ( 8,0.0056 ) ( 9,0.0063 ) ( 10,0.0069 ) ( 11,0.0076 ) ( 12,0.0083 ) ( 13,0.009 ) ( 14,0.0097 ) ( 15,0.0104 ) ( 16,0.0111 ) ( 17,0.0118 ) ( 18,0.0125 ) ( 19,0.0132 ) ( 20,0.0139 ) ( 20,1 ) ( 19,1 ) ( 18,1 ) ( 17,1 ) ( 16,1 ) ( 15,1 ) ( 14,1 ) ( 13,0.0184 ) ( 12,0.0128 ) ( 11,0.0128 ) ( 10,0.0128 ) ( 9,0.0128 ) ( 8,0.0124 ) ( 7,0.009 ) ( 6,0.0084 ) ( 5,0.0068 ) ( 4,0.0068 ) ( 3,0.0052 ) ( 2,0.0047 ) ( 1,0.0047 )
};
\addplot[rgb color={170, 213, 0}, very thick,  mark size=2pt, mark=., fill opacity=0.5, draw opacity=0.75] coordinates {
( 1,7e-04 ) ( 2,0.0014 ) ( 3,0.0021 ) ( 4,0.0028 ) ( 5,0.0035 ) ( 6,0.0042 ) ( 7,0.0049 ) ( 8,0.0056 ) ( 9,0.0063 ) ( 10,0.0069 ) ( 11,0.0076 ) ( 12,0.0083 ) ( 13,0.009 ) ( 14,0.0097 ) ( 15,0.0104 ) ( 16,0.0111 ) ( 17,0.0118 ) ( 18,0.0125 ) ( 19,0.0132 ) ( 20,0.0139 ) ( 20,1 ) ( 19,1 ) ( 18,1 ) ( 17,1 ) ( 16,1 ) ( 15,1 ) ( 14,1 ) ( 13,1 ) ( 12,0.0177 ) ( 11,0.0102 ) ( 10,0.0102 ) ( 9,0.0102 ) ( 8,0.0102 ) ( 7,0.0102 ) ( 6,0.0078 ) ( 5,0.0067 ) ( 4,0.0067 ) ( 3,0.0063 ) ( 2,0.0055 ) ( 1,0.0044 )
};
\addplot[rgb color={85, 170, 0}, very thick,  mark size=2pt, mark=., fill opacity=0.5, draw opacity=0.75] coordinates {
( 1,7e-04 ) ( 2,0.0014 ) ( 3,0.0021 ) ( 4,0.0028 ) ( 5,0.0035 ) ( 6,0.0042 ) ( 7,0.0049 ) ( 8,0.0056 ) ( 9,0.0063 ) ( 10,0.0069 ) ( 11,0.0076 ) ( 12,0.0083 ) ( 13,0.009 ) ( 14,0.0097 ) ( 15,0.0104 ) ( 16,0.0111 ) ( 17,0.0118 ) ( 18,0.0125 ) ( 19,0.0132 ) ( 20,0.0139 ) ( 20,1 ) ( 19,1 ) ( 18,1 ) ( 17,1 ) ( 16,1 ) ( 15,1 ) ( 14,1 ) ( 13,1 ) ( 12,1 ) ( 11,0.0171 ) ( 10,0.0083 ) ( 9,0.0083 ) ( 8,0.0083 ) ( 7,0.0083 ) ( 6,0.0083 ) ( 5,0.0083 ) ( 4,0.0067 ) ( 3,0.0049 ) ( 2,0.0045 ) ( 1,0.0045 )
};
\addplot[rgb color={0, 128, 0}, very thick,  mark size=2pt, mark=., fill opacity=0.5, draw opacity=0.75] coordinates {
( 1,7e-04 ) ( 2,0.0014 ) ( 3,0.0021 ) ( 4,0.0028 ) ( 5,0.0035 ) ( 6,0.0042 ) ( 7,0.0049 ) ( 8,0.0056 ) ( 9,0.0063 ) ( 10,0.0069 ) ( 11,0.0076 ) ( 12,0.0083 ) ( 13,0.009 ) ( 14,0.0097 ) ( 15,0.0104 ) ( 16,0.0111 ) ( 17,0.0118 ) ( 18,0.0125 ) ( 19,0.0132 ) ( 20,0.0139 ) ( 20,1 ) ( 19,1 ) ( 18,1 ) ( 17,1 ) ( 16,1 ) ( 15,1 ) ( 14,1 ) ( 13,1 ) ( 12,1 ) ( 11,1 ) ( 10,0.0151 ) ( 9,0.0115 ) ( 8,0.0098 ) ( 7,0.0087 ) ( 6,0.0072 ) ( 5,0.0045 ) ( 4,0.0045 ) ( 3,0.0045 ) ( 2,0.0045 ) ( 1,0.0045 )
};
\addplot[rgb color={0, 85, 85}, very thick,  mark size=2pt, mark=., fill opacity=0.5, draw opacity=0.75] coordinates {
( 1,7e-04 ) ( 2,0.0014 ) ( 3,0.0021 ) ( 4,0.0028 ) ( 5,0.0035 ) ( 6,0.0042 ) ( 7,0.0049 ) ( 8,0.0056 ) ( 9,0.0063 ) ( 10,0.0069 ) ( 11,0.0076 ) ( 12,0.0083 ) ( 13,0.009 ) ( 14,0.0097 ) ( 15,0.0104 ) ( 16,0.0111 ) ( 17,0.0118 ) ( 18,0.0125 ) ( 19,0.0132 ) ( 20,0.0139 ) ( 20,1 ) ( 19,1 ) ( 18,1 ) ( 17,1 ) ( 16,1 ) ( 15,1 ) ( 14,1 ) ( 13,1 ) ( 12,1 ) ( 11,1 ) ( 10,1 ) ( 9,0.0121 ) ( 8,0.0099 ) ( 7,0.0073 ) ( 6,0.007 ) ( 5,0.0065 ) ( 4,0.005 ) ( 3,0.0046 ) ( 2,0.0046 ) ( 1,0.0038 )
};
\addplot[rgb color={0, 43, 170}, very thick,  mark size=2pt, mark=., fill opacity=0.5, draw opacity=0.75] coordinates {
( 1,7e-04 ) ( 2,0.0014 ) ( 3,0.0021 ) ( 4,0.0028 ) ( 5,0.0035 ) ( 6,0.0042 ) ( 7,0.0049 ) ( 8,0.0056 ) ( 9,0.0063 ) ( 10,0.0069 ) ( 11,0.0076 ) ( 12,0.0083 ) ( 13,0.009 ) ( 14,0.0097 ) ( 15,0.0104 ) ( 16,0.0111 ) ( 17,0.0118 ) ( 18,0.0125 ) ( 19,0.0132 ) ( 20,0.0139 ) ( 20,1 ) ( 19,1 ) ( 18,1 ) ( 17,1 ) ( 16,1 ) ( 15,1 ) ( 14,1 ) ( 13,1 ) ( 12,1 ) ( 11,1 ) ( 10,1 ) ( 9,1 ) ( 8,0.0097 ) ( 7,0.0079 ) ( 6,0.0067 ) ( 5,0.0053 ) ( 4,0.0049 ) ( 3,0.0045 ) ( 2,0.0037 ) ( 1,0.0037 )
};
\addplot[rgb color={0, 0, 255}, very thick,  mark size=2pt, mark=., fill opacity=0.5, draw opacity=0.75] coordinates {
( 1,7e-04 ) ( 2,0.0014 ) ( 3,0.0021 ) ( 4,0.0028 ) ( 5,0.0035 ) ( 6,0.0042 ) ( 7,0.0049 ) ( 8,0.0056 ) ( 9,0.0063 ) ( 10,0.0069 ) ( 11,0.0076 ) ( 12,0.0083 ) ( 13,0.009 ) ( 14,0.0097 ) ( 15,0.0104 ) ( 16,0.0111 ) ( 17,0.0118 ) ( 18,0.0125 ) ( 19,0.0132 ) ( 20,0.0139 ) ( 20,1 ) ( 19,1 ) ( 18,1 ) ( 17,1 ) ( 16,1 ) ( 15,1 ) ( 14,1 ) ( 13,1 ) ( 12,1 ) ( 11,1 ) ( 10,1 ) ( 9,1 ) ( 8,1 ) ( 7,0.0077 ) ( 6,0.0063 ) ( 5,0.0053 ) ( 4,0.0044 ) ( 3,0.0044 ) ( 2,0.0038 ) ( 1,0.003 )
};
\addplot[rgb color={25, 0, 213}, very thick,  mark size=2pt, mark=., fill opacity=0.5, draw opacity=0.75] coordinates {
( 1,7e-04 ) ( 2,0.0014 ) ( 3,0.0021 ) ( 4,0.0028 ) ( 5,0.0035 ) ( 6,0.0042 ) ( 7,0.0049 ) ( 8,0.0056 ) ( 9,0.0063 ) ( 10,0.0069 ) ( 11,0.0076 ) ( 12,0.0083 ) ( 13,0.009 ) ( 14,0.0097 ) ( 15,0.0104 ) ( 16,0.0111 ) ( 17,0.0118 ) ( 18,0.0125 ) ( 19,0.0132 ) ( 20,0.0139 ) ( 20,1 ) ( 19,1 ) ( 18,1 ) ( 17,1 ) ( 16,1 ) ( 15,1 ) ( 14,1 ) ( 13,1 ) ( 12,1 ) ( 11,1 ) ( 10,1 ) ( 9,1 ) ( 8,1 ) ( 7,1 ) ( 6,0.006 ) ( 5,0.0049 ) ( 4,0.0041 ) ( 3,0.0034 ) ( 2,0.0034 ) ( 1,0.0034 )
};
\addplot[rgb color={50, 0, 172}, very thick,  mark size=2pt, mark=., fill opacity=0.5, draw opacity=0.75] coordinates {
( 1,7e-04 ) ( 2,0.0014 ) ( 3,0.0021 ) ( 4,0.0028 ) ( 5,0.0035 ) ( 6,0.0042 ) ( 7,0.0049 ) ( 8,0.0056 ) ( 9,0.0063 ) ( 10,0.0069 ) ( 11,0.0076 ) ( 12,0.0083 ) ( 13,0.009 ) ( 14,0.0097 ) ( 15,0.0104 ) ( 16,0.0111 ) ( 17,0.0118 ) ( 18,0.0125 ) ( 19,0.0132 ) ( 20,0.0139 ) ( 20,1 ) ( 19,1 ) ( 18,1 ) ( 17,1 ) ( 16,1 ) ( 15,1 ) ( 14,1 ) ( 13,1 ) ( 12,1 ) ( 11,1 ) ( 10,1 ) ( 9,1 ) ( 8,1 ) ( 7,1 ) ( 6,1 ) ( 5,0.0046 ) ( 4,0.0037 ) ( 3,0.0034 ) ( 2,0.0029 ) ( 1,0.0024 )
};
\addplot[rgb color={75, 0, 130}, very thick,  mark size=2pt, mark=., fill opacity=0.5, draw opacity=0.75] coordinates {
( 1,7e-04 ) ( 2,0.0014 ) ( 3,0.0021 ) ( 4,0.0028 ) ( 5,0.0035 ) ( 6,0.0042 ) ( 7,0.0049 ) ( 8,0.0056 ) ( 9,0.0063 ) ( 10,0.0069 ) ( 11,0.0076 ) ( 12,0.0083 ) ( 13,0.009 ) ( 14,0.0097 ) ( 15,0.0104 ) ( 16,0.0111 ) ( 17,0.0118 ) ( 18,0.0125 ) ( 19,0.0132 ) ( 20,0.0139 ) ( 20,1 ) ( 19,1 ) ( 18,1 ) ( 17,1 ) ( 16,1 ) ( 15,1 ) ( 14,1 ) ( 13,1 ) ( 12,1 ) ( 11,1 ) ( 10,1 ) ( 9,1 ) ( 8,1 ) ( 7,1 ) ( 6,1 ) ( 5,1 ) ( 4,0.0034 ) ( 3,0.0026 ) ( 2,0.0025 ) ( 1,0.0023 )
};
\addplot[rgb color={129, 43, 166}, very thick,  mark size=2pt, mark=., fill opacity=0.5, draw opacity=0.75] coordinates {
( 1,7e-04 ) ( 2,0.0014 ) ( 3,0.0021 ) ( 4,0.0028 ) ( 5,0.0035 ) ( 6,0.0042 ) ( 7,0.0049 ) ( 8,0.0056 ) ( 9,0.0063 ) ( 10,0.0069 ) ( 11,0.0076 ) ( 12,0.0083 ) ( 13,0.009 ) ( 14,0.0097 ) ( 15,0.0104 ) ( 16,0.0111 ) ( 17,0.0118 ) ( 18,0.0125 ) ( 19,0.0132 ) ( 20,0.0139 ) ( 20,1 ) ( 19,1 ) ( 18,1 ) ( 17,1 ) ( 16,1 ) ( 15,1 ) ( 14,1 ) ( 13,1 ) ( 12,1 ) ( 11,1 ) ( 10,1 ) ( 9,1 ) ( 8,1 ) ( 7,1 ) ( 6,1 ) ( 5,1 ) ( 4,1 ) ( 3,0.0024 ) ( 2,0.0018 ) ( 1,0.0017 )
};
\addplot[rgb color={184, 87, 202}, very thick,  mark size=2pt, mark=., fill opacity=0.5, draw opacity=0.75] coordinates {
( 1,7e-04 ) ( 2,0.0014 ) ( 3,0.0021 ) ( 4,0.0028 ) ( 5,0.0035 ) ( 6,0.0042 ) ( 7,0.0049 ) ( 8,0.0056 ) ( 9,0.0063 ) ( 10,0.0069 ) ( 11,0.0076 ) ( 12,0.0083 ) ( 13,0.009 ) ( 14,0.0097 ) ( 15,0.0104 ) ( 16,0.0111 ) ( 17,0.0118 ) ( 18,0.0125 ) ( 19,0.0132 ) ( 20,0.0139 ) ( 20,1 ) ( 19,1 ) ( 18,1 ) ( 17,1 ) ( 16,1 ) ( 15,1 ) ( 14,1 ) ( 13,1 ) ( 12,1 ) ( 11,1 ) ( 10,1 ) ( 9,1 ) ( 8,1 ) ( 7,1 ) ( 6,1 ) ( 5,1 ) ( 4,1 ) ( 3,1 ) ( 2,0.0015 ) ( 1,0.001 )
};
\addplot[color=black, very thick,  mark size=2pt, mark=.] coordinates {
( 1,7e-04 ) ( 2,0.0014 ) ( 3,0.0021 ) ( 4,0.0028 ) ( 5,0.0035 ) ( 6,0.0042 ) ( 7,0.0049 ) ( 8,0.0056 ) ( 9,0.0063 ) ( 10,0.0069 ) ( 11,0.0076 ) ( 12,0.0083 ) ( 13,0.009 ) ( 14,0.0097 ) ( 15,0.0104 ) ( 16,0.0111 ) ( 17,0.0118 ) ( 18,0.0125 ) ( 19,0.0132 ) ( 20,0.0139 )
};
\end{axis}
\end{tikzpicture}
 \caption{\revise{Regions of ordered p-value sequences, for $m=20$ and $\alpha=0.05$.
The $x$-axis gives the rank and the $y$-axis gives the ordered p-values $\pvalue_{(1)}\leq \cdots \leq \pvalue_{(m)}$.
The black curve is the BY rejection boundary, i.e., an ordered sequence that lies entirely above the black curve leads BY to reject no hypotheses.
The colored regions show ordered sequences for which $\textcBY$ rejects at least $k=2, 3, \ldots, 13,15,19$ hypotheses while BY rejects none.
These regions were constructed by greedily choosing the largest $\pvalue_{(m)}, \ldots, \pvalue_{(1)}$, in that order, to achieve the minimal desired local e-values.}} \label{fig:BY}
\end{figure}
 \begin{figure}[!ht]
\centering
\begin{tikzpicture}[scale=.7]
\begin{axis}[
	xmin = 1,
	xmax = 20,
	xtick = {5,10,15,20},
	ymin = 0,
	ymax = 0.1,
	ytick = {0,0.02,0.04,0.06,0.08,0.1},
	yticklabels = {0,0.02,0.04,0.06,0.08,0.1},
	 ytick scale label code/.code={},
	ylabel=$p$-value,
	xlabel=rank,
        height=10cm,
	width=10cm,
		legend style={legend columns=1, at={(0.275,0.975)}}
]
 \addplot[color=blue, very thick,  mark size=2pt, mark=*, fill=blue, fill opacity=0.5, draw opacity=0.75] coordinates {
( 1,0 ) ( 2,0 ) ( 3,0 ) ( 4,0 ) ( 5,0 ) ( 6,0 ) ( 7,0.003 ) ( 8,0.0035 ) ( 9,0.0039 ) ( 10,0.0044 ) ( 11,0.0048 ) ( 12,0.0052 ) ( 13,0.0057 ) ( 14,0.0061 ) ( 15,0.0065 ) ( 16,0.007 ) ( 17,0.0074 ) ( 18,0.0078 ) ( 19,0.0083 ) ( 20,0.0087 ) ( 20,1 ) ( 19,1 ) ( 18,1 ) ( 17,1 ) ( 16,1 ) ( 15,1 ) ( 14,1 ) ( 13,1 ) ( 12,1 ) ( 11,1 ) ( 10,1 ) ( 9,1 ) ( 8,1 ) ( 7,0.0044 ) ( 6,0.002 ) ( 5,0.002 ) ( 4,0.002 ) ( 3,0.002 ) ( 2,0.002 ) ( 1,0.002 )
};
 \addlegendentry{$k=7$}
 \addplot[color=red, very thick,  mark size=2pt, mark=*, fill=red, fill opacity=0.5, draw opacity=0.75] coordinates {
( 1,0 ) ( 2,0 ) ( 3,0 ) ( 4,0 ) ( 5,0 ) ( 6,0 ) ( 7,0 ) ( 8,0 ) ( 9,0 ) ( 10,0 ) ( 11,0 ) ( 12,0 ) ( 13,0 ) ( 14,0 ) ( 15,0 ) ( 16,0 ) ( 17,0 ) ( 18,0 ) ( 19,0.0083 ) ( 20,0.0087 ) ( 20,1 ) ( 19,0.0827 ) ( 18,0.0276 ) ( 17,0.0276 ) ( 16,0.0083 ) ( 15,0.0083 ) ( 14,0.0083 ) ( 13,0.0083 ) ( 12,0.0083 ) ( 11,0.0083 ) ( 10,0.0083 ) ( 9,0.0083 ) ( 8,0.0083 ) ( 7,0.0083 ) ( 6,0.0083 ) ( 5,0.0083 ) ( 4,6e-04 ) ( 3,6e-04 ) ( 2,6e-04 ) ( 1,6e-04 )
};
\addlegendentry{$k=19$}
\addplot[black, dashed, very thick, domain=0:20] {0.05*x/20};
\addlegendentry{BH}
\end{axis}
\end{tikzpicture}~
\begin{tikzpicture}[scale=.7, define rgb/.code={\definecolor{mycolor}{RGB}{#1}},
                    rgb color/.style={define rgb={#1},mycolor}]
\begin{axis}[
xmin = 1,
	xmax = 20,
	xtick = {5,10,15,20},
	ymin = 0,
	ymax = 0.1,
	ytick = {0,0.02,0.04,0.06,0.08,0.1},
	yticklabels = {0,0.02,0.04,0.06,0.08,0.1},
	 ytick scale label code/.code={},
	ylabel=,
	xlabel=rank,
        height=10cm,
	width=10cm
]
\addplot[rgb color={129, 43, 166}, very thick,  mark size=2pt, mark=., fill opacity=0.5, draw opacity=0.75] coordinates {
( 1,0 ) ( 2,0 ) ( 3,0.0013 ) ( 4,0.0017 ) ( 5,0.0022 ) ( 6,0.0026 ) ( 7,0.003 ) ( 8,0.0035 ) ( 9,0.0039 ) ( 10,0.0044 ) ( 11,0.0048 ) ( 12,0.0052 ) ( 13,0.0057 ) ( 14,0.0061 ) ( 15,0.0065 ) ( 16,0.007 ) ( 17,0.0074 ) ( 18,0.0078 ) ( 19,0.0083 ) ( 20,0.0087 ) ( 20,1 ) ( 19,1 ) ( 18,1 ) ( 17,1 ) ( 16,1 ) ( 15,1 ) ( 14,1 ) ( 13,1 ) ( 12,1 ) ( 11,1 ) ( 10,1 ) ( 9,1 ) ( 8,1 ) ( 7,1 ) ( 6,1 ) ( 5,1 ) ( 4,1 ) ( 3,0.0015 ) ( 2,7e-04 ) ( 1,7e-04 )
};
\addplot[rgb color={50, 0, 172}, very thick,  mark size=2pt, mark=., fill opacity=0.5, draw opacity=0.75] coordinates {
( 1,0 ) ( 2,0 ) ( 3,0 ) ( 4,0 ) ( 5,0.0022 ) ( 6,0.0026 ) ( 7,0.003 ) ( 8,0.0035 ) ( 9,0.0039 ) ( 10,0.0044 ) ( 11,0.0048 ) ( 12,0.0052 ) ( 13,0.0057 ) ( 14,0.0061 ) ( 15,0.0065 ) ( 16,0.007 ) ( 17,0.0074 ) ( 18,0.0078 ) ( 19,0.0083 ) ( 20,0.0087 ) ( 20,1 ) ( 19,1 ) ( 18,1 ) ( 17,1 ) ( 16,1 ) ( 15,1 ) ( 14,1 ) ( 13,1 ) ( 12,1 ) ( 11,1 ) ( 10,1 ) ( 9,1 ) ( 8,1 ) ( 7,1 ) ( 6,1 ) ( 5,0.0027 ) ( 4,0.0013 ) ( 3,0.0013 ) ( 2,0.0013 ) ( 1,0.0013 )
};
\addplot[rgb color={0, 0, 255}, very thick,  mark size=2pt, mark=., fill opacity=0.5, draw opacity=0.75] coordinates {
( 1,0 ) ( 2,0 ) ( 3,0 ) ( 4,0 ) ( 5,0 ) ( 6,0 ) ( 7,0.003 ) ( 8,0.0035 ) ( 9,0.0039 ) ( 10,0.0044 ) ( 11,0.0048 ) ( 12,0.0052 ) ( 13,0.0057 ) ( 14,0.0061 ) ( 15,0.0065 ) ( 16,0.007 ) ( 17,0.0074 ) ( 18,0.0078 ) ( 19,0.0083 ) ( 20,0.0087 ) ( 20,1 ) ( 19,1 ) ( 18,1 ) ( 17,1 ) ( 16,1 ) ( 15,1 ) ( 14,1 ) ( 13,1 ) ( 12,1 ) ( 11,1 ) ( 10,1 ) ( 9,1 ) ( 8,1 ) ( 7,0.0044 ) ( 6,0.002 ) ( 5,0.002 ) ( 4,0.002 ) ( 3,0.002 ) ( 2,0.002 ) ( 1,0.002 )
};
\addplot[rgb color={0, 85, 85}, very thick,  mark size=2pt, mark=., fill opacity=0.5, draw opacity=0.75] coordinates {
( 1,0 ) ( 2,0 ) ( 3,0 ) ( 4,0 ) ( 5,0 ) ( 6,0 ) ( 7,0 ) ( 8,0 ) ( 9,0.0039 ) ( 10,0.0044 ) ( 11,0.0048 ) ( 12,0.0052 ) ( 13,0.0057 ) ( 14,0.0061 ) ( 15,0.0065 ) ( 16,0.007 ) ( 17,0.0074 ) ( 18,0.0078 ) ( 19,0.0083 ) ( 20,0.0087 ) ( 20,1 ) ( 19,1 ) ( 18,1 ) ( 17,1 ) ( 16,1 ) ( 15,1 ) ( 14,1 ) ( 13,1 ) ( 12,1 ) ( 11,1 ) ( 10,1 ) ( 9,0.0065 ) ( 8,0.003 ) ( 7,0.003 ) ( 6,0.003 ) ( 5,0.003 ) ( 4,0.003 ) ( 3,0.003 ) ( 2,0.003 ) ( 1,0.003 )
};
\addplot[rgb color={85, 170, 0}, very thick,  mark size=2pt, mark=., fill opacity=0.5, draw opacity=0.75] coordinates {
( 1,0 ) ( 2,0 ) ( 3,0 ) ( 4,0 ) ( 5,0 ) ( 6,0 ) ( 7,0 ) ( 8,0 ) ( 9,0 ) ( 10,0 ) ( 11,0.0048 ) ( 12,0.0052 ) ( 13,0.0057 ) ( 14,0.0061 ) ( 15,0.0065 ) ( 16,0.007 ) ( 17,0.0074 ) ( 18,0.0078 ) ( 19,0.0083 ) ( 20,0.0087 ) ( 20,1 ) ( 19,1 ) ( 18,1 ) ( 17,1 ) ( 16,1 ) ( 15,1 ) ( 14,1 ) ( 13,1 ) ( 12,1 ) ( 11,0.0096 ) ( 10,0.0044 ) ( 9,0.0044 ) ( 8,0.0044 ) ( 7,0.0044 ) ( 6,0.0044 ) ( 5,0.0044 ) ( 4,0.0044 ) ( 3,0.0044 ) ( 2,0.0044 ) ( 1,0.0044 )
};
\addplot[rgb color={255, 255, 0}, very thick,  mark size=2pt, mark=., fill opacity=0.5, draw opacity=0.75] coordinates {
( 1,0 ) ( 2,0 ) ( 3,0 ) ( 4,0 ) ( 5,0 ) ( 6,0 ) ( 7,0 ) ( 8,0 ) ( 9,0 ) ( 10,0 ) ( 11,0 ) ( 12,0 ) ( 13,0.0057 ) ( 14,0.0061 ) ( 15,0.0065 ) ( 16,0.007 ) ( 17,0.0074 ) ( 18,0.0078 ) ( 19,0.0083 ) ( 20,0.0087 ) ( 20,1 ) ( 19,1 ) ( 18,1 ) ( 17,1 ) ( 16,1 ) ( 15,1 ) ( 14,1 ) ( 13,0.0141 ) ( 12,0.0063 ) ( 11,0.0063 ) ( 10,0.0063 ) ( 9,0.0063 ) ( 8,0.0063 ) ( 7,0.0063 ) ( 6,0.0063 ) ( 5,0.0063 ) ( 4,7e-04 ) ( 3,7e-04 ) ( 2,7e-04 ) ( 1,7e-04 )
};
\addplot[rgb color={255, 195, 0}, very thick,  mark size=2pt, mark=., fill opacity=0.5, draw opacity=0.75] coordinates {
( 1,0 ) ( 2,0 ) ( 3,0 ) ( 4,0 ) ( 5,0 ) ( 6,0 ) ( 7,0 ) ( 8,0 ) ( 9,0 ) ( 10,0 ) ( 11,0 ) ( 12,0 ) ( 13,0 ) ( 14,0 ) ( 15,0.0065 ) ( 16,0.007 ) ( 17,0.0074 ) ( 18,0.0078 ) ( 19,0.0083 ) ( 20,0.0087 ) ( 20,1 ) ( 19,1 ) ( 18,1 ) ( 17,1 ) ( 16,1 ) ( 15,0.0218 ) ( 14,0.0093 ) ( 13,0.0093 ) ( 12,0.0093 ) ( 11,0.0093 ) ( 10,0.0093 ) ( 9,0.0093 ) ( 8,0.0013 ) ( 7,0.0013 ) ( 6,0.0013 ) ( 5,0.0013 ) ( 4,0.0013 ) ( 3,0.0013 ) ( 2,0.0013 ) ( 1,0.0013 )
};
\addplot[rgb color={255, 110, 0}, very thick,  mark size=2pt, mark=., fill opacity=0.5, draw opacity=0.75] coordinates {
( 1,0 ) ( 2,0 ) ( 3,0 ) ( 4,0 ) ( 5,0 ) ( 6,0 ) ( 7,0 ) ( 8,0 ) ( 9,0 ) ( 10,0 ) ( 11,0 ) ( 12,0 ) ( 13,0 ) ( 14,0 ) ( 15,0 ) ( 16,0 ) ( 17,0 ) ( 18,0 ) ( 19,0.0083 ) ( 20,0.0087 ) ( 20,1 ) ( 19,0.0827 ) ( 18,0.0276 ) ( 17,0.0276 ) ( 16,0.0083 ) ( 15,0.0083 ) ( 14,0.0083 ) ( 13,0.0083 ) ( 12,0.0083 ) ( 11,0.0083 ) ( 10,0.0083 ) ( 9,0.0083 ) ( 8,0.0083 ) ( 7,0.0083 ) ( 6,0.0083 ) ( 5,0.0083 ) ( 4,6e-04 ) ( 3,6e-04 ) ( 2,6e-04 ) ( 1,6e-04 )
};
\addplot[rgb color={255, 0, 0}, very thick,  mark size=2pt, mark=., fill opacity=0.5, draw opacity=0.75] coordinates {
( 1,0 ) ( 2,0 ) ( 3,0 ) ( 4,0 ) ( 5,0 ) ( 6,0 ) ( 7,0 ) ( 8,0 ) ( 9,0 ) ( 10,0 ) ( 11,0 ) ( 12,0 ) ( 13,0 ) ( 14,0 ) ( 15,0 ) ( 16,0 ) ( 17,0 ) ( 18,0 ) ( 19,0.0083 ) ( 20,0.0087 ) ( 20,1 ) ( 19,0.0827 ) ( 18,0.0276 ) ( 17,0.0276 ) ( 16,0.0083 ) ( 15,0.0083 ) ( 14,0.0083 ) ( 13,0.0083 ) ( 12,0.0083 ) ( 11,0.0083 ) ( 10,0.0083 ) ( 9,0.0083 ) ( 8,0.0083 ) ( 7,0.0083 ) ( 6,0.0083 ) ( 5,0.0083 ) ( 4,6e-04 ) ( 3,6e-04 ) ( 2,6e-04 ) ( 1,6e-04 )
};

\end{axis}
\end{tikzpicture}
 \caption{\revise{Regions of ordered p-value sequences, for $m=20$ and $\alpha=0.05$.
The $x$-axis gives the rank and the $y$-axis gives the ordered p-values $p_{(1)}\leq \cdots \leq p_{(m)}$.
The colored regions show ordered sequences for which $\textcSu$ rejects at least $k=3, 5, \ldots,19$ hypotheses while Su rejects none.
These regions were constructed by greedily choosing the largest $p_{(m)}, \ldots, p_{(1)}$, in that order, to achieve the minimal desired local e-values.
The black dashed line is the rejection threshold of the BH procedure, i.e., the BH procedure rejects the $r$ smallest p-values if $r$ is the largest rank such that $p_{(r)}$ is below this line.
The upper bound on the region of p-value sequences where the $\textcSu$ procedure contains p-value sequences that make 19 or more rejections is higher than the BH threshold, which means that the BH procedure makes strictly fewer rejections for those sequences of p-values.}} \label{fig:Su}
\end{figure}
 We also provide an example in which the $\cSu$ procedure rejects all hypotheses while the original Su procedure rejects only one.
\begin{example}\label{ex: closed su example}
Let $m \geq 3$. Then $\rejset^\Su_\alpha = [1]$, but $\rejset^\cSu_\alpha = [m]$ under the following event:
\[
    \pvalue_1 = \frac{ \alpha}{m\ell_\alpha}, \quad \pvalue_2 = \cdots = \pvalue_m = \frac{ \alpha}{\ell_\alpha} + \delta, \quad 0 < \delta \leq \frac{\alpha}{(m-1)\ell_\alpha}.
\]
\end{example}
\begin{proof}
That Su rejects only $[1]$ is immediate from the definition: $\pvalue_i > \alpha i / (m\ell_\alpha)$ for $i = 2, \ldots, m$. The $\textcSu$ procedure rejects $[m]$, since $\mathbf{e}_S = \alpha^{-1}$ if $1 \in S$. If $1 \notin S$, the constraint on $\delta$ ensures that $\alpha \mathbf{e}_S = \alpha / (\alpha + \ell_\alpha \delta) \geq (m-1)/m \geq |S|/m$.\end{proof}
 }
\ifarver{
\section{Example where $\textCeBH$ with calibrated p-values is more powerful than $\textcBY$}
}
{
\subsection{Example where $\textCeBH$ with calibrated p-values is more powerful than $\textcBY$}
}
\label{sec: cebh calibrated vs cby example}
We now show that neither $\cBY$ nor $\CeBH$ applied to calibrated p-values dominates the other in terms of power. It is clear from \Cref{tab:by_real_data_discoveries} that $\cBY$ can dominate $\CeBH$ with calibrated p-values in practice. Here, we give an example where $\CeBH$ with calibrated p-values dominates $\cBY$.
We compare the local e-values used in the $\cBY$ procedure 
\begin{align}
\evalue_S &= \frac{1}{|S|} \sum_{i \in S} \frac{|S| \cdot \mathbf{1}\{h_{|S|}\pvalue_i \leq \alpha\}}{\alpha \lceil |S|h_{|S|}\pvalue_i/\alpha \rceil},
\end{align}
with the local e-values used in the $\CeBH$ procedure with calibrated p-values
\begin{align}
\evalue'_S &= \frac{1}{|S|} \sum_{i \in S} \frac{m \cdot \mathbf{1}\{h_m \pvalue_i \leq \alpha\}}{\alpha \lceil mh_m \pvalue_i/\alpha \rceil}.
\end{align}

Let $m = 5$ and $|S| = m - 1$:
\[
0 < \pvalue_1 < \frac{\alpha}{mh_m}, \quad \frac{2\alpha}{|S|h_{|S|}} < \pvalue_2 < \frac{3\alpha}{mh_m}, \quad \frac{3\alpha}{|S|h_{|S|}} < \pvalue_3 < \frac{\alpha}{h_m}, \quad \alpha < \pvalue_4 < \pvalue_5 < 1.
\]

\revise{Then the largest rejection set for $\CeBH$ with calibrated p-values (using $\evalue'_S$) is $[3]$, while the largest rejection set for $\cBY$ (using $\evalue_S$) is $[2]$.} To see this, take $S = [m] \setminus \{1\}$. We have:
\begin{align}
\alpha \evalue_S &= \frac{1}{\lceil |S|h_{|S|}\pvalue_2/\alpha \rceil} + \frac{1}{\lceil |S|h_{|S|}\pvalue_3/\alpha \rceil} = \frac{1}{3} + \frac{1}{4} = \frac{7}{12} < \frac{|[3] \cap S|}{|[3]|} = \frac{2}{3}, \\
\alpha \evalue'_S &= \frac{m}{|S|} \left( \frac{1}{\lceil mh_m \pvalue_2/\alpha \rceil} + \frac{1}{\lceil mh_m \pvalue_3/\alpha \rceil} \right) = \frac{5}{4} \left( \frac{1}{3} + \frac{1}{5} \right) = \frac{2}{3} \geq \frac{|[3] \cap S|}{|[3]|} = \frac{2}{3}.
\end{align}

It is tedious but straightforward to verify that for all other subsets $S \neq [m] \setminus \{1\}$ we have $\alpha \evalue'_S$ larger or equal than $|[3] \cap S|/|[3]|$.

\ifarver{}{
\section{Benjamini-Hochberg and its adaptive variants}\label{sec: bh and variants}

The seminal FDR control method by \citet{benjamini_controlling_false_1995} remains the most popular by far. It rejects the following discovery set:
$$
\rejset_\alpha^{\BH} \coloneqq \rejset^{(\mathbf{r}_\alpha^\BH)}, \text{ where }\mathbf{r}_\alpha^\BH \coloneqq \max\{r \in [m]: \pvalue_{(r)} \leq \alpha r / m \},$$
\revise{where $\mathbf{r}_\alpha^\BH$ is taken as 0 if the maximum does not exist. However, it has some puzzling aspects. First, BH is not guaranteed under arbitrary dependence. Its classical finite-sample validity result holds under PRDS \eqref{eq: PRDS}, which constrains the joint distribution of the p-values of true and false hypotheses, although other sufficient conditions are known \citep{goeman_uniform_improvement_2026}. Such an assumption is unusual among multiple testing procedures, which generally make assumptions only on the distribution of true hypotheses, and \citet{su2018fdr} made a forceful argument against it. Second, \BH\ is known to be inadmissible under PRDS. The minimally adaptive BH (MABH) procedure constructed by \citet{solari2017minimally} gives a small uniform improvement; more recently, \citet{goeman_uniform_improvement_2026} constructed a closed BH procedure that uniformly improves BH under PRDS and a weaker sufficient condition. MABH rejects}
\begin{align}
    \mathbf{r}_\alpha^{\MABH} \coloneqq \max\left\{i \in [m]: i^{-1}\cdot \pvalue_{(i)} \leq \frac{\alpha}{m - 1}\right\}, \rejset_\alpha^{\MABH} \coloneqq \begin{cases}
        \emptyset & \text{if }\rejset_\alpha^\BH = \emptyset;\\
        \rejset^{(\mathbf{r}_\alpha^{\MABH})} & \text{otherwise.}
    \end{cases}
\end{align}
For the \BH\ procedure, we can define the following local e-value
\begin{equation} \label{eq: BH local e-value}
\evalue_S = \frac{1}{ |S|} \sum_{i \in S} \widetilde{\evalue}_i \text{ where } \widetilde{\evalue}_i =  \frac{m}{\alpha |\rejset_\alpha^\BH|} \mathbf{1}\Big\{\pvalue_i \leq \frac{\alpha |\rejset_\alpha^\BH|}{m} \Big\}.
\end{equation}
\revise{\citet{li_note_e-values_2025} showed that $\widetilde\evalue_i$ is a compound e-value under PRDS and that \eBH\ applied to these compound e-values reconstructs \BH, but they did not directly show it is an e-value as well. The validity of $\widetilde\evalue_i$, and hence of $\evalue_S$, can also be seen directly from the dependency control condition of \citet{blanchard_two_simple_2008}. For each null $i$, take $U=\pvalue_i$, $V=|\rejset_\alpha^\BH|$, and the linear shape function $\beta(v)=v$. Proposition 3.6 of \citet{blanchard_two_simple_2008} gives
\[
\expect_\pdist\left[\frac{\mathbf{1}\{\pvalue_i\leq \alpha |\rejset_\alpha^\BH|/m\}}{|\rejset_\alpha^\BH|}\right]\leq \frac{\alpha}{m}
\]
under PRDS, with the ratio defined as zero when $|\rejset_\alpha^\BH|=0$. Multiplying by $m/\alpha$ proves that $\widetilde\evalue_i$ is an e-value. Averaging over $i\in S$ then proves that $\evalue_S$ is an e-value for $H_S$.}
\revise{We then define $\rejcol_\alpha^{\cBH}$ as the e-Closure procedure resulting from this e-collection.}

The e-collection derived from \eqref{eq: BH local e-value} depends on all $m$ p-values through $|\rejset_\alpha^\BH|$ instead of only on p-values pertaining to hypotheses $S$ (as is the case for the e-collections of $\cBY$ and $\cSu$). Thus, it is clear that an assumption on the joint distribution of p-values of true and false hypotheses is needed.

We can also define the corresponding local e-values for the minimally adaptive BH procedure as follows.
\[
\evalue_{[m]} = 1\{\rejset^\BH_\alpha \neq \emptyset\}/\alpha, \qquad
\evalue_S = \frac{m-1}{ |S|} \sum_{i \in S} \frac{1}{\alpha |\rejset_\alpha^\MABH|} \mathbf{1}\Big\{\pvalue_i \leq \frac{\alpha |\rejset_\alpha^\MABH|}{m-1} \Big\}\text{ for }S \subset [m]\label{eq: mabh local evalue}
\]
\begin{proposition} \label{thm: e for MABH}
    \revise{Under PRDS, $\ecol = (\evalue_S)_{S\in 2^{[m]}}$ with e-values defined in \eqref{eq: mabh local evalue} is a valid e-collection.}
\end{proposition}
We define $\rejcol_\alpha^{\cMABH}$ as the e-Closure procedure resulting from the e-collection comprising these local e-values.

\revise{This approach generalizes to other adaptive versions of the BH procedure. Unlike MABH, the validity claim below assumes that the p-values $(\pvalue_1, \dots, \pvalue_m)$ are independent and that the estimator satisfies the stated monotonicity and moment conditions.} For any $i \in N_\pdist$, define
\begin{align}
\hat{\boldsymbol{\pi}}_0^i \coloneqq \hat\pi(\pvalue_1, \dots, \pvalue_{i - 1}, 0, \pvalue_{i + 1}, \dots, \pvalue_m), \qquad \expect_\pdist\left(\frac{1}{\hat{\boldsymbol{\pi}}_0^i}\right) \leq \frac{m}{|N_\pdist|}, \label{eq: pi0 condition}
\end{align}
where $\hat\pi$ is an estimator of the null proportion and is also required to be coordinatewise nondecreasing.
Let $\hat{\boldsymbol{\pi}}_0 \coloneqq \max_{i \in [m]} \hat{\boldsymbol{\pi}}_0^i$. Now, we define the adaptive BH discovery set as
\begin{align}
    \mathbf{r}_\alpha^{\adaBH} \coloneqq \max \left(\left\{i \in [m]: i^{-1} \cdot \pvalue_{(i)} \leq \frac{\alpha}{\hat{\boldsymbol{\pi}}_0m}\right\} \cup \{0\}\right), \qquad \rejset_\alpha^{\adaBH} \coloneqq \rejset^{(\mathbf{r}_\alpha^\adaBH)}.
\end{align}
Storey's procedure \citep{storey_direct_approach_2002a, storey2004strong, benjamini2006adaptive}, which sets $\hat\pi(p_1, \dots, p_m) = (1 + \sum_{i=1}^m \mathbf{1}\{p_i > \lambda\}) / ((1 - \lambda)m)$ for a fixed tuning parameter $\lambda \in (0, 1)$, is an instance of this adaptive BH procedure. Methods of this type have been well studied because they can estimate the null proportion and increase power when the number of non-null hypotheses is substantial \citep{benjamini2006adaptive,sarkar_methods_controlling_2008,blanchard2009adaptive,dohler_unified_class_2023,gao_adaptive_null_2025}.

We can define the e-collection for improving the adaptive BH procedure. We now formulate the compound e-values and resulting local e-values that correspond to an adaptive procedure.
\begin{align}
    \mathbf{e}_S = \frac{1}{m}\sum_{i \in S}\widetilde\evalue_i \text{ where }\widetilde\evalue_i = \frac{m}{\alpha (|\rejset_\alpha^{\adaBH}| \vee 1)} \mathbf{1}\Big\{\pvalue_i \leq \frac{\alpha (|\rejset_\alpha^{\adaBH}| \vee 1)}{m \hat{\boldsymbol{\pi}}_0} \Big\}, \label{eq: adaBH local e-value}
\end{align}
\revise{The validity of these compound e-values follows directly from the argument used to prove Theorem 11 of \citet{blanchard2009adaptive}. Write $G(\mathbf{p})=1/\hat{\boldsymbol{\pi}}_0(\mathbf{p})$. For each null $i$, condition on the other p-values and apply their Lemma 27 with $U=\pvalue_i$, $g(U)=|\rejset_\alpha^{\adaBH}(\pvalue_1,\ldots,\pvalue_{i-1},U,\pvalue_{i+1},\ldots,\pvalue_m)|\vee 1$, and $c=\alpha G(\mathbf{p}_{0,i})/m$, where $\mathbf{p}_{0,i}$ replaces $\pvalue_i$ by zero. Monotonicity of $G$ in each of its arguments yields
\[
\expect_\pdist[\widetilde\evalue_i]\leq \expect_\pdist[G(\mathbf{p}_{0,i})]\leq \frac{m}{|N_\pdist|}.
\]
Consequently, under $H_S$, $\expect_\pdist[\mathbf{e}_S]\leq |S|/|N_\pdist|\leq 1$, proving that the local e-values are valid. Equivalently, Theorem 11 establishes FDR control for this adaptive BH procedure, after which \eqref{eq: universal compound evalue} gives the same compound e-values.}
The adaptive closed BH procedure, $\cadaBH$, is then defined as $\rejcol^{\cadaBH}_\alpha \coloneqq \Rcal^\FDR_\alpha(\mathbf{E})$ for the e-collection based on \eqref{eq: adaBH local e-value}. \revise{The theorem below shows that this e-collection exactly reconstructs adaptive BH, with no additional nonempty discovery sets.}
\revise{Unfortunately, the closed \BH\ procedure $\cBH$ resulting from \eqref{eq: BH local e-value} never gives a larger rejected set. Moreover, it gives no nontrivial simultaneity unless $\rejset_\alpha^{\BH}= [m]$. On that full-rejection branch, the simultaneous collection contains every singleton; by \Cref{thm: FWER}, the global selector formed by taking the union of all available singletons controls FWER. This is an unconditional guarantee for that selector, not a claim that BH conditionally controls FWER given the event $\rejset_\alpha^{\BH}=[m]$. In contrast, the corresponding closed MABH and adaptive BH procedures have no such exception: they exactly reconstruct their original procedures even when all hypotheses are rejected.} \begin{theorem}\label{thm: BH-no-simultaneity}
\revise{Assume PRDS for the BH and MABH claims. For the adaptive BH claim, assume independent p-values and that $\hat{\boldsymbol\pi}_0$ satisfies \eqref{eq: pi0 condition} and is coordinatewise nondecreasing. For the BH procedure, if $\rejset_\alpha^\BH\neq [m]$, then $\rejcol_\alpha^{\cBH}=\{\emptyset,\rejset_\alpha^\BH\}$. Otherwise, if $\rejset_\alpha^\BH=[m]$, then $\rejcol_\alpha^{\cBH}=2^{[m]}$. For minimally adaptive BH and adaptive BH, we always have $\rejcol_\alpha^{\cMABH}=\{\emptyset,\rejset_\alpha^\MABH\}$ and $\rejcol_\alpha^{\cadaBH}=\{\emptyset,\rejset_\alpha^\adaBH\}$.}
\end{theorem}
\begin{proof}
    First, consider the case where $\rejset = \rejset^\BH_\alpha$ and $\rejcol = \rejcol^\cBH_\alpha$. We can rewrite each local e-value as follows:
\begin{eqnarray*}
\mathbf{e}_S &=&  \frac{m |\rejset \cap S|}{\alpha |S|(|\rejset| \vee 1)}.
\end{eqnarray*}
If
$\rejset =[m]$, then for all $S$, $ \mathbf{e}_S = 1/ \alpha$, which is always greater than or equal to $\textnormal{FDP}_S(R)/\alpha$ for every $R$ and $S$, thus  $\rejcol =2^{[m]}$.

\revise{If $\rejset=\emptyset$, all the local e-values are zero and the conclusion is immediate.} Suppose $\rejset \neq \emptyset$ and $\rejset \neq [m]$.
We show that for every non-empty $R \neq \rejset$, there exists an $S$ such that
\begin{eqnarray*}
\frac{m|\rejset \cap S|}{ |S||\rejset|}&<& \frac{|R \cap S|}{ |R|}.
\end{eqnarray*}
Let $R \subset \rejset$. Then
\begin{align}
    0 < (m - |\rejset|)(|\rejset|- |R|).
\end{align}
This implies that
\begin{align}
    m|\rejset| + |R||\rejset| - |\rejset|^2 = (|R| + m - |\rejset|)|\rejset| > m|R|.
\end{align}
Now, we take $S = R \cup ([m] \setminus \rejset)$.
\begin{align}
    \frac{m|\rejset \cap S|}{|S||\rejset|} = \frac{m |R|}{(|R|+m-|\rejset|) |\rejset|}  < 1 = \frac{|R \cap S|}{|R|}.
\end{align}

Let $R \not\subseteq \rejset$, that is, $R$ contains some  $i \notin \rejset$, and
take $S=[m] \setminus \rejset$. Then,
\begin{eqnarray*}
    \frac{m|\rejset \cap S|}{|S||\rejset|} = 0<   \frac{|\{i\}|}{|R|} \leq  \frac{|R \cap ( [m] \setminus \rejset )|}{|R|}.\\
\end{eqnarray*}
Thus, we have shown the desired result for $\rejset_\alpha^\BH$ and $\rejcol_\alpha^\cBH$.

\revise{Now let $\rejset=\rejset_\alpha^\MABH$ and $\rejcol=\rejcol_\alpha^{\cMABH}$. The case $\rejset=\emptyset$ is immediate, so suppose $\rejset\neq\emptyset$. The original set $\rejset$ belongs to $\rejcol$: for $S=[m]$ this follows from $\evalue_{[m]}=1/\alpha$, while for every proper $S$ it follows from $(m-1)/|S|\geq1$. Consider any nonempty $R\neq\rejset$. If $|R|>|\rejset|$, choose $j\in\rejset$ and set $S=[m]\setminus\{j\}$. Then
\[
\alpha\evalue_S=\frac{|\rejset|-1}{|\rejset|}
<1-\frac{1}{|R|}
\leq\frac{|R\cap S|}{|R|}.
\]
If instead $|R|\leq|\rejset|$, choose $i\in\rejset\setminus R$ and again set $S=[m]\setminus\{i\}$. In this case,
\[
\alpha\evalue_S=\frac{|\rejset|-1}{|\rejset|}<1=\frac{|R\cap S|}{|R|}.
\]
Thus, for every nonempty $R\neq\rejset$ there exists some $S$ s.t. $\alpha \evalue_S < \FDP_S(R)$, including when $\rejset=[m]$, and hence $\rejcol_\alpha^{\cMABH}=\{\emptyset,\rejset_\alpha^\MABH\}$.}

\revise{Finally, let $\rejset=\rejset_\alpha^{\adaBH}$. The indicator in \eqref{eq: adaBH local e-value} equals one precisely for the hypotheses in $\rejset$. 
\[
\alpha\evalue_S=\frac{|\rejset\cap S|}{|\rejset|\vee1}.
\]
Otherwise, consider a nonempty $R\neq\rejset$. If $R\not\subseteq\rejset$, choose $i\in R\setminus\rejset$ and $S=\{i\}$, yielding $\alpha\evalue_S=0<|R\cap S|/|R|$. If $R\subsetneq\rejset$, take $S=R$, yielding $\alpha\evalue_S=|R|/|\rejset|<1=|R\cap S|/|R|$. Therefore $\rejcol_\alpha^{\cadaBH}=\{\emptyset,\rejset_\alpha^\adaBH\}$, including when $\rejset_\alpha^\adaBH=[m]$.}
\end{proof}

\revise{The particular e-collections considered in this section therefore do not substantially improve \BH\ or its adaptive variants. \citet{goeman_uniform_improvement_2026}, however, develops a different e-Closure construction that does uniformly improve over \BH\ in ways that are different from minimally adaptive \BH.}
 }

\section{Deferred proofs}\label{sec: deferred proofs}

\subsection{Proof of \Cref{lemma:simult_posthoc}}
\revise{Order $2^{[m]}$ lexicographically and let $\mathrm{R}^*(x)$ be the lexicographically first maximizer of $\textnormal{f}_{N_{\mathrm{P}}}(R)$ over $R\in\mathcal{R}(x)$. Because the action space is finite and $\mathcal{R}$ and the finitely many losses are measurable, $\mathrm{R}^*:S\to2^{[m]}$ is an $\mathcal A$-measurable selector. It attains the pointwise maximum, proving the ``$\leq$'' direction of \eqref{eq:simult_select}; the reverse direction holds because every admissible selector takes values in $\mathcal R(\mathbf X)$ almost surely.}

\subsection{Proof of \Cref{ex: eBH+}}\label{sec: eBH+ example proof}
It is tedious but straightforward to show that $(2m-2i+1)/m < m/i$ for $i=1,\ldots, m$, so that the example is not void. It follows immediately from (\ref{eq: eBH}) that eBH rejects nothing. Choose $R = [m]$ and any non-empty $S \subseteq [m]$ with $s=|S|$. Then
\[
\sum_{i \in S} \frac{\revise{\mathbf{e}_{(i)}}}s \geq \sum_{i = m-s+1}^m \frac{\revise{\mathbf{e}_{(i)}}}s \geq \sum_{i = m-s+1}^m \frac{2m-2i+1}{ms\alpha} =
\frac{2m-2 \frac{m + (m-s+1)}2 +1}{m\alpha} = \frac{s}{m\alpha} = \frac{|R \cap S|}{|R|\alpha}.
\]
\subsection{Proof of \Cref{thm: compound ebh equal eclosure}}\label{sec: compound ebh equal closure thm proof}
    To show $\widetilde{\rejcol}^\SC_\alpha \subseteq \widetilde{\rejcol}^\CeBH_\alpha$, consider any $R \in \widetilde{\rejcol}^\SC_\alpha$. \revise{The closure inequality is immediate when $S\cap R=\emptyset$; otherwise,}
\begin{align}
    \evalue_S = \frac{1}{m}\sum\limits_{i \in S} \widetilde{\evalue}_i
        \geq
        \frac{|S\cap R|}{m} \min_{i \in S \cap R} \widetilde{\evalue}_i
        \geq
        \frac{|S\cap R|}{m}\min_{i \in R} \widetilde{\evalue}_i
        \geq
        \frac{|S\cap R|}{m} \cdot \frac{m}{\alpha|R|}
        \geq
        \frac{\FDP_S(R)}{\alpha}.
    \label{eq:c-lb-ineq}
\end{align}

Further, for any $R \in \widetilde{\rejcol}^\CeBH_\alpha$ and any $i \in R$, we have the following:
\begin{align}
    \frac{\widetilde{\evalue}_i}{m} = \evalue_{\{i\}} \geq \frac{\FDP_{\{i\}}(R)}{\alpha} = \frac{1}{\alpha |R|}.
\end{align} Hence, we get that $\widetilde{\evalue_i} \geq m / (\alpha|R|)$ for all $i \in R$, which implies $R \in \widetilde{\rejcol}^\SC_\alpha$ and $\widetilde{\rejcol}^\CeBH_\alpha \subseteq \widetilde{\rejcol}^\SC_\alpha$. Thus, we have shown our desired result.

\subsection{Proof of \Cref{prop: compound e-values from e-collections}}\label{sec: compound e-values from e-collections proof}
The validity of the compound e-values follows from the definition of an e-collection.
We note that the above compound e-values in \eqref{eq: universal compound evalue} satisfy \eqref{eq: intersection bound}. This is because, for any $S \in 2^{[m]}$ where $|S \cap \rejset| > 0$, we see that
\begin{align}
    \revise{\sum_{i \in S} \widetilde{\evalue}_i = \frac{m|S \cap \rejset|}{\alpha|\rejset|} \cdot \mathbf{1}\{\alpha \evalue_S \geq \FDP_S(\rejset)\} \leq \frac{m|S \cap \rejset|}{|\rejset|}\cdot \frac{\evalue_S}{\FDP_S(\rejset)} = m \evalue_S.}
\end{align} The inequalities follow from the definition of the e-Closure Principle and by the fact that $\mathbf{1}\{x \geq 1\} \leq x$ for any $x \geq 0$. The case where $|S \cap \rejset| = 0$ is also satisfied trivially.

\subsection{Proof of \Cref{ex: closed by example}}\label{sec: closed by example proof}
It is immediate from the definition that the BY rejection set is empty.
\revise{Now, choose any non-empty $S \subseteq [m]$. We have \( \mathbf{1}\{ h_{|S|} \mathbf{p}_{(i)} \leq \alpha \} = 1 \) for all \(i\in[m]\) when $S \neq [m]$. For $i\leq m-1$ this follows from the displayed bounds in the example; for $i=m$, it follows separately from $\pvalue_{(m)}\leq\alpha/h_{m-1}\leq\alpha/h_{|S|}$. Thus, for \( S \neq[m] \),}
\[
    \revise{\left\lceil \frac{|S| h_{|S|} \mathbf{p}_{(i)}}{\alpha} \right\rceil\leq m \quad \text{for all } i \in [m].}
\]
Hence, for all \( S \neq [m] \), we obtain
\[
    \revise{\alpha \mathbf{e}_S \geq \sum_{i \in S } \left\lceil \frac{|S| h_{|S|} \mathbf{p}_{(i)}}{\alpha} \right\rceil^{-1} \geq \frac{|S|}{m} = \frac{|S \cap [m]|}{|[m]|}.}
\]
In case \( S = [m] \), we have
\[
    \revise{\alpha \mathbf{e}_{[m]} = \sum_{i = 1}^{m-1} \left\lceil \frac{m h_m \mathbf{p}_{(i)}}{\alpha} \right\rceil^{-1} \geq \sum_{i = 2}^{m} \frac{1}{i} \geq 1.}
\]
Thus, this establishes that \(\rejset_\alpha^\cBY = [m]\).

\subsection{Proof of \Cref{lemma: su-calibrator}}\label{sec:su-calibrator-proof}
The Lambert W function has the property that $\exp(w(x)) = x/w(x)$. Therefore,
\[
\log\Big(\frac{e\ell_\alpha}{\alpha}\Big)
= \log\left(-\frac{e}{\alpha}w\left(-\frac{\alpha}{e}\right)\right)
= \log\left(\frac{w\left(-\alpha / e\right)}{-\alpha / e}\right)
=-w\left(-\frac{\alpha}{e}\right) = \ell_\alpha.
\]
\revise{Equivalently, $\ell_\alpha=1+\log(\ell_\alpha/\alpha)$.}
Let $\mathbf{p}$ be a p-value and $\mathbf{e} = e(\mathbf{p})$.
\revise{We can see that the integral of $e$ over $[0, 1]$ is}
\[
    \revise{\int_0^1 e(p)\,dp
    =
    \int_0^{\alpha / \ell_\alpha} \frac1\alpha\, dp + \int_{\alpha /  \ell_\alpha}^1 \frac{1}{p\ell_\alpha}\, dp
    =
    \frac{1+\log(\ell_\alpha/\alpha)}{\ell_\alpha}
    = 1.}
\]
\revise{Because $e$ is decreasing, the usual superuniformity argument gives $\expect[e(\mathbf{p})] \leq \int_0^1 e(p)\,dp = 1$ for every p-value $\mathbf{p}$. Thus, $e$ is a calibrator and $\mathbf{e}$ is an e-value.}

\subsection{Proof of \Cref{thm: su-improvement}}
\label{sec:su-improvement-proof}
Let $R$ be a self-consistent discovery set at level $\alpha / \ell_\alpha$, for which $\rejset^\Su_\alpha$ is the largest such set. \revise{If $S\cap R=\emptyset$, the closure inequality is immediate. Otherwise,}

\revise{we have $\mathbf{p}_i \leq |R|\alpha/(m\ell_\alpha)$ for all $i \in R$. Choose any $S \in 2^{[m]}$ such that $S \cap R \neq \emptyset$. Then}
\[
\mathbf{p}_S \leq \mathbf{p}_{|R \cap S|:S}\frac{|S|}{|R \cap S|}
\leq \frac{|R|\alpha}{m\ell_\alpha} \frac{|S|}{|R \cap S|}
\leq \frac{|R|\alpha}{|R \cap S|\ell_\alpha}.
\]
It follows that $\alpha\mathbf{e}_S \geq |R\cap S|/|R|$ if $S \cap R \neq \emptyset$ and $\mathbf{p}_S \geq \alpha / \ell_\alpha$. The same is trivially true if either $S \cap R = \emptyset$ or $\mathbf{p}_S < \alpha / \ell_\alpha$. Therefore, we have $R \in \rejcol^\cSu_\alpha$.

\revise{To show that the improvement region is nonempty for every $m>1$, write $q=\alpha/\ell_\alpha\in(0,1)$ and take ordered p-values}
\[
\revise{\pvalue_{(1)}=\frac qm,\qquad
\frac{2q}{m}<\pvalue_{(2)}\leq\min\left\{\frac{2q}{m-1},1\right\},
\qquad \pvalue_{(3)}=\cdots=\pvalue_{(m)}=1.
}
\]
\revise{The Su procedure rejects only the first hypothesis. For $R=\{1,2\}$, the closure inequality is trivial when $S\cap R=\emptyset$; if $1\in S$, the Simes p-value is at most $q|S|/m$, giving $\evalue_S\geq\alpha^{-1}$; and if $S\cap R=\{2\}$, the bound on $\pvalue_{(2)}$ gives $\evalue_S\geq(2\alpha)^{-1}$. Hence $\{1,2\}\in\Rcal_\alpha^\FDR(\mathbf E)$.}

\subsection{Proof of \Cref{thm: simultaneity improvement}}
\label{sec: simultaneity improvement proof}

    Let $\mathbf{e}_S$ be given as in \eqref{eq: knockoff local e-value} and $R\neq \emptyset$ be any set given by \eqref{eq: knockoff_char}. Then,
    \begin{align}
        \mathbf{e}_S\geq \frac{|S\cap R|}{1 + \sum_{j \in S} \mathbf{1}\{\mathbf{w}_j \leq -\mathbf{c}_\alpha^\Kn\}} \geq \frac{|S\cap R|}{|R|\alpha},
    \end{align}
    where the first inequality follows from the first property in \eqref{eq: knockoff_char} and the second inequality follows from the second property.
     It remains to show that for any $R\subseteq [m]$ that does not satisfy the properties specified in \eqref{eq: knockoff_char} it holds that $R\notin \rejcol_\alpha^{\cKn}$. Let such an $R$ be given. We have $\mathbf{w}_i< \mathbf{c}_\alpha^\Kn$ for some $i\in R$ or $\frac{|R|}{1 + \sum_{j \in [m]} \mathbf{1}\{\mathbf{w}_j \leq -\mathbf{c}_\alpha^\Kn\}}< 1/\alpha$. In the former case, we have \revise{$\mathbf{e}_{\{i\}}=0<1/(\alpha|R|)$.} In the latter case, consider $S=\{i\in [m]:\mathbf{w}_i \leq -\mathbf{c}_\alpha^\Kn \text{ or } i\in R\}$. Then, we have that $\mathbf{e}_{S}\leq \frac{|R|}{1 + \sum_{j \in [m]} \mathbf{1}\{\mathbf{w}_j \leq -\mathbf{c}_\alpha^\Kn\}}< 1/\alpha$.

\subsection{Proof of \Cref{thm: FWER}}\label{sec: FWER proof}
Let $\boldsymbol{\mathcal{S}} = \{S \in \boldsymbol{\mathcal{R}}\colon |S|=1\}$. Then, for every $\mathrm{P} \in M$,
\[
\mathrm{P}(\mathbf{R} \cap N_\mathrm{P} \neq \emptyset) =
\expect_\mathrm{P}\big( \max_{i \in \mathbf{R}} 1\{i \in N_\mathrm{P}\}\big) =
\expect_\mathrm{P}\bigg( \max_{R \in \boldsymbol{\mathcal{S}}} \frac{|R \cap N_\mathrm{P}|}{|R| \vee 1} \bigg)  \leq
\expect_\mathrm{P}\bigg( \max_{R \in \boldsymbol{\mathcal{R}}} \frac{|R \cap N_\mathrm{P}|}{|R| \vee 1} \bigg)  \leq \alpha.
\]
This gets us our desired result.

\subsection{Proof of \Cref{thm: multiloss}}
\label{sec: multiloss proof}
The validity of $(\Rcal^{\ER_\F}_\alpha(\mathbf{E}))_{\F \in \mathcal{F}}$ is due to
\begin{align}
\expect_{\pdist}\left[\sup_{\F\in \mathcal{F}} \max_{R \in \Rcal_\alpha^{\ER_\F}(\mathbf{E})}\  \F_{N_\pdist}(R)\right]
\leq  \expect_{\pdist}\left[\alpha \evalue_{N_\pdist}\right] \leq \alpha.
\end{align}

Given a collection of discovery set collections $(\rejcol^{\ER_\F})_{\F \in \mathcal{F}}$, we can define
\begin{align}
    \evalue_S  = \frac{\sup_{\F \in \mathcal{F}}\max_{R \in \rejcol^{\ER_\F}} \F_S(R)}{\alpha}.
\end{align} As a result, $(\evalue_S)_{S \in 2^{[m]}}$ is an e-collection by virtue of $(\rejcol^{\ER_\F})_{\F \in \mathcal{F}}$ satisfying \eqref{eq:sup-fcond}.

\subsection{Proof of \Cref{thm:simult_fdp}}
\label{sec: simult fdp proof}
The simultaneous true discovery guarantee of $\mathrm{d}_\mathbf{E}$ follows by the simultaneous control of the e-Closure Principle over multiple error functions and rejection sets (\Cref{thm: multiloss}), since
$$\expect_{\pdist}\left[\sup_{d\in [m]_0} \max_{R \in \rejcol_{\F^d}}\  \F_{N_\pdist}^d(R)\right] = \pdist\left(\exists d\in [m]_0, R \in \rejcol_{\F^d}: d>|R\setminus N_{\pdist}|\right).$$ Furthermore, if $\mathbf{e}_S = \boldsymbol{\phi}_S/\alpha$ for $S\in 2^{[m]}$ and some family of local tests $\boldsymbol{\Phi}$, then, for an arbitrary $R\in 2^{[m]}$,
\begin{align}
    \mathrm{d}_\mathbf{E}(R)&=\max\left\{d\in [m]_0: \boldsymbol{\phi}_S \geq 1\{|R\setminus S|<d\} \quad \forall S\in 2^{[m]}\right\} \\
    &= \max\left\{d\in [m]_0: |R\setminus S|\geq d \text{ for all } S\in 2^{[m]} \text{ with } \boldsymbol{\phi}_S=0\right\}\\
    &= \min\{|R\setminus S|:  S \in 2^{[m]}, \boldsymbol{\phi}_S=0\} \\
    &= \mathrm{d}_{\boldsymbol{\Phi}}(R).
\end{align} Thus, we have our desired result.

\subsection{Proof of \Cref{thm: e for MABH}}
\label{sec: e for MABH proof}

If $\rejset^\BH_\alpha = \emptyset$, then also $\rejset^\MABH_\alpha = \emptyset \in \rejcol_\alpha^\FDR(\ecol)$, since the empty set is always part of any $\rejcol_\alpha^\FDR(\ecol)$.

Suppose $\rejset^\BH_\alpha \neq \emptyset$. Then $\evalue_{[m]} = 1/\alpha \geq \FDP_{[m]}(\rejset^\MABH_\alpha)/\alpha$, and for any non-empty $S \neq [m]$, we have
\[
\evalue_S = \frac{(m-1) |\rejset_\alpha^{\MABH} \cap S|}{\alpha |S|(|\rejset_\alpha^{\MABH}| \vee 1)} \geq \frac{\FDP_{S}(\rejset^\MABH_\alpha)}{\alpha}.
\]
Thus, we have shown that $\rejset^\MABH_\alpha \in \rejcol_\alpha^\FDR(\ecol)$.

$\evalue_{[m]}$ is a valid e-value for $H_{[m]}$ since $\rejset_\alpha^\BH$ is a FDR controlling procedure at level $\alpha$. The other $\evalue_S$ are valid by the superuniformity lemma of \citet{ramdas_unified_treatment_2019}.

Thus, we have shown our desired result.

\end{document}